\documentclass[a4paper,10pt,onecolumn]{article}
\usepackage[numbers]{natbib}

\usepackage[utf8x]{inputenc}
\usepackage[margin=2.5cm]{geometry}
\usepackage{amsthm}
\usepackage{amsmath}
\usepackage{amssymb}
\usepackage{bm}
\usepackage{amsfonts}

\usepackage{txfonts}

\usepackage{graphicx}

\newcommand{\E}{{\mathbf{E}}}
\newcommand{\R}{{\mathbb{R}}}
\newcommand{\Rd}{{\mathbb{R}^d}}

\newcommand{\I}{\mathbf{1}}
\newcommand{\bessel}{\mathcal{J}}
\newcommand{\Sd}{S^{d-1}}
\newcommand{\cov}{{\mathrm{\bf Cov}}}
\newcommand{\var}{{\mathrm{\bf Var}}}
\newcommand{\rel}{{\mathrm{Rel}}}
\newcommand{\x}{{\mathbf{x}}}
\newcommand{\y}{{\mathbf{y}}}

\newcommand{\two}{^{(2)}}
\newcommand{\rtwo}{\rho^{(2)}}

\newcommand{\kap}[1]{\rtwo(#1)-\lambda^2}
\newcommand{\Fou}{\mathcal{F}}

\renewcommand{\k}{\freq}
\newcommand{\freq}{{k}}
\newcommand{\angfreq}{{\omega}}

\newcommand{\ntapers}{{M}}
\newcommand{\itapers}{{{m}}}
\newcommand{\itapersalt}{{{m'}}}
\newcommand{\subtapers}{{\tilde m}}
\newcommand{\nsubtapers}{{\tilde M}}

\newcommand{\bside}{l}
\newcommand{\besorder}{{\nu}}

\usepackage{xcolor}

\newtheorem{Definition}{Definition}[section]
\newtheorem{Theorem}{Theorem}[section]
\newtheorem{Lemma}{Lemma}[section]
\newtheorem{Corollary}{Corollary}[section]
\newtheorem{Proposition}{Proposition}[section]

\begin{document}
\title{What is the Fourier Transform of a Spatial Point Process?}

\author{Tuomas Rajala,
        Sofia~C.~Olhede, Jake P. Grainger
        and~David J. Murrell,% <-this % stops a space
\thanks{T. Rajala is with Natural Resources Institute Finland, 00790 Helsinki, Finland.}% <-this % stops a space
\thanks{S. C. Olhede and J. P. Grainger are with the Institute of Mathematics, EPFL, Lausanne, Switzerland.}% 
\thanks{D. J. Murrell is with the Research Department of Genetics, Evolution and Environment,
Centre for Biodiversity and Environment Research, University College London, UK.}% <-this % stops a space
%\thanks{Manuscript received January X, X; revised X X, X.}
}

\maketitle

\begin{abstract}
This paper determines how to define a discretely implemented Fourier transform when analysing an observed spatial point process. To develop this transform we answer four questions; first what is the natural definition of a Fourier transform, and what are its spectral moments, second we calculate fourth order moments of the Fourier transform using Campbell's theorem. Third we determine how to implement tapering, an important component for spectral analysis of other stochastic processes. Fourth we answer the question of how to produce an isotropic representation of the Fourier transform of the process. This determines the basic spectral properties of an observed spatial point process.
\end{abstract}


\section{Introduction}\label{sec:intro}
Spatial point processes are an important form of observational data structure, for example in forest ecology~\cite{waagepetersen2016}, communications networks~\cite{novlan2013analytical,li2015statistical},   epidemiology~\cite{Gatrell1996}, social science~\cite{Tita2010}, pharmacology~\cite{habel2017athree} and medicine~\cite{anderson2018hierarchical} amongst many other application fields. Understanding the properties of a point process can be approached from many different perspectives~\cite{Diggle2013}, and the aim of this paper is to determine how to extract the frequency (or scale/wavenumber) behaviour of an observed point process, as well as connect that to its spectral representation.

The most common assumption used when analysing spatial processes is that of spatial homogeneity (or stationarity). This means that if you shift all of the observations by a fixed spatial shift, the distribution of those observations does not change from the distribution of the original sample; and if the observational window is changed, then understanding the full set of observations remains tractable. The consequence
of this probabilistic invariance in distribution is the spectral representation of a stochastic process. The existence of the spectral representation of a stochastic process means the Fourier transform fully characterises the second order properties of that stochastic process. The Fourier transform also characterises the second order properties of a point process, see e.g.~\cite{Daley2003}. Yet unlike random fields and time series, spectral analysis of point processes is still in its infancy, see also~\cite{Bartlett1963,Bartlett1964,mugglestone1996practical}, and critically, the digital processing of a point process remains fully outstanding. 
{Recent interest in Fourier features in machine learning based approaches for point patterns such as ~\cite{ilhan2020modeling,john2018large} show the potential of using Fourier based information as features for estimation and detection. The work in this manuscript both establishes what Fourier features to calculate for homogeneous processes from a sampling perspective, and their large but finite sampling area properties, just like~\cite{guillaumin2022debiased} determined the large but finite properties of Fourier representations for random fields. We note that machine learning techniques~\cite{ton2018spatial,lazaro2010sparse} utilised Fourier features in learning algorithms for parametric models before precise and general understanding was established.}

Thus despite inspirational work by the aforementioned authors, several key aspects of discrete spectral analysis lie unresolved when applied to observed point processes. In particular, 1)
given the Fast Fourier Transform is unavailable (unlike the case for regularly sampled time series or random fields), how do we efficiently implement the calculation of whichever discrete Fourier transform we chose to define? 2) How does that discrete transform relate to the underlying spectral measure of the process? 3) Can the discrete transform be improved by linear operations such as tapering as is the case for other stochastic processes? 4) If yes, how do we then select a taper? 5) How do we define a radial or isotropic transformation to simplify the representation of the spatial point process? There are many, many more questions to answer before bringing the spectral analysis of spatial point processes to the sophistication of the analysis of random fields, but these currently unresolved questions are positioned as the first
%like a massive%
hurdles to overcome in the path of whomever wishes to develop more sophisticated spectral analysis techniques. Once the fundamental properties of the spectral representation of a stationary point process have been developed and understood, results could be extended to inhomogeneous point processes, but that is not the aim of this work.

{%\color{red} 
To put this in the context of information theory; understanding how to compute the Fourier transform from a spatially compact signal had already sparked a lengthy debate from the first introduction of tapering~\cite{thomson1982spectrum}, and selection of optimal tapers using localization operators~\cite{slepian1983some,slepian1964prolate,connes2022UV}. While some aspects of tapering in the context of point patterns is reflected by topics in mature (multi)dimensional tapering and spectral estimation~\cite{abreu2017mse,karnik2019fast,karnik2022thomson,guillaumin2022debiased}, understanding how to adapt these ideas to spatial point processes remains fully outstanding, and must be answered using our understanding of irregularly sampled processes~\cite{chou1992spectral}. Understanding and characterising point patterns is naturally a problem of general interest, e.g.~\cite{clark2020local,clark2022cramer}.

Why then do we want to understand the spectral information of a point pattern? The spectral information of any stochastic process is the same as the spatial information as there is a bijection between the two, but the former is directly showing the scale of variation of the process. To illustrate, let us simulate a complex point pattern and compare its spectral content with its usual spatial representation. Log-Gaussian Cox processes~\cite{moller1998log}, popularly used in applications, are marginally stationary processes whose patterns are more variable than models that do not depend on latent random variables. The pattern (a single realization), and its estimated spectrum (the subject of this paper) are shown in Figure~\ref{fig:fig-12}. 
We have chosen a process which is isotropic, a choice we have made so that its usual spatial representation and spectrum is isotropic also, and thus easier to interpret.
In the left hand subplot we see the pattern itself, in the middle subplot we show a common spatial summary, the pair correlation function (pcf) comparing both the estimate and truth, and in the right hand subplot we show the spectrum on a decibel scale ($10\log_{10}$), both estimated and true. While the raw pattern (left) looks mainly unremarkable, naturally its periodicity is present in the pcf (middle), but is more obvious in the spectrum (right). We can even estimate the periodicity as having frequency 1. This is immediately recoverable directly from its estimated spectrum.
% In subplot (a) we see the pattern itself, in subplot (b) its estimated spectrum on a decibel scale, in (c) we have estimated a common spatial summary of its product density, and in subplot (d) its spectrum on a decibel scale. While the raw pattern in (a) looks mainly unremarkable, naturally its periodicity is present in (c), but is strongly standing out in (b) and (d). We can even estimate the periodicity as having frequency X. This is immediately recoverable directly from its estimated spectrum.
}

\begin{figure}
    \centering
    \includegraphics[width=\linewidth]{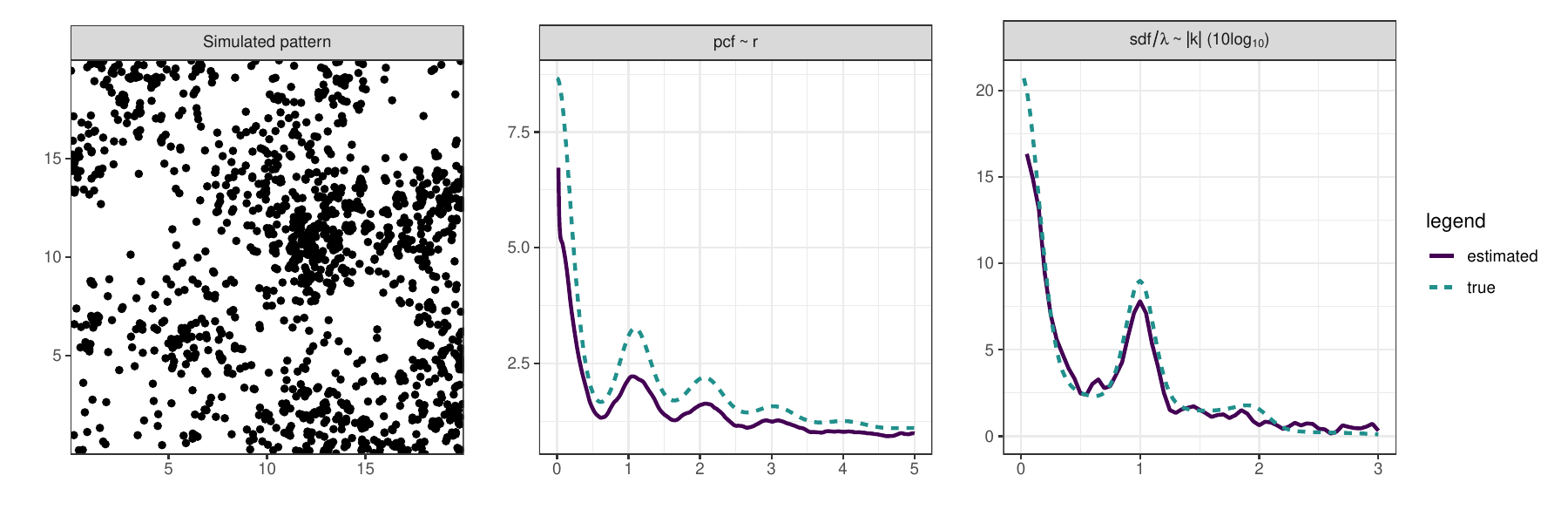}
    \caption{A realisation of a log-Gaussian Cox process (left), the estimated against true product density (middle), and estimated against true spectral density function on a decibel scale (right).}
    \label{fig:fig-12}
\end{figure}

% \begin{figure}
% 	\centering
% 	\includegraphics[width=0.9\linewidth]{sdf-example-figure-Maternthinnedthomas.pdf}
% 	\caption{}
% 	\label{fig:fig-12}
% \end{figure}

{%\color{red} 
To get to the point where we can estimate 
the spectrum of a point pattern, we need to 
answer the five questions posed earlier.
We will therefore start by determining how to compute Fourier transforms of point processes, discussing the question of how to form direct spectral estimators~\cite{Percival1993} in this setting.  This will be answered in terms of} the mean and (co)-variance of the different discrete Fourier transforms that we define. The first surprising result is that the natural direct spectral estimator is biased unless it is mean corrected (this bias was noted by~\cite{diggle1987nonparametric}, but only corrected in an ad-hoc fashion). %\david{perhaps reword last sentence as it might appear as though ref [17] has remarked on the bias caused by the mean, whereas I presume they discuss a bias without pinning it down?}
We show how to perfectly eliminate this bias. 

Second, to understand the variance and covariance of a direct spectral estimator, we need to calculate a fourth order moment of our choice for the discrete spectral transform. This is complicated since Campbell's formula is required to derive its form for a point process. This approach can be contrasted with using Isserliss' theorem to determine higher order moments for a Gaussian Process~\cite{isserlis1918ona}. Calculating covariances between direct spectral estimators gives us a way to determine what grid of wavenumbers the spectral estimator is uncorrelated at, {and thus where to evaluate it}.

Third, tapering~\cite{thomson1982spectrum} is required to define the direct spectral estimator and avoid leakage; but how to taper a point pattern is an open question. Most tapers are defined for regularly spaced stochastic processes; but there are some continuous tapers. We choose to use the continuous tapers of Riedel and Sidorenko~\cite{Riedel1995} to construct the direct spectral estimator, and use multitapers to reduce the variance of such estimators. {%\color{red} 
We will determine how to implement both. }

Fourth, many spatial models are radially symmetric, or isotropic. This prompts us to describe how to construct an isotropic representation of the spectral content of a point processes. There are two possible approaches, namely using the Bessel function~\cite{diggle1987nonparametric} to do the transformation; or isotropizing a general spectral representation~\cite{Durran2017}. 
We describe how to define the appropriate tapering in this instance, inspired by other isotropic harmonic decompositions~\cite{jacovitti2000multiresolution}.

Calculating the Fourier transform of a point process is useful beyond the point pattern itself. If we sample a stochastic process with a point process then the spectrum of a point pattern is equally important to that of the stochastic process to determining the spectrum of the observations, e.g.~\cite{LiiMasry1994}. For this reason; the study of the frequency information of an observed process is of interest in its own right, and beyond, and we will discuss the implications that arise for point processes. 
%\david{Check my suggestions for changes to this last sentence, to make sure the meaning is retained.}.

In Section~\ref{sec:not} we define the basic concepts required by the second order representation of a point pattern. In Section~\ref{sec:dirspectral} we discuss tapering of the DFT and the direct moments of the DFT. 
{%\color{red} 
In Section~\ref{sec:expbilin}  we determine how to form a spectral density estimator from our understanding of how to form the DFT and its mean and variance. The next step is to study the  covariance structure of such spectral density estimators, see
Section~\ref{sec:varbilinear}. We then use that understanding to define how to do linear smoothing in Section~\ref{sec:linear}.
For multidimensional spectral representations, the full anisotropic structure can be hard to interpret; we therefore propose 1d isotropic or isotropised summaries to characterise such processes, in Section
\ref{sec:isotropy}. We then present a simulation study in Section~\ref{sec:simulations}, and demonstrate the potential of the proposed methodology in an example from forest ecology in Section~\ref{sec:BCI}. We conclude with a discussion in Section~\ref{sec:conclusions}.}

%%%%%%%%%%%%%%%%%%%%%%% Section 

\section{Notation and Definitions}\label{sec:not}
{%\color{red} 
In this Section we shall give the basic notation necessary for the spectral analysis of a spatially homogeneous point process.}
We assume that we observe a spatially homogeneous point process $X$ with intensity $\lambda=\rho^{(1)}>0$ and with $\rho^{(2)}(z)$ defined as the (second order) product density of $X$. Assuming $X$ is spatially homogeneous means  that $\rho\two(z)$ only has one argument ($z\in {\mathbb{R}}^d$, namely the spatial shift) rather than depending on two local spatial variables (say $x$ and $y$ rather than just $z=x-y$). We assume that $X$ is a simple point process on ${\mathbb{R}}^d$, meaning no duplicate points are allowed. We take $d=2,3,\dots$, and so exclude the case of $d=1$, which is relatively well-studied, see for example the references in \cite{daley1971weakly,jowett1972prediction,vere1974elementary}. For a (Borel) set $A$ in ${\mathbb{R}}^d$, we define $|A|$ for the volume of $A$ and $N(A) = |X\cap A|$ for the number of points of $X\subset {\mathbb{R}}^d$ in $A$. For set $A$ and $z\in \mathbb{R}^d$, denote the $z$ shifted set by $A_z=\{x+z: x\in A\}$. We denote by $B$ any subset of $ {\mathbb{R}}^d$
where points are observed (or equivalently our observation domain). For any complex number $z$ we use the superscript $z^\ast$ to denote the conjugate of $z$.

In complete analogue with a random field we shall define the spectral density of $X$ as the Fourier transform of the complete covariance function of $X$. This is the standard approach, and was first proposed by~\cite{Bartlett1964}.
The complete covariance function of a stationary point process is therefore
\begin{equation}
\label{covy}
\gamma(z) \equiv \lambda\delta(z) + \rho\two(z) - \lambda^2,\quad z\in \Rd,
\end{equation}
where $\delta(\cdot)$ is the Dirac delta function.   
The spectral density function (sdf)~\cite{Bartlett1964} of the point process $X$ is then defined as the Fourier transform of the complete covariance function:
\begin{equation}
\label{def:spectrum}
f(\k) \equiv \Fou[\gamma](\k)\ = \int_{{\mathbb{R}}^d}e^{-2\pi i\k\cdot z}\gamma(z)\,dz\  =\  \lambda + \int_{\Rd}e^{-2\pi i\k\cdot z}[\rho^{(2)}(z)-\lambda^2]dz, \qquad \k
\in {\mathbb{R}}^d.
\end{equation}
This representation is in direct analogy with the corresponding spectral decomposition of a random field or a time series. 
The symbol $\Fou$ denotes the Fourier transform, with $i$ as the imaginary unit. We refer to the argument $\k$ as the ``wavenumber'' rather than frequency to acknowledge its multi-dimensionality. This is a natural choice of a Fourier transform in analogy to the analysis of random fields~\cite{stein2012interpolation}. To avoid additional constants in the inverse Fourier transforms,  
we parameterise the Fourier transform with wavenumber instead of the customary angular frequency $\angfreq = 2\pi \k$ used by~\cite{Bartlett1964} and some others.

For a time series or random field there are a number of spectral results, ranging from the spectral representation theorem (I)~\cite[p130]{Percival1993} and to Bochner's (Herglotz) theorem (II)~\cite{Percival1993}. 
Both these results  are not exactly available in the point process setting, but for a time series $X_t$ we can note them:
\begin{align}
X_t&\overset{(I)}{=}\mu+\int_{-\frac{1}{2}}^{\frac{1}{2}} e^{2 \pi i\freq t}dZ(\freq), \quad \var\{dZ(\freq)\}\overset{(II)}{=}f(\freq)d\freq, \label{Xt}
\\
\gamma(\tau)&\overset{(III)}{=}\int_{-\frac{1}{2}}^{\frac{1}{2}}
f(\freq)e^{2\pi i \freq \tau}\,d\freq,\quad f(\freq)\overset{(IV)}{=}\sum_\tau \gamma(\tau) e^{-2\pi i\freq \tau}. \label{gamt}
\end{align}
Equation~\eqref{Xt} decomposes $X_t$ into random contributions associated with each frequency $\freq$. %\tuomas{[Notation :above $\freq$ %which looks a lot like $\omega$ is the wavenumber, and $\xi=2\pi w$ %is angular frequency]}. 
The covariance in Eqn~\eqref{gamt} is also decomposed into weighted frequency contribution. Both of these observations are important for the interpretation of the Fourier representation of $X_t$. ``Important'' contributions correspond to $f(\freq)$ being considerably larger relative to other  $f(\freq')$ as in that scenario $\var\{dZ(\freq)\}$ is bigger than $\var\{dZ(\freq')\}$ and therefore $|dZ(\freq)|$ likely to be larger than $|dZ(\freq')|$.
Spectral representations of point processes are also discussed in~\cite[Ch.~8]{Daley2003}. We can yet again decompose the (complete) covariance 
in terms of a spectral representation as in (II), and it takes the form in terms of a $d$-dimensional Fourier transform of
\begin{equation}
f(\freq)=\int e^{-2\pi i \freq\cdot z}\gamma(z)\;dz=\lambda+\int e^{-2\pi i \freq\cdot z} \left\{\rho^{(2)}(z)-\lambda^2\right\}\;dz.
\label{weightav}
\end{equation}
We refer to $f(\freq)$ as the \underline{spectrum} of $X$. We note that the Fourier transform can be inverted to yield the equality of (an analogy of III):
{{\begin{equation}
    \rho^{(2)}(z)=\lambda^2+\int_{{\mathbb{R}}^d} f(\freq) e^{2\pi i \freq\cdot z}\,d\freq.
    \label{eqn:rhobyscale}
\end{equation}}}

We assume the spectrum is the primary object of interest in our study of point patterns, as the covariance $\gamma(z)$ can be fully determined from it, and the covariance characterises the point process. We see directly from~\eqref{weightav}
that $\rho^{(2)}(z)-\lambda^2$ and {\em not} the complete covariance $\gamma(z)$ is playing the role of a time series covariance.
From a nonparametric perspective the spectrum characterises what wavenumbers are more notable (distinct) in the process. As~\eqref{weightav} specifies a constant level $\lambda$ to all wavenumbers, the notable (distinct) wavenumbers are determined from the Fourier transform of $\rho^{(2)}(z)-\lambda^2$.

%\david{I'm struggling to find a better way of writign the next %sentence, but it needs a little polish}This is as when transforming $\gamma(z)$ rather than $\rho^{(2)}(z)-\lambda^2$ we otherwise get a constant boost of $\lambda$ that spreads across all wavenumbers. However, this does not indicate to us which wavenumbers are ``more'' important unlike for a random field. By important we here mean contribute to the process more than any other wavenumber. 
%\tuomas{frequency and wavenumber used side-by-side, here and later, is this ok? I like frequency as it is common in spectralspeak, but technically we should be talking about wavenumbers hmm hmm}.

We have to exercise some caution when interpreting the Fourier transform as a bijective transform. Yes, it should contain the same information about scale, but the meaning of the word ``scale'' will be different from a simple spatial understanding. In time or space the notion is associated with the support of $\gamma(z)$ or $\rho^{(2)}(z)$. Being supported over a wavenumber $\freq$ means variation over scales $1/\|\freq\|$ is present in the correlation function, or to represent the variability in $\gamma(z)$ we need wavenumbers $\|\freq\|$ present in the spectral representation. The function $\rho^{(2)}(z)-\lambda^2$ is often approached in terms of what scales it is non-zero at, but the variation in the function can also be associated with long or short scales. So for a covariance function, we now have two notions of what it means to possess scales $\|\freq\|$; either $\gamma(\mathbf{u}/\|\freq\|)$ is non-zero for unit vector $\mathbf{u}$ or $\gamma(z)$ is variable (changing) and the scale of the change is $\|\freq\|$. Imagine observing a sinusoid with period $k$. The function will hit unity at a regular set of $z$ values, and so it will be supported at those $z$ values, but also its Fourier transform will be supported at $1/k$.

This is clear from Figure~\ref{fig:fig-1} where covariance is present at smaller or medium scales smoothly for the clustered Thomas processes, or repelled for the Mat\'{e}rn hard-core processes. As the Mat\'{e}rn process' complete covariance function is discontinuous, its Fourier transform is supported over all scales (due to the Heisenberg-Gabor uncertainty principle~\cite{cohen1995time}). This can be deduced from~\eqref{eqn:rhobyscale}. To reproduce the discontinuity we need all scales in the Fourier representation.

\begin{figure}[h]
	\centering
	\includegraphics[width=0.9\linewidth]{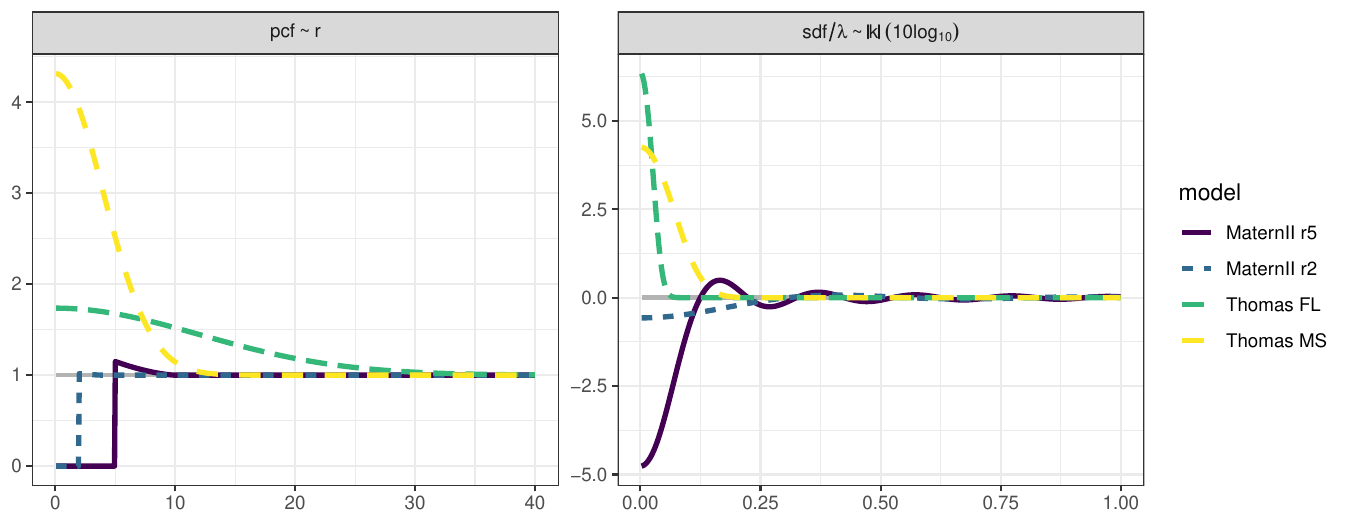}
	\caption{Pair correlation $pcf = \rho^{(2)}/\lambda^2$ (pcf, left) and corresponding scaled spectral density function $f/\lambda$ (sdf/$\lambda$, right) for two stationary and isotropic non-Poisson point process models, two variants each, with known pcfs. For the Thomas processes the sdf is given in the text. For the regular Mat\'ern II process the sdf was numerically approximated using the Hankel transform (cf. Section 7). See details of the processes and the variants in Section~\ref{sec:simulations}.}
	\label{fig:fig-1}
\end{figure}

%\david{But isn't it the Mat\'ern process that is discontinuous and leads to inhibition at small wavenumbers but also shows the montonic damping in the sdf/lambda v wavenumber? We have it as the Thomas cluster process that portrays these patterns. The Mat\'ern is a hardcore process that excludes other points in the vicinity of any single point -is this the discontinuity mentioned here?}\tuomas{Yes, was supposed to be Mat\'ern (I think). I fixed it.} 

%\david{does more important = significantly contribute to the variance in $\rho^{(2)}$ ?}\sofia{we are trying to find patterns so any frequency that stands out is important.}

We see from~\eqref{weightav} that not all wavenumbers are given equal weighting. First all wavenumbers are given an equal weighting $\lambda$ and then the Fourier transform of $\rho^{(2)}(z)-\lambda^2$ determines the wavenumbers that are up--weighted (added) or down--weighted (subtracted) relative to the overall level of $\lambda$. We then observe what wavenumbers are important to the representation of the point process, which gives more information, conveniently decomposed on a scale--by--scale manner.

For a random field or a stochastic process in $d$--dimensions a few cartoon characteristics are important. For a discretely sampled process in $\mathbb{Z}^d$ the simplest random process, white noise, is constant across wavenumbers yielding a spectrum  that takes the form of $\sigma^2$ on $[-\frac{\pi}{\Delta},\frac{\pi}{\Delta}]^d$, where $\Delta>0$ is the sampling period, and zero otherwise. For a point process
the choice of definition of the spectral density does not imply a decay because of the inclusion of the term $\lambda \delta(z)$ in space. However once we remove this term, we expect a decay of the remaining spectrum $f(\freq)-\lambda$ as   $\|\freq\|\rightarrow \infty$. Otherwise
~\eqref{weightav} would contain additional singularities.
Furthermore at $\freq=0$ we retain
\begin{equation}
f(0)=\lambda+\int \left\{\rho^{(2)}(z)-\lambda^2\right\}\;dz.
\label{weightav2}
\end{equation}
The second term in this expansion is not required to be zero, and for the Thomas process for example, it is not zero, %\david{is this well known, self-evident, or is it something we need to show?}
as it takes the form of~\cite[p~55]{mugglestone1996practical}: 
% \begin{equation}
% f(\omega)=\lambda_0+\lambda_0'e^{-\frac{1}{2}a^2\omega\cdot\omega}.
% \end{equation}
{
\begin{equation}
f(\freq)=\lambda+\lambda\mu e^{-4\pi^2\|\freq\|^2 \sigma^2}
\end{equation}
where $\mu$ is the per-cluster expected point count and $\sigma$ is the Gaussian dispersal kernel standard deviation.
% where $\kappa$ is the cluster intensity, $\mu$ is the per-cluster expected point count and $\sigma$ is the Gaussian dispersal kernel standard deviation.
}
It is clear from this expression that we cannot arrive at a zero contribution due to the exponential.

As we shall be studying the moments of a point process, it is convenient to restate
Campbell's theorem \cite[Sec. 4.3.3]{SKM}, and we shall use this result multiple times. The  theorem applies to any measurable function $h:R^{nd}\mapsto R^+$ with $n=1,2,...$, and states that (assuming product densities of order $n$ $\rho^{(n)}$ are well defined for any given $n$)
\begin{equation}
\label{eqn:Campbell}
\E \sum_{x_1,...,x_n\in X}^{\neq}h(x_1,...,x_n)=\int_{{\mathbb{R}}^{nd}}h(x_1,...,x_n)\rho^{(n)}(x_1,...,x_n)dx_1...dx_n,
\end{equation}
where the summation is over distinct point tuples. Note that if the point pattern $X$ is stationary then $ \rho^{(1)}(x)=\lambda$ is constant for any $x$, and $\rho^{(n)}(x_1,...,x_n)=\rho^{(n)}(x_1-x_n, \ldots,x_{n-1}-x_n,0) $. 
With some abuse of notation, we shall use the same symbol for non-stationary as well as stationary product densities, where the latter function has $n-1$ arguments for an $n$th order product density, e.g.\ we write $\rho^{(n)}(z_1,...,z_{n-1})=\rho^{(n)}(z_1, \ldots,z_{n-1},0)$. 

We define a new parameterisation to capture how the process's $f(\freq)$ differs from that of a Poisson process. We define the $n$th deviation of the $n$th order product density as
\begin{equation}
\label{eqn:excess}
    \widetilde{\rho}^{(n)}(z_1,\dots,z_{n-1})=\frac{{\rho}^{(n)}(z_1,\dots,z_{n-1})-\lambda^n}{\lambda^n},\quad n=2,3,\dots ,
\end{equation}
where the argument of the product density $z_l\in {\mathbb{R}}^d$ for $l=1,\dots, n-1$.
For a Poisson process these $n=2,3,\dots$ deviations from Poissonianity 
will all be identically zero. By understanding deviations from Poisson behaviour, we get greater insight into the underlying process of interest.
For completeness we also define the Fourier transform of the deviations:
\[ \widetilde{f}^{(n)}(\freq)={\cal F}\left\{\widetilde{\rho}^{(n)}(z)\right\},\]
where if $n=2$ we suppress the superscript.
With this definition we find that 
\begin{equation}
\label{eqn:fourierexcess}
    \widetilde{f}(\freq)=\frac{f(\freq)-\lambda}{\lambda^2},\quad w \in{\mathbb{R}}^d,\quad f(\freq)=\lambda+\lambda^2\widetilde{f}(\freq).
\end{equation}
Thus in a sense, $\widetilde{f}(\freq)$ encapsulates the deviation of the function from a constant spectral density of $\lambda$ via the term $\lambda^2\widetilde{f}(\freq)$.

Already \cite[p.~337]{Daley2003} discusses some differences in spectral representation of time series versus that of point processes. We would argue that decay of the spectrum is still reasonable to assume once $\lambda$ has been subtracted (so the decay of $\widetilde{f}(\freq)$ is reasonable to assume for large magnitude wavenumbers).  Having established the theoretical spectral description of point processes, we now turn to their sampling properties.

%Figure 

\section{Direct Spectral Summaries of Point Patterns}\label{sec:dirspectral}
{%\color{red} 
In this section, we shall revisit the possible definitions of a spectral estimator for a point pattern. We shall start by defining a direct Fourier transform {and} taper this definition to ameliorate edge effects, resulting in families of spectral estimators. We also discuss other choices of spectral estimators used in the literature already.}
To see the spectral characteristics of $X$ we start from what is known as Bartlett's periodogram estimator~\cite{diggle1987nonparametric} based on observing a point pattern in region $B\subset {\mathbb{R}}^d$,  written as
\begin{align}
\label{bartlett2}
I_0(\k) &:=\hat{\lambda}  + |B|^{-1}\sum_{x, y\in X\cap B}^{\neq} e^{-2\pi i\k\cdot (x-y)} ,\quad \k\in  \Rd,
\end{align}
where we have set $\hat{\lambda} =|B|^{-1}|X\cap B|$. This definition uses all the data (or points) available to us, $\{x:\;x\in X\cap B\}$, but is simply {\em a possible choice} amongst all possible direct spectral estimators.  Normally a direct spectral estimator is one which is \underline{bilinear} in the DFT of the data, see for example the discussion in~\cite[Ch 5--6]{Percival1993}. We recall a bilinear form is a function 
$D(\x,\y)$ satisfying in the first argument $D(\x_1+\x_2,\y)=D(\x_1,\y)+D(\x_2,\y)$, and equivalently in the second argument. 
Bilinear forms have been thoroughly discussed in signal processing for the analysis of time series
~\cite{loughlin1993bilinear}. As a point process consists of locations the best we can hope for in terms of bilinearity will be in terms of \underline{sesquilinearity} of the Fourier transform as the point locations will appear in the argument of the complex exponential. Sesquilinearity of $D(\x,\y)$ simply generalises a bilinear form to the Hermitian symmetry, additionally requiring $D(\x,\y)=D(\y^\ast,\x)$ as well as $D(\x,\y)$ satisfying $D(\x_1+\x_2,\y)=D(\x_1,\y)+D(\x_2,\y)$. 

In time series the bilinear form has been chosen to ensure that spectral estimators are real-valued, and often non-negative though some bilinear estimators are not,  see e.g. Guyon's spectral estimator~\cite{guyon1982parameter}. 
We define the tapered DFT of a point process $X$ for a specified general (square) integrable function $h(x)$ (referred to as a `data taper' by ~\cite{Percival1993}) with Fourier transform $H(\k)$ and unit norm (i.e. $\|h\|_2=1$), to be 
\begin{equation}
J_h(\k):=\sum_{x \in X\cap B} h(x) e^{-2\pi i\k\cdot x},\quad \k\in {\mathbb{R}}^d.
\label{taper1}
\end{equation}
Spectral estimators in time series are bilinear in the observed real-valued  process $X_t$ 
and sesquilinear in its Fourier transform $J_h(\k)$. We shall still require that the form of the estimator is sesquilinear in the DFT of the point process, but the estimators will not be bilinear in $X$, the point pattern, as this will not be possible to achieve. We shall also define
\[\widetilde{J}_h(\freq)=J_h(\freq)-\lambda H(\freq).\]
Because $x$ appears in the argument of the complex exponential, we cannot define an estimator that is bilinear directly in the data. In practice also one has to decide what taper function to use. For 1D point processes tapering has been used~\cite{cohen2014multi}, but also multitapering has been used on interpolated data~\cite{das2020multitaper}. Here, we do not interpolate, do not implement localised analysis, and do not implement analysis in 1D. It is more complex to implement interpolation in higher dimensions, as they are not orderly, unlike a time series time argument.

Based on $J_h(\k)$ the natural spectral estimator becomes its modulus square or
\begin{equation}
I_h(\k) :=\left|J_h(\k)\right|^2=\sum_{x\in X\cap B}h(x)h^*(x) + \sum_{x,y\in X\cap B}^{\neq}h(x)h^*(y)e^{-2\pi i\k\cdot (x-y)}.
\label{Ih}
\end{equation}
If we take $h(z)=h_0(z)=|B|^{-\frac12}\I(z\in B)$ then we recover Bartlett's periodogram in~\eqref{bartlett2} from~\eqref{Ih}. If we do not take the $h_0(z)$ taper then Bartlett's periodogram is not exactly recovered.
The point of using a general function $h(z)$ 
is that abruptly ending the inclusion of points when we leave the region $B$ causes ripples in wavenumber, and therefore a worse estimation of the spectrum, as is also the case of time series when not using a taper~\cite{Percival1993}. We study the whole {\em family of estimators} $I_h(\k)$, where a new member of the family is defined for each choice of $h$. We also note that~\eqref{Ih} is a direct spectral estimator~\cite[p.~207]{Percival1993}. Multidimensional tapers are available to us~\cite{hanssen1997multidimensional,simons2011spatiospectral}, and can be pre-computed. Our choice of tapering corresponds to using continuous tapers evaluated at random locations.
We have deduced an asymptotic distribution for $J_h(\k)$
using~\cite{biscio2019general} but this result cannot be applied to $I_h(\k)$ even if it superficially may seem to be of the correct form.

Note that the sum in~\eqref{Ih} cannot be evaluated in a computationally efficient manner, unlike the DFT, as the locations $\{x\}$ are \underline{not} regularly spaced, which is unfortunate, but unavoidable. Finally~\eqref{Ih} is  bilinear in 
$J_h(\freq)$ but it is not bilinear in $X$, unlike the commensurate expressions for the DFT of time series and random fields.

For completeness we also note the isotropic estimator of
\cite[Eqn.~3.3]{diggle1987nonparametric} in 2D, namely
\begin{align}
\label{diggle}
\bar{I}_0(\|\k\|) &:=\hat{\lambda}+\frac{1}{|B|}\sum_{x,y\in X\cap B}^{\neq}
\bessel_0\left(2\pi\|\k \| \|x-y\|\right),
\end{align}
where $\bessel_0(x)$ denotes the Bessel function of order 0, specified in Section \ref{sec:isotropy}. The $d$-dimensional extension is given later in equation~\eqref{diggleD}.
$\bar{I}_0$ is less clearly bilinear in the data, but if we start from the estimator $I_0$ of~\eqref{bartlett2} and average it over orientations analytically, then we arrive at this form. We therefore with a slight abuse of our terminology also refer to it as a bilinear estimator. 
There is one additional correction made by~\cite{diggle1987nonparametric}. As we shall see, for low wavenumbers there is a bias inherent in~\eqref{diggle}. To address this problem~\cite[Eqn.~3.4]{diggle1987nonparametric} suggests taking for some choice of lower bound $t_0>0,$
\begin{equation}\label{diggle4}
I_D(\|\k\|)=\left\{ \begin{array}{lcr}
\bar{I}_0(t_0) & {\mathrm{if}} & \|\k\|\leq t_0\\
\bar{I}_0(\|\k\|) & {\mathrm{if}} &  \|\k\| > t_0
\end{array}\right. .
\end{equation}
The authors of \cite{diggle1987nonparametric} suggest taking $t_0$ so that $\bar{I}_0(t_0)$ is at a minimum, and also suggest smoothing the estimated 
${I}_D$, which we clarify in Section~\ref{sec:linear}. The authors also discuss iterated methods of bias correction in their Section~5, and correcting biased estimators of the correlation. 

Note also that we have modified the estimator in~\eqref{diggle} to 
divide by 
$|B|$ rather then the observed number of points in the region, $N(B)$, as do for example~\cite{diggle1987nonparametric}, as it is much preferable not to have a random denominator.  We shall discuss the usage of $I_0(\cdot)$ versus $I_D(\cdot)$ further on in Section~\ref{sec:isotropy}, and the implications of when the process is truly isotropic or not. 

As can be surmised from~\eqref{Ih} the expectation of the general estimator $\E\left\{ I_h(\k)\right\}$ is a convolution between the Fourier transform of the observation window $B$ and the true spectrum, and if $h(z)$ does not go to zero nicely at the boundary of $B$ then the periodogram becomes quite complicated in terms of its expectation. 
Inspection of~\eqref{Ih} also raises a second problem, namely that there is a constant contribution 
$|B|^{-1}|X\cap B|$ which does not give any wavenumber specific information, and is correlated between wavenumbers. This term is new to time series and random fields, if not new to the expectation of the periodogram of randomly sampled stochastic processes, see~\cite{LiiMasry1994,chou1992spectral}.

%\tuomas{[This paragraph overlaps severely with text around eq \ref{Ih}]} We can generalise the family of estimators from Eqn.~\eqref{Ih} to take the form of 
%\begin{align}
%\label{Ih}
%I_h(\k) &:= \left| \sum_{x \in X\cap B} h(x) e^{-2\pi i\k\cdot x} %\right|^2 =|J_h(\k)|^2,\quad \k\in  \Rd.
%\end{align}

%Should we use no data taper, or equivalently use the classical taper $h_0$ from (8), then we introduce the notation of writing $I_0(\freq)$ for the periodogram for clarity. 

Should we choose to taper isotropically in 2D then we instead use an isotropic taper $h_I(\cdot)$
\begin{equation}
\label{diggle2}
    \bar{I}_{h}(\|\k\|) :=\hat{\lambda}+\sum_{x,y\in X\cap B}^{\neq}
h_I\left(\|x-y\| \right) \bessel_0\left(2\pi \|\k \| \cdot \|x-y\|\right).
\end{equation}
Note that this cannot necessarily be constructed. Say $J_h(\cdot)$ with an isotropic taper would be made with the largest spherical domain and a compact support could be constructed, but what does that imply for $I_h(\cdot)$?
Radial tapers can be determined as described in~\cite{simons2011spatiospectral}, whether in 2D continuous space or discrete space. 
Equation~\eqref{diggle2} can be of course be extended into higher dimensions. In general we would have in dimension $d=1,2,3,4,\dots$
\begin{equation}
\label{diggleD}
    \bar{I}_{h}(\|\k\|) :=\hat{\lambda}+\sum_{x,y\in X\cap B}^{\neq}
h_I\left(\|x-y\| \right)  \bessel_{d/2-1}\left(2\pi\|\k \| \|x-y\|\right) \|x-y\|^{-(d/2-1)}.
\end{equation}

We can also use more than one taper, and will use a sequence of orthogonal functions 
$\{h_\itapers\}$~\cite{Percival1993}, that will be used to get a new estimator. We will assume they are all of unit energy and are all (pairwise) orthogonal. This means we will require $||h_\itapers||_2=1$ and 
\begin{equation}
\int_{{\mathbb{R}}^d} h_\itapers(z) h_\itapersalt(z)\,dz=\delta_{\itapers\itapersalt}.
\end{equation}
%Tapering has been viewed with variable approbation. Criticism comes especially for throwing away data~\cite{fougere1977solution}. \tuomas{[are these two sentence connected, hard to see]}
We shall study the properties of direct spectral estimators of the form of Eqn~\eqref{Ih}. This will help us to characterise the second order properties of the process $X$.
Most of time series analysis is based on discrete regular sampling for instance, and so most tapers are designed for that scenario. We chose to use continuous tapers, rather than interpolating the points to a regular grid.
This leaves as possible tapers to use the spheroidal wavefunctions (continuous but hard to compute), as well as the cosine tapers of~\cite{Riedel1995}. These correspond to separable taper choices; non-separable choices with be discussed in Section \ref{sec:isotropy}.

\section{Distributional Properties of Bilinear Spectral Estimators}\label{sec:expbilin}

{%\color{red} 
In this Section we derive the asymptotic marginal distribution of the tapered periodogram, and its asymptotic expectation.
}
We shall start from the simplest spectral estimator, namely the Bartlett periodogram as specified by~\eqref{bartlett2} and calculate its properties. 
We also define the transfer function corresponding to a taper $h$ to be
\begin{equation}
\label{eqn:Hdef}
H=\Fou[h],\quad H_0=\Fou[h_0].
\end{equation} 
We shall now be quite concrete and understand some common cases of spatial data, and their sampling, and choose as early special cases Cartesian product domains.
We shall focus on the cuboid sampling domain:
%\[{B}_\square=[0,\bside_1]\times\dots \times [0,\bside_d].\]
%\tuomas{[should we centre this? 
% \[{B}_\square=[-\tfrac{\bside_1}{2}, \tfrac{\bside_1}{2}]\times\dots \times [-\tfrac{\bside_d}{2},\tfrac{\bside_d}{2}].\]
\[{B}_\square(\bm{l})=[-{\bside_1}/{2}, {\bside_1}/{2}]\times\dots \times [-{\bside_d}/{2},{\bside_d}/{2}].\]
%]}
\begin{Lemma}\label{lemma1}
Assume that $X$ is a homogeneous point process with intensity $\lambda$ and twice differentiable spectrum $f(\k)$ at all values of $\k\neq 0$. Assume $X$ is observed in a cuboid domain $B_\square(\bm{l})$ with a centroid at 0, which is growing in every dimension or $\min_j l_j\rightarrow \infty$. Then the expected value of the periodogram $I_0(\freq)$ satisfies the equation
\begin{equation}
\label{EIk}
\E I_0(\k)= {f}(\k)+ |B|^{-1}\lambda^2 T(B, \k)+o(1),\quad \k\neq 0,
\end{equation}
for 
\[T(B,\k)=\prod_{j=1}^d \frac{\sin^2(\pi \k_j\bside_j)}{(\pi \k_j)^2}.\]
\end{Lemma}
\begin{proof}
See Appendix~\ref{sec:Pf1}.
\end{proof}
This is not what we would expect as a large sample expectation of the spectrum, given our experience for time series and random fields in that the term $|B|^{-1}\lambda^2 T(B, \k)$ has been added. To get there, we need to understand the DFT
$J_h(\freq)$ better.

We note that in general $h(x) e^{-2\pi i\k\cdot x}$ is identically distributed but not independent for many choices of point processes $X$, making the choice of a Central Limit Theorem (CLT) a bit more complex. %We note that in general $x$ and $y$ which are two points in the pattern are not independent, and their dependence is characterised by the function $\gamma(z)$.  
%As a point process might have some form of dependency structure, showing asymptotic normality of the DFT requires special central limit theorems. 
For a large class of processes one such CLT is given by \cite{biscio2019general}. If we compare the quantity in \eqref{taper1} with~\cite{biscio2019general}, then we see that $q\in {\mathbb{N}}$ and $p=1$, in their notation.
We do not have to worry about $e^{-2\pi i\k\cdot x}$ being a function of several of the points, and this is why $q=1$. %Requirement ${\cal H}1$ is therefore trivially not a problem as an assumption, ${\cal H}2$ is about the scaling of the observation area, as we assume the sides scale the same not a problem \tuomas{[needs rewording]}, ${\cal H}3$ (higher order centred powers) are limited through the usage of Campbell's theorem and as $q=1$ we do not have to worry about pathological scaling between components in ${\cal H}4$. 
Citing~\cite{biscio2019general}, we can deduce from their Theorem 1 that as $|B|$ diverges $J_h(\k)$ is asymptotically Gaussian. We recover the second moments by direct calculations:
\begin{Proposition}\label{DFTprop}
Assume that $X$ is a stationary point process with intensity $\lambda$ and with spectrum $f(\freq)$, and that $h$ is a unit energy taper supported in the domain $B_\square(\bm{l})$ itself only. Then 
the first moments of the direct spectral estimator $J_h(\freq)$ are given by:
\begin{align}
\label{EDFT}
\E\left\{ J_h(\freq)\right\}&=\E \sum_{x\in X\cap B}h(x) e^{-2\pi i \freq\cdot x}\\
\nonumber
&=\lambda \int_{\mathbb{R}^d}h(x)e^{-2\pi i \freq\cdot  x}\;dx=
\lambda H(\freq),\quad \freq\in {\mathbb{R}^d}\\
\nonumber
\var\left\{ J_h(\freq)\right\}&=
\lambda  + \lambda^2\int_{\mathbb{R}^d}
\left| H(\freq'-\freq)\right|^2\widetilde{f}(\freq')\;d\freq'+\lambda^2 \left|H(\freq)\right|^2-\lambda^2 |H(\freq)|^2\\
\label{varDFT}
&=\lambda  + \lambda^2\int_{\mathbb{R}^d}
\left| H(\freq'-\freq)\right|^2\widetilde{f}(\freq')\;d\freq'\\
\rel\left\{ J_h(\freq)\right\}&=
\lambda \int_{\mathbb{R}^d}  
H(\freq'-2\freq) H(\freq')d\freq'+\int_{\mathbb{R}^d} U(\freq,z) e^{-2\pi i \freq\cdot z}\{\rho(z)-\lambda^2\}dz,
\end{align}
with $U(\freq,z)=\int h(x)h(z+x)e^{-2\pi i \freq\cdot (2x)}\;dx$.
\end{Proposition}
%\tuomas{[Need to check that we use $\freq$ and $\freq'$ consistently, sometimes $k'$ appears (I changed them here already)]}
\begin{proof}
See Appendix~\ref{DFTproppf} and notice the definition of~\eqref{eqn:Hdef}.
\end{proof}
{%\color{red} 
Why do we bother with determining the first two moments of the point pattern? The spectrum characterises any stationary/homogeneous process and so is a key summary to estimate~\cite{brillinger1974fourier}. Estimation will be linear in the periodogram, but once $J_h(\freq)$ is approximately Gaussian, the distribution of the periodogram is determined from these moments, using theory developed for quadratic forms in normal variates~\cite{johnson1970contii}.}

In the above Proposition the term $\rel\left\{\cdot\right\}$ denotes ``relation'', also known as the ``complimentary covariance'', see e.g.~\cite{miller1974complex}. 
With these moments, and the result of ~\cite{biscio2019general} we deduce the following corollary.
\begin{Corollary}
Assume $X$ satisfies the constraints of~\cite{biscio2019general}, then the DFT satisfies:
\begin{align}
    \label{CLT}
    J_h(\k)&\overset{L}{\rightarrow}N_C\left(\lambda H(\k),f(\freq),r(\freq)\right),\\
r(\freq)&=\lambda \int_{\mathbb{R}^d}  
H(\freq'-2\freq) H(\freq')d\freq'+\int_{\mathbb{R}^d} U(\freq,z) e^{-2\pi i \freq\cdot  z}\{\rho(z)-\lambda^2\}dxdz,\;
U(\freq,z)=\int h(x)h(z+x)e^{-2\pi i \freq\cdot (2x)}\;dx.
\nonumber
\end{align}
\end{Corollary}
We can view $J_h(\k)$ as a complex-valued scalar or a real-valued two-vector. $N_C(\mu,\Sigma,C)$ is the general complex normal and its arguments are its mean $\mu$, it covariance $\Sigma$, as well as its relation or complimentary covariance $C$. It should be contrasted with the complex proper normal $N_C(\mu,\Sigma),$ that has zero relation.

What can we then say about $I_h(\k)$?
Using the continuous mapping theorem~\cite{dasgupta2008asymptotic} we can deduce from~\eqref{CLT} that if we consider only arguments $\k$ so that the complementary covariance is negligible then 
\begin{equation}
     |J_h(\k)|^2\overset{L}{\rightarrow}\frac{f(\freq)}{2}\chi^2_2\left(\frac{\lambda^2 |H(\k)|^2}{f(\freq)}\right),
\end{equation}
where the parameters of the non--central $\chi^2_\nu(\mu^2)$
($\nu$ denoting the degrees of freedom, and $\mu^2$ the non-centrality parameter). We give the form of $H(\freq)$ in~\eqref{eqn:Hdef}.
%\sofia{I need to check with the multiplication of $f(\freq)$ and how it affects the non-centrality paramater.}
We can also use the uniform integrability of the variable $ J_h(\k)$, to get the asymptotic moments  of $I_0(\k)$ from these limits. 
In fact, the expectation of any member of that family of tapered estimates takes the form of:
\begin{equation}
\label{expectation}
\E\left\{ I_h(\k) \right\}= \int_{\Rd} |H(\k'-\k)|^2 f(\k')d\k' + \lambda^2 |H(\k)|^2.
\end{equation}
This can be derived straightforwardly by direct computation from~\eqref{Ih} by taking expectations with Campbell's formula and using the convolution theorem. We see immediately the bias of this estimator, namely the $\lambda^2 |H(\k)|^2$ term.
To remove the non-centrality bias, we define the bias-corrected periodogram from the de-biased discrete Fourier transform
\begin{equation}
\label{biascorrect}
\widetilde{I}_h(\freq)=
\left| \widetilde{J}_h(\freq)\right|^2.
\end{equation}
%where we (as before) have defined
%\begin{equation}
%\label{Hw}
%H(\freq)=\int_B h(x) e^{-2\pi i \freq\cdot x}\,dx.
%\end{equation}
Note also the discussion in~\cite[p.~292]{Daley2003}, where the theory of signed measures is used to make the equivalent definition, if with additional mathematical sophistication. Therefore $J_h(\freq)-\lambda H(\freq)=\widetilde{J}_h(\freq)$ can be called the signed measure.
The quantity $\widetilde{J}_h(\freq)$ is the DFT of the process $X^0$ dual to the mean-corrected random measure $N^0(dx)=N(dx) - \lambda dx$ where $N$ is the dual  counting measure of the point process $X$. As we have subtracted $H(\freq)$ off the discrete Fourier transform, it is no longer strictly positive, but once we take the modulus square we are guaranteed to arrive at a real-valued positive quantity.

There is one wavenumber which is problematic, namely $\freq=0$. For the periodogram we have $h_0(x)=\I_B(x)/\sqrt{|B|}$. Thus we have
\begin{equation}
J_0(\freq)=\frac{1}{\sqrt{|B|}}\sum_{x\in X\cap B}e^{-2\pi i \freq\cdot x},\quad 
J_0(0)=\frac{1}{\sqrt{|B|}} N(B)=\sqrt{\|B\|}\widehat{\lambda}.
\end{equation}

This removes the problem of a non-zero mean of the DFT, as long as we assume that we know the intensity $\lambda$, and so are in the position to remove this effect. Assuming knowledge of this quantity is not a major hurdle, as it can be estimated consistently for largish areas ($|B|\rightarrow \infty$), by just counting the number of points and dividing by the area. 
For completeness we also define the signed measure DFT as
\begin{equation}
\widetilde{J}_h(\freq)=J_h(\freq)-\hat{\lambda} H(\freq)\quad
.
\label{signed}
\end{equation}
We note directly from~\eqref{CLT}, that 
\begin{equation}\widetilde{J}_h(\freq)\rightarrow N_C\{0,f(\freq),r(\freq)\}.\label{CLT2}\end{equation}
Also it follows
\begin{equation}
\widetilde{J}_0(0)=\sqrt{|B|}\widehat{\lambda}-\widehat{\lambda}H_0(0)=\sqrt{|B|}\widehat{\lambda}-\sqrt{|B|}\widehat{\lambda}=0,
\end{equation}
trivially. Thus we cannot estimate the DFT at wavenumber zero. 
For a time series analysis when calculating DFTs we subtract off the sample mean, and then also the mean corrected DFT is zero at wavenumber zero. In a time series setting the periodogram at wavenumber zero is often not plotted.

%We note that two ways of bias subtraction for the periodogram is possible, namely $\widetilde{I}(\freq)$ as described by~\eqref{biascorrect} and additionally we could have defined
We note that a second way to do bias correction is by means of
\begin{equation}
\label{eqn:u2}
\breve{I}(\freq;h)= \sum_{x,y\in X} h(x)h(y) e^{-2\pi i \freq\cdot (x-y)}-\widehat{\lambda^2}\left|H(\freq)\right|^2.
\end{equation}
The advantage of the first estimator in~\eqref{biascorrect} is that it is naturally non-negative and removes the error before we take the modulus square. This also should have a positive impact on the variance. In addition, where as $\hat\lambda$ is an unbiased estimator for $\lambda$, $\widehat{\lambda^2}$ can be a biased estimator for $\lambda^2$.

For additional generality, we could consider a more sophisticated estimator than~\eqref{eqn:u2} or~\eqref{biascorrect} see e.g.~\cite[Ch.~6]{cohen1995time}, and define
\begin{equation}
\label{Ig}
    I_g(\freq)=\sum_{x\in X\cap B} \sum_{y\in X\cap B} g(x,y) e^{-2\pi i\freq\cdot \left(x-y\right)}-\widehat{\lambda^2} G(\freq,\freq),
\end{equation}
%\tuomas{[Do we need to worry about the bias term here?]}
where $g(x,y)$ is a {\em kernel} which may have a specified support, that has to be selected, and $G(\freq,\freq)=\sum_x \sum_y g(x,y)e^{-2\pi i\freq\cdot \left(x-y\right)}$. This estimator suffers from not being positive by design. 
For example we could define $g(x,y)=h(x)h^\ast(y)$ with $h(x)$ a multi-dimensional data taper ~\cite{hanssen1997multidimensional}, or we could make a non-separable choice of $g(x,y)=h(\|x-y\|)$. Depending on the choice of $g(x,y)$ the quantity $I_g(\freq)$ may satisfy a number of desirable characteristics such as positivity, asymptotic unbiasedness and  computational efficiency.

Looking at the debiased periodogram, we can now determine further properties of $I_h(\k)$, or properties of
$\widetilde{I}_h(\k)$. To produce an estimator that is smoothed we need to study the covariance and variance of $\widetilde{I}_h(\k_1)$ and 
$\widetilde{I}_h(\k_2)$.
We first look at $\E \left\{ \widetilde{I}_h(\k)\right\}$. We find that
its form for large spatial regions is specified by the following theorem. 

%The quantity $\widetilde{I}(\freq;h)$ % will be referred to as the ``unbiased periodogram'', (and the subscript unity will be removed) and 
%is preferable to the normal periodogram as it removes a number of unpalatable effects. 

\begin{Lemma}\label{lemdebias}
Assume that $X$ is a homogeneous point process with a spectral density $f(\freq)$. Then the bias-corrected tapered periodogram 
has a first moment given by:
\begin{align}
\nonumber
\E \left\{ \widetilde{I}_h(\freq)\right\}
%&=\E I_h(\k)-2\lambda\Re\left\{J_h(\freq)H^\ast(\freq) \right\}
%+\lambda^2 \left|H(\freq)\right|^2\\
&=\int_{\Rd} |H(\k'-\k)|^2 f(\k')d\k'.
\label{Ih2}
\end{align}
\end{Lemma}
%\tuomas{[Goth-R is for real part, do we need to define it? Actually do we need the step in between?]}
\begin{proof}
See Appendix~\ref{pflemdebias}.
\end{proof}
Thus as long as $H(\freq)$ is getting more concentrated in wavenumber (e.g. $H(\freq)\rightarrow \delta(\freq)$), this is asymptotically in $|B|$ an unbiased estimator of $f(\freq)$.
Using the signed measure DFT of~\eqref{signed} to study $X$ is more convenient, as this lets us avoid the contribution that affects the low wavenumbers. 

Let us write down what grid we get large sample unbiasedness for, customarily called the \emph{Fourier grid}, and this is useful when the observational domain $B$ is a box.
\begin{Definition}
The Fourier wavenumber grid for a point process observed on a cuboid domain $B_\square(\bm{l})$ corresponds to the points
\begin{equation}
{\cal K}(\bm{l}):={\cal K}(B_\square(\bm{l}))=\{\k_n \;:\;
\k_n=\begin{pmatrix} p_{n1}/l_1,\dots, p_{nd}/l_d\end{pmatrix},\quad p_{nj}\in{\mathbb{Z}}\}.
\end{equation}
\end{Definition}
Note that the physical units of the wavenumber elements is per unit length, such as $m^{-1}$. Referring to Lemma~\ref{lemma1} we see that the zero wavenumber, or taper $h_0$ related bias term $T(B,\freq)$, vanishes using this grid not only removing any asymptotic bias at low wavenumbers, but also giving us a sampling of the wavenumbers that approximately leads to independent periodogram ordinates, rather like in the random field case, see Corollary~\ref{Propcovar2}. %\tuomas{[I can't see the latter claim from the lemma. Please elaborate. Might need to forward-refer to some covariance result]}. 

% Note that the centring of the box $B_\square$ at zero, e.g. requiring the box to take the form
% \[B_\square=\begin{pmatrix} 0 & \bside_1\end{pmatrix}\times \dots \times 
% \begin{pmatrix} 0 & \bside_d\end{pmatrix},
% \]
% is required for the stated result.
These results establish what  wavenumber grid we should evaluate a standard spectral estimator of a time series at, in analogy to the the {\em{Fourier frequencies}}~\cite[p.~197--198]{Percival1993} in time series. Their basic importance follows because we can expect the direct Fourier transform to be Gaussian and so uncorrelated implies independence. 

A second feature of time series is the notion of the Nyquist wavenumber~\cite[p.~98]{Percival1993}. This does not exist for point processes.
It may seem counter intuitive that there is no upper limit to the wavenumbers we can estimate. When analysing a process that has been sampled in space, such as a random field, we expect to see aliasing. Aliasing is when variation that is happening very rapidly is confused with slower cycles, because when we have sampled more sparsely, rapid variation cannot be resolved. For random locations of the point process the pairwise distances can be any real-valued value, so the Nyquist wavenumber does not (in a sense) exist.

\begin{Theorem}{Large--Domain Expectation of the tapered periodogram.}\label{thm1}
Assume that $X$ is a stationary point process with intensity $\lambda$ and twice differentiable spectrum $f(\k)$ at all values of $\k\neq 0$, and that $h$ is a unit energy taper (e.g. $\|h\|=1$) supported in the cuboid domain $B_\square(\bm{l})$ only. Assume $X$ is observed in $B_\square(\bm{l})$, which is growing in every dimension, that is $\min \bside_j\rightarrow \infty$. Then the expected value of the tapered periodogram $I_h(\freq)$ satisfies the equation
\begin{equation}
\label{Prop2expr}
\E I_h(\freq)= f(\freq)+\lambda^2|H(\freq)|^2+o(1),\quad w\neq 0,
\end{equation}
where $H(\freq)$ is the Fourier transform of $h(x)$, as defined by \eqref{eqn:Hdef}.
\end{Theorem}
\begin{proof}
See Appendix~\ref{sec:Pf2}.
\end{proof}
{%\color{red} 
This theorem establishes that the tapered periodogram is a biased estimator of the spectrum of the point pattern. Previous authors made up an ad-hoc corrections~\cite{diggle1987nonparametric} to remove the bias. Our result will suggest a method to arrive at a positive spectral density estimator that is not biased. This is the basic and important result; equivalent to~\cite[Section V]{brillinger1974fourier}. 
In fact, asymptotics are well--understood both for time series and random fields. Our (previous) understanding of asymptotics for spectral estimation of point processes hearkens back to~\cite{Daley2003} or \cite{diggle1987nonparametric}. In the latter reference the authors informally refer back to Chemistry theory 
by Hansen and McDonald, discussing the properties of gasses.}

%\begin{Corollary}\label{prop1}
%Assume that $X$ is a stationary point process %with intensity $\lambda$ and twice %differentiable spectrum $f(\k)$ at all values %of $\k\neq 0$, and that $h$ is a unit energy taper supported in the domain $B$ only. Assume $B$ is growing in every dimension. Then the expected value of the tapered periodogram $I_h(\k)$ satisfies the equation
%\begin{equation}
%\label{Prop2expr3}
%\E I_h(\k)= %f(\k)+\lambda^2|H(\freq)|^2+o(1),\quad \k\neq 0,
%\end{equation}
%where $H(\k)$ is the Fourier transform of %$h(x)$. 
%\end{Corollary}
%
%\begin{proof}
%The proof can be found in Appendix~\ref{pfthm2}.
%\end{proof}
%This result is not unexpected and generalises Theorem~\ref{thm1}. The main difference here is that we consider a general domain $B$ rather than a cuboid domain $B_\square$. 
%This theorem allows us to determine the expectation of $\widetilde{I}_h(\freq)$, in general, as is made clear in the following corollary.

\begin{Corollary}{Bias--corrected periodogram.}\label{cor1}
Assume that $X$ is a stationary point process with intensity $\lambda$ and twice differentiable spectrum $f(\k)$ at all values of $\k\neq 0$, and that $h$ is a unit energy taper supported in the cuboid domain $B_\square(\bm{l})$ only. Assume $X$ is observed in $B_\square(\bm{l})$, which is growing in every dimension. Then the expected value of the tapered periodogram $\widetilde{I}_h(\k)$ satisfies the equation
\begin{equation}
\label{Prop2expr2}
\E \widetilde{I}_h(\k)= f(\k)+o(1),\quad \k\neq 0.
\end{equation} 
%\tuomas{[remove this perhaps, since it is derived two paragraphs down]}\sofia{OK, I am not quite sure what happened here.}

\end{Corollary}
\begin{proof}
See Appendix~\ref{proofcor1}.
\end{proof}

{%\color{red} 
This corollary then informs us how to do estimation, by removing the bias.} 
This corollary shows the importance of removing the spectral bias before squaring -- otherwise it gets hard to isolate and remove the effect at zero wavenumber, which is exacerbated at higher intensities (growing $\lambda$), as is clear from~\eqref{Prop2expr}. 
Note that the estimator $\widetilde{I}_h(\k)$ is debiased relative to $\E\left\{  I_h(\freq)\right\}$, but is still guaranteed to be non-negative. This can be compared to removing $\widehat{\mu}$ from a time series before calculating the periodogram.

%This gives us insight into the act of debiasing. In a sense our understanding of the spectral nature of point processes is informed by the spectral representation of randomly sampled processes~\cite{chou1992spectral}. We can continue to explore
%the differences and similarities between spatial point processes and random fields.

Finally, to employ smoothing for purposes of variance reduction, we will extend the tapered estimator to include multiple tapers. Define $\ntapers\ge 1$ estimates of the spectrum via
\begin{equation}
I_{{\itapers}}^t(\k) =\sum_{x,y\in X \cap B}h_{\itapers}(x)e^{-i2\pi \k\cdot x} h^\ast_{\itapers}(y)e^{i2\pi \k\cdot y},\quad \itapers=1,\dots, \ntapers,
\label{tapered}
\end{equation}
which is not bias-corrected but
%\tuomas{[How about the bias term here?]} \sofia{yes, also it is all messed up. since we are tapering in both spatial directions we should cycle over numbering in both-I'll return to this} 
where we assume
\[\int_{\Rd} h_{\itapers}(x) h^\ast_\itapersalt(x)dx = \delta_{\itapers \itapersalt},\]
where $\delta_{\itapers \itapersalt}=1$ if $\itapers=\itapersalt$ and $0$ otherwise. Of course~\eqref{tapered} 
still suffers from low-wavenumber bias. To define a debiased estimator we take
{{\begin{equation}
\widetilde{I}_{{\itapers}}^t(\k) =\left(\sum_{x\in X \cap B}h_\itapers(x)e^{-i2\pi \k\cdot x}-\lambda H_\itapers(\k)\right) \left(\sum_{y\in X \cap B} h^\ast_\itapers(y)e^{i2\pi \k\cdot y}-\lambda H^\ast_\itapers(\k)\right),\quad \itapers=1,\dots, \ntapers
.\label{tapered-debias}
\end{equation}}}
We subsequently average these estimates over $\itapers$ to reduce variance.

In~\eqref{tapered} we have a product of $h_\itapers(x)h^\ast_\itapers(y)$. This is an inevitable consequence of the bilinear form of the periodogram from the Fourier transform. The most natural way of making multidimensional and separable tapers would be for $x=\begin{pmatrix} x_1 & \dots & x_d\end{pmatrix}^T$ to define the component-wise product taper
$h_\itapers(x)=\tilde{h}_{\itapers_1}(x_1) \dots \tilde{h}_{\itapers_d}(x_d)$. We will then need to re-enumerate $\itapers_1,\dots, \itapers_d$ into a single index. For instance in 2D if we have three tapers in 1D, then we will end up with nine tapers in 2D, and $\itapers$ will range from one to nine. 
These estimates can also be arrived at starting from~\eqref{taper1} and $J_h(\k)$. As we shall use $\ntapers$ tapers linearly, and $\ntapers^2$ for the quadratic estimators, we shall use the subscript `0' to refer to the untapered periodogram, and $\itapers=1,\dots,\ntapers$ to refer to the subsequent tapers. 

Tapering is very common for stochastic processes. Initially the idea was hotly contested, but its utility is now firmly established both for time series and random fields~\cite{dahlhaus1987edge,walden2000unified}. The idea of tapering is to ameliorate edge effects that leads to the asymptotically leading bias term. This is important in 1D~\cite{thomson1982spectrum,walden2000unified} but of greater importance in 2D and higher~\cite{dahlhaus1987edge}. In 2D and higher the edge effects become more pronounced, and indeed asymptotically dominant~\cite{kent1996spectral,mardia1984maximum}. We also use multitapering to stabilise the variance of any estimator, as described for time series in~\cite{walden2000unified}. We shall describe the necessary steps to perform linear estimation of the spectrum of a 2D and higher dimension point process in Section \ref{sec:linear}. 

Let us set theory aside for a moment and consider some spectral estimates for archetypal point patterns. Figure \ref{fig:dataexample} demonstrates the debiased periodogram, the debiased multitapered periodogram, with sine-tapers, and a rotationally averaged 1D summary of the peridogram (defined in Sections \ref{sec:linear}\&\ref{sec:isotropy}) for four point patterns exhibiting different structural behaviour. The aspect ratio of the figure panels are kept the same only to have an orderly figure: Due to debiasing we can estimate the periodogram on wavenumbers of our own choosing, and not just on the Fourier grid which in these examples would be different for different tapers. The  wavenumber-grid in this illustration is a regular grid with a step size 0.005 in both dimensions. 

We see that information is present at wavenumbers up to $\|\freq\|\approx 0.2$. The tapering smooths the periodogram, as expected. There is a dark well near the origin for the regular pattern, and a bright hump for the clustered pattern, and both features transfer to the rotationally averaged curves (compare with Figure \ref{fig:fig-1}). The anisotropy of the fourth pattern is hinted by the anisotropy visible in the 2D periodograms as elongation of the ``hump'' in the second dimension. The elongation is perpendicular to the elongations of the clusters in the data because wavenumbers relate inversely to spatial units. The rotationally averaged periodogram was also computed from the (not shown) non-debiased periodogram for comparison. There is a very prominent bias near $\|\freq\|=0$ for the non-debiased version, which illustrates why the proposed debiasing step is relevant when applying the method in practice. 

\begin{figure}[!ht]
	\centering
	\includegraphics[width=0.95\linewidth]{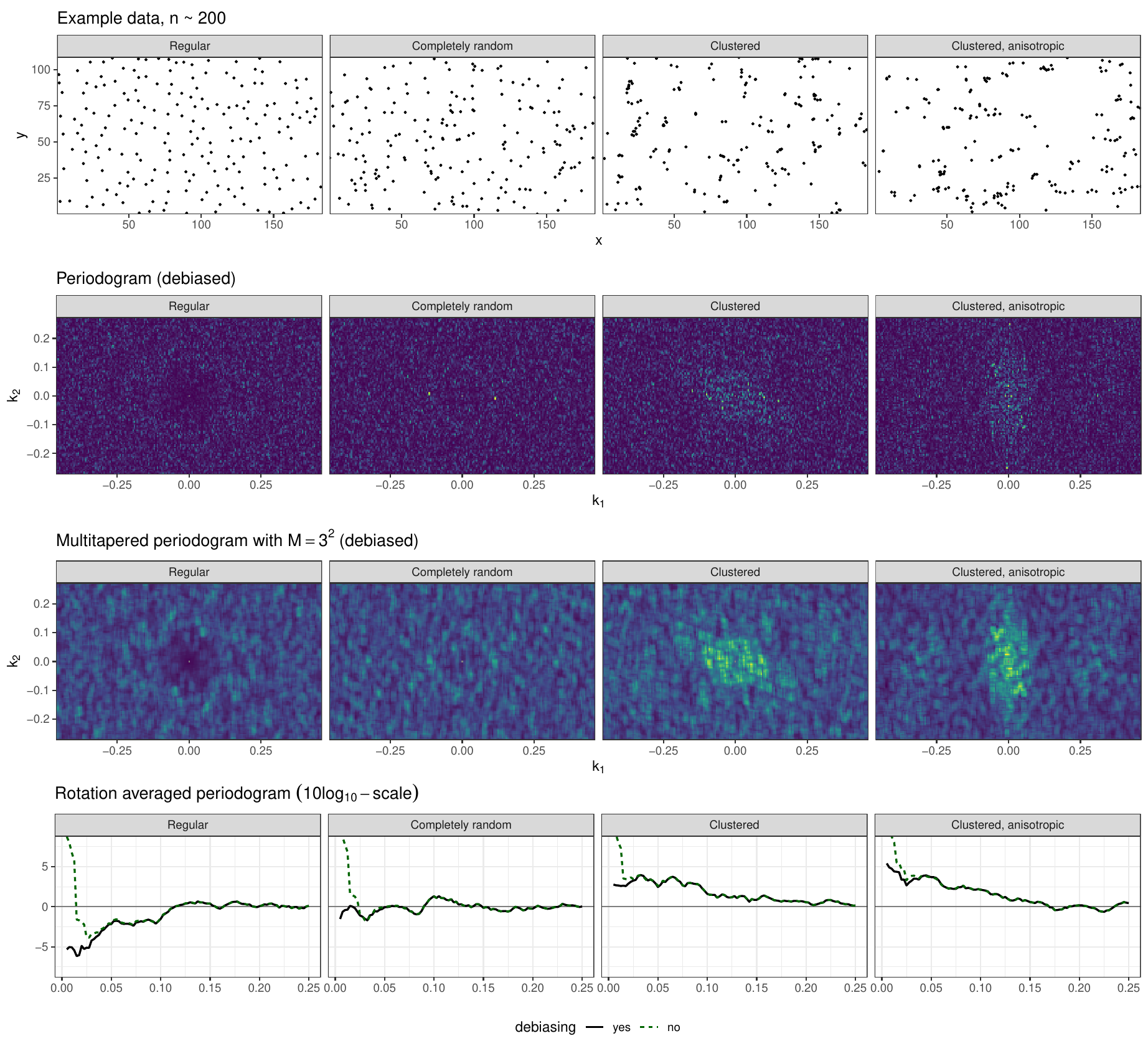}
	\caption{\emph{Top row}: Four example data patterns in a rectangle of area 2e4 units, exhibiting regularity, complete randomness, small scale clustering, and clustering with a left-right axis directionality. \emph{Second row}: Their debiased periodograms excluding $\k=0$. \emph{Third row}: Their multitapered debiased periodograms, $3\times 3$ sine-tapers. \emph{Bottom row}: Their rotationally averaged (un-tapered) periodograms, including the case without debiasing. The 2D summaries use a $\k$-grid with a step size $0.005$ in both directions. The rotational averaging was estimated at $\|\freq\|=0.0050, 0.0075, ..., 0.2500$ using a box kernel having a radius 1.5 times the $\k$-grid step size 0.005.}
	\label{fig:dataexample}
\end{figure}

\section{Covariance of bilinear spectral estimators}\label{sec:varbilinear}
{%\color{red} 
In this Section we determine the covariance of the periodogram, in anticipation of smoothing the periodogram.}
If we compare the developments of the previous sections to that of spectral analysis, it would seem that we are in a good position to estimate the spectrum, especially as we can assume the spectral deviation function $\widetilde{f}(\freq)$ is smooth. The smoothness of $\widetilde{f}(\freq)$ follows from the decay of $\widetilde{\rho}(z)$. Simple Fourier theory stipulates 
that the decay of $\widetilde{\rho}(z)$ yields the smoothness of $\widetilde{f}(\freq)$. From a mean-square error perspective on estimating the spectrum, to understand how to smooth the periodogram away from zero wavenumber, and to do so we need to further study the variance and covariance of direct spectral estimators. 

%\david{something that is not possible for the spectrum that contains a Dirac delta function} \david{have attempted to rewrite -do we need to link back to equation?}\sofia{this is why we defined the tilde version, once we subtract out the constant which yields the spike at  zero for the non-tilde version.}
%a proposition that would have been ridiculous to assume for the  spectrum, that contains a delta Dirac contribution.
%\tuomas{[this sentence could use rewording, casual]} However, to understand smoothing we need to have a notion of the variance and covariance of the direct spectral estimator at different wavenumbers. %\tuomas{removed this}This becomes difficult for a generic point pattern because even if stationarity holds, higher order moment specifications may matter. 

Our main concern is %\tuomas{}then
: \textbf{1)} Is the variance finite? \textbf{2)} Can we find a grid of wavenumbers so that the estimated spectrum is uncorrelated at these points? %\tuomas{[This point is unclear perhaps reword. Doesn't averaging almost always reduce variance? Should this be about the covariance?]} 
To be able to answer such questions we must study what the variance and covariance of the DFT is. 
%\tuomas{}Before determining how to do bilinear estimation of the spectral density let 
Let us determine the second order properties of spectral estimators. Core to our understanding of smoothing will be the variance and covariance of the tapered functions. This will be established in the following theorem. For brevity define
\begin{equation}
    \phi_\k(x)=e^{-2\pi i \k\cdot x },\quad \k,x\in{\mathbb{R}}^d.
\end{equation}
With this notation we have
\begin{equation}
I^t_\itapers(\k) =\sum_{x,y\in X\cap B}h_{\itapers}(x)h_{\itapers}^\ast(y)\phi_\k(x-y),\quad \itapers=1,\dots \ntapers.
\label{tapered3}
\end{equation}
Above we note that we have the same value of $\itapers$ across $x$ and $y$ in~\eqref{tapered3} as this corresponds to a modulus square.

\begin{Theorem}\label{Propcovar}
Let $I^t_\itapers(\freq)$ denote the (tapered) periodogram given by
~\eqref{tapered}. The covariance between the (tapered)  periodogram and itself across two wavenumbers $\k_1,\k_2\in{\mathbb{R}}^d$ then takes the form of
\begin{align*}
\cov& \{I^t_\itapers(\k_1) , I^t_\itapersalt(\k_2)\}
=\int_{B^4}
 \rho^{(4)}(x,y,z,v)
 h_\itapers(x) h_\itapers(y)h_\itapersalt(z) h_\itapersalt(v)\phi_{\k_1}(x-y) \phi_{\k_2}(z-v)\;dxdydvdz\\
 \nonumber
 &+\int_{B^3} \rho^{(3)}(x,y,v)
 \left[h_\itapers(x) h_\itapers(y)h_\itapersalt(x) h_\itapersalt(v)\phi_{\k_1}(x-y)\left\{ \phi_{\k_2}(x-v)+  \phi_{-\k_2}(x-v)\right\}\right.\\
&\quad+h_\itapers(x) h_\itapers(y)h_\itapersalt(y) h_\itapersalt(v)\phi_{\k_1}(x-y)\left\{ \phi_{\k_2}(y-v)+  \phi_{-\k_2}(y-v)\right\}\\
&\quad\left.+\phi_{\k_2}(x-y)h_\itapers^2(v)h_\itapersalt(x)h_\itapersalt(y)+\phi_{\k_1}(x-y) h_\itapersalt^2(v)h_\itapers(x)h_\itapers(y)\right]\;dxdydv\\
&+\int_{B^2}\rho^{(2)}(x,y)\left[ h_\itapers(x) h_\itapers(y)h_\itapersalt(x) h_\itapersalt(y)\phi_{\k_1}(x-y)\left\{ 
\phi_{\k_2}(x-y)+\phi_{-\k_2}(x-y)\right\}\right.\\
&\quad+\left\{h_\itapers^2(x)h_\itapersalt(x)h_\itapersalt(y)+ h_\itapersalt^2(x)h_\itapers(x)h_\itapers(y)\right\}
\left.\left\{\phi_{\k_1}(x-y)+\phi_{\k_2}(x-y)\right\}+h_\itapers^2(x)h_\itapersalt^2(y)\right]\,dxdy\\
 &+\lambda 
 \int_{B} h_\itapers^2(x)h_\itapersalt^2(x)\,dx -\lambda^2-\lambda \int_{B^2} h_\itapers(x)h_\itapers(y)h_\itapersalt(x)h_\itapersalt(y)\left[\phi_{\k_1}(x-y)+\phi_{\k_2}(x-y)\right]\rtwo(x,y)dxdy\\
 &-\left\{  \int_{B^2} h_\itapers(x)h_\itapers(y)\phi_{\k_1}(x-y)\rtwo(x,y)dxdy \right\}\cdot \left\{  \int_{B^2} h_\itapersalt(x)h_\itapersalt(y)\phi_{\k_2}(x-y)\rtwo(x,y)dxdy \right\}.
\end{align*}
 %\tuomas{[$\freq_i$ change to something, we use this notation for components (i.e. dimensions)? E.g. Figure 2.] Also, in the formula, should we go with the usual parenthesis order $\left\{[(-)]\right\}$?]}
 %\sofia{i am not quite sure what you propose for the former? Biometrika bracket convention? Sure. I'll do it after i've finished or you do it :)}
\end{Theorem}
\begin{proof}
The proof is in Appendix~\ref{sec:Pf3}.
\end{proof}
%Note that we can in fact write:
%\[\widetilde{I}_h(\freq)=I_h(\freq)-2\lambda\Re\left\{ J_h(\freq)H^\ast(\freq)\right\}
%+\lambda^2 \left|H(\freq)\right|^2.\]\tuomas{[Little orphan paragraph, lost in the blizzard??]}
 
This theorem derives the general expression for the covariance of the periodogram, and will help us in general to determine how to do linear smoothing.
The cross-correlation also relates to~\cite[Section V]{brillinger1974fourier}, and helps us understand when we can treat Fourier coefficients as uncorrelated (and with asymptotic Gaussianity as independent). Uncorrelatedness is important to the smoothing, as it ensures a variance reduction by averaging.
Whilst, the expression derived in Theorem~\ref{Propcovar} is general, and only requires the assumptions of homogeneity of the point process in space, it is hard to understand so let us study some special cases.
%\david{[The next sentence doesn't make sense -can someone rewrite?]} 
%\jake{This will be important for determining how to do %multitaper analysis, so that the multitapering will then ensure %variance reduction.
%Whilst, the expression derived in Theorem~\ref{Propcovar} is general, and only requires the assumptions of homogeneity of the point process in space, it is hard to understand.
%Especially as it is hard to order the terms in size without imposing restrictive assumptions.}
%\tuomas{}We can still consider the untapered case by setting $p=l$ and taking $h(x)=h_0(x)$. \tuomas{}%|B|^{-1/2}\I_B(x)$,
%to arrive at an expression. 
Let us study this general correlation in the instance of the Poisson process where we know that
\begin{equation}
\label{assumpPoisson}
\rho^{(n)}(x_1,\dots, x_n)=\lambda^n,
\end{equation}
and consider the case of $\k_1=\k_2=\k$.

For completeness we also define the spectral bandwidth $b_h$ by
\begin{align}
\label{eqn:defbandwidth}
    b^2_h&:=\int_{{\mathbb{R}}^d}\|\k\|^2 \left|H(\k) \right|^2d\k.
\end{align}
%\tuomas{[Should we emphasise all definitions by ":=" because they pop up all along the text?]}
%ok.
\begin{Corollary}[Covariance of Spectral Estimates]\label{Propcovar2}\label{Propcovar3}
Let $I^t_\itapers(\k)$ denote the tapered periodogram given by
~\eqref{tapered} for a Poisson process with intensity $\lambda$ using a single taper $\itapers$, and assume that $\|k\|>\max\{b_{h_{\itapers}},b_{h_{\itapersalt}}\}$. We can determine the covariance between the periodogram using different tapers ($\itapers$ and $\itapersalt$) or at different frequencies by the following:
\begin{align*}
\cov& \{I^t_\itapers(\k) , I^t_\itapersalt(\k)\}=\lambda 
 \int_{B} h_\itapers^2(x)h_\itapersalt^2(x)\,dx+\lambda^2\delta_{\itapers\itapersalt}+o(1).
\end{align*}
Note that $o(1)$ is in $|B|$ diverging.
The covariance between the (tapered) periodogram and itself at different wavenumbers ($\k_1$ and $\k_2$) takes the form of
\begin{align*}
\cov& \{I^t_\itapers(\k_1) , I^t_\itapers(\k_2)\}=\lambda \|h_\itapers\|_4^4=o(1).
\end{align*}
\end{Corollary}
\begin{proof}
The proof can be found in Appendices~\ref{sec:Pf4} and~\ref{sec:Pf5}. 
\end{proof}
At this stage we have simplified our assumptions to the Poisson process which is quite disappointing. The reason why we have chosen to do so is clear from Theorem~\ref{Propcovar}. If we assume the process is Poisson then several terms cancel. However if the Fourier transform is turning Gaussian 
as we have discussed in~\eqref{CLT} in Section~\ref{sec:dirspectral}, with uniform integrability of $J_h(\k)$ the following proposition holds.
\begin{Proposition}\label{Propcovar3b}
Let $I^t_\itapers(\k)$ denote the tapered periodogram given by~\eqref{tapered}, and $\tilde{I}^t_\itapers(\k)$ the tapered bias corrected periodogram. Assume $X$ satisfies the assumptions given for~\eqref{CLT} and that $\|\widetilde{f}^{(m)}\|_0<\infty$ for $m=2,3,4,5,6$. Then
\begin{align*}
\cov \{\widetilde{I}^t_\itapers(\k_1) , \widetilde{I}^t_\itapersalt(\k_2)\}
=o(1)+\left|\E\{
\widetilde{J}_\itapers(\k_1)\widetilde{J}^\ast_\itapersalt(\k_2)\}\right|^2.
\end{align*}
To derive the form of the latter term we note
\begin{align}
\nonumber
 \E\{
\widetilde{J}_\itapers(\k_1)\widetilde{J}^\ast_\itapersalt(\k_2)\}&=   \cov \left\{  {J}_\itapers(\k_1),{J}_\itapersalt(\k_2)\right\}=\lambda\delta_{\itapers\itapersalt}+ \lambda^2\cov_1 \left\{ {J}_\itapers(\k_1),{J}_\itapersalt(\k_2)\right\},
    \end{align}
    where we have $\tilde{f}(\freq)$ given by \eqref{eqn:fourierexcess} and so
    \begin{align}
    \nonumber
        \cov_1 \left\{ {J}_\itapers(\k_1),{J}_\itapersalt(\k_2)\right\}
        &\equiv \iint_{R^d\times R^d}\widetilde{\rho}(x-y)h_\itapers(x)e^{-2\pi i \freq_1\cdot  x}h^\ast_\itapersalt(y)e^{2\pi i \freq_2\cdot  y}\,dx\,dy+o(1)\\
        %&=\int_{R^d\times R^d} 
        %\widetilde{\rho}^{(2)}(x-y)h_\itapers(x)\exp%(-2\pi i \freq_1\cdot  x)h^\ast_\itapersalt(y)e^{2\pi i \freq_2\cdot  y}\,dx\,dy\\
        %\nonumber
        %&=\int_{R^d\times R^d} 
        %\int_{R^d}\widetilde{f}^{(2)}(\freq')e^{2\pi i (x-y)\cdot \freq'}d\freq' \cdot h_\itapers(x)e^{-2\pi i \freq_1\cdot  x}h^\ast_\itapersalt(y)e^{2\pi i \freq_2\cdot  y}\,dx\,dy\\
        \nonumber
        &=\iint_{R^d\times R^d} 
        \int_{R^d}\widetilde{f}(\freq')e^{-2\pi i x\cdot(\freq_1-\freq')}e^{2\pi i y\cdot(\freq_2-\freq')}d\freq' \cdot h_\itapers(x)h^\ast_\itapersalt(y)\,dx\,dy+o(1)\\
        &=\int_{R^d}\widetilde{f}(\freq') H_\itapers(\freq_1-\freq') H^\ast_\itapersalt(\freq_2-\freq')d\freq'+o(1) .
    \end{align}
    %\sofia{Checked.}
    We note that as $H_\itapers(\freq)$ has concentrated support we will be able to apply similar arguments to those of Appendix D to determine nearly uncorrelated DFTs.
\end{Proposition}
\begin{proof}
See Appendix~\ref{proofcovar3b}
\end{proof}

\begin{Definition}
The tapered Fourier wavenumber grid for a point process observed on a cuboid domain $B_\square(\bm{l})$ corresponds to the points
\begin{equation}
{\cal K}_\ntapers(B_\square(\bm{l}))=\{\k_n \;:\;
\k_n=\begin{pmatrix} p_{n1}\tau(\epsilon) /l_1,\dots, p_{nd}\tau(\epsilon)/l_d\end{pmatrix},\quad p_{nj}\in{\mathbb{Z}}\},
\end{equation}
where $\tau(\epsilon)$ is the smallest positive real number such that $\left|\int H_\itapers(\tau(\epsilon)-\k') H^\ast_\itapersalt(\k') \;d\k'\right|<\epsilon,\,$ for all $\itapers,\itapersalt\in\{1\ldots\ntapers\}$.
% \[\tau(\epsilon)=\min{\tau} \quad \text{such that for } 1\leq\itapers,\itapersalt\leq\ntapers \quad \left|\int H_\itapers(\tau(\epsilon)-\k') H^\ast_\itapersalt(\k') \;d\k'\right|<\epsilon.\]
\end{Definition}

These results hint at producing general methodology; we can for the Poisson process figure out how to do non-parametric estimation. This is however not enough  and we need to determine the variance of the periodogram more generally. %\tuomas{}The following proposition establishes its form.
To go further than Proposition~\ref{Propcovar3} and avoid the Poisson assumption we refer back to~\eqref{CLT}. As ${J}_h(\freq)$ is becoming Gaussian for larger sample sizes this implies its moments must also start to behave like the moments of the Gaussian  ${J}_h(\freq)$. 
%We determine the consequences of this in the following Proposition.

Note that this result is not contradictory to Theorem~\ref{Propcovar} from a dimensional perspective as this theorem concerns the modulus square of the Fourier transform.
With these moments we can now consider the problem of estimation. 
We might feel an increasing degree of unease as we have succumbed to the usual fallacy of only deriving the most general results for a Poisson process. However as the observational domain becomes larger we would expect $|H_\itapers(\freq)|$ to concentrate and so having a non-constant spectrum will become less of an issue. %This leads to the following result.

%%%%%%%%%%%%%%%%%%

\section{Linear Smoothing}\label{sec:linear}
{%\color{red} 
Having already derived the first and second order properties of the periodogram, in this Section
we propose an accurate estimator of $f(\k)$ formed by smoothing. %\tuomas{}, as the primal quantity of interest in studying the point process $X$
}
If we fully characterise $f(\k)$ then we have fully characterised $\gamma(z)$, as the Fourier transform is a bijection. Just like for a time series or random field, a key question concerns the correct method for quadratic characterisation and estimation.

%\sofia{Got to here.}

%tuomas{}For the process $X$ itself, then its first moment is encapsulated by $\lambda$ for a stationary (or homogeneous) processes. We recall the spectrum uniquely maps to the complete covariance function $\gamma(z)$.  
We want to start from the random variables $\{I_h(\freq)\}$ or $\{\widetilde{I}_h(\freq)\}$ and directly from those quantities (non-parametrically) estimate the spectral density. As long as we seek to estimate $f(\k)$ at points $\k$ where the function $f(\k)$ is continuous,
averaging $\{I_0(\k)\}$ (or $\{\widetilde{I}_h(\freq)\}$) locally seems like a sensible strategy. 
\cite{diggle1987nonparametric} proposed smoothing the 2D isotropic $\widetilde{I}_D(\k)$ as defined in~\eqref{diggle4}, and applied a weighted average smoother with user tuned weights. This approach, however, does  not determine how to sample wavenumbers, or remove correlation. We would additionally like to include the new form of bias removal before smoothing.

We have two clear options available to us, either smoothing the raw periodogram at the Fourier frequencies (de-biasing is not needed at the  Fourier frequencies, because we defined the grid such that the bias will be 0), or by using multitapering. We then have two possible estimators, namely the bias--corrected multi-taper estimator of~\cite{walden2000unified} as well as 
a multi-dimensional kernel density estimator~\cite{duong2005cross}. They are both linear in the estimated spectrum. We shall use the periodogram rather than another estimator for the kernel density estimation, as it is easier to keep track of bias and correlation.

We therefore define the multitaper estimator~\cite{walden2000unified} (and refer to~\cite{hanssen1997multidimensional,liu1992multiple})
\begin{equation}
    \overline{I}_{\ntapers}^t(\freq)=\frac{1}{\ntapers}\sum_{\itapers=1}^{\ntapers}
    \widetilde{I}^t_{\itapers}(\freq),\quad \freq\in{\mathbb{R}}^d,
    \label{eqn:multitaper}
\end{equation}
where $ \widetilde{I}^t_{\itapers}(\freq)$ is defined in~\eqref{tapered}. 
% Choosing to use continuous tapers in~\eqref{biascorrect} is not the only option, and this definition allows for nonseparable tapers~\cite{simons2011spatiospectral}, as it does not loop over Cartesian product tapers. 

Why is it advantageous to use the estimator in~\eqref{eqn:multitaper}? Because $\widetilde{I}^t_{\itapers}(\freq)$ is uncorrelated across $\itapers=1,\dots, \ntapers$ it follows that the variance of $\overline{I}^t_\ntapers(\freq)$ decreases like $1/\ntapers$. Or
\begin{align}
\nonumber
    \var\{  \overline{I}_{\ntapers}^t(\freq)\}&=\frac{1}{\ntapers^2}\sum_{\itapers=1}^{\ntapers}\sum_{\itapersalt=1}^{\ntapers}
   \cov\{ \widetilde{I}^t_{\itapers}(\freq),\widetilde{I}^t_{\itapersalt}(\freq)\}\\
   \nonumber
   &=\frac{1}{P{\ntapers}^2}\sum_{\itapers=1}^{\ntapers}\sum_{\itapersalt=1}^{\ntapers}
  \left(\lambda \delta_{\itapers \itapersalt}+\lambda^2 (\tilde{f}(\freq) +o(1))\delta_{\itapers \itapersalt}\right)^2\\
  \nonumber
   &=\frac{\lambda^2}{\ntapers^2}\sum_{\itapers=1}^{\ntapers}\left(1+\lambda \tilde{f}(\freq) \right)^2+o(1/\ntapers)
   =\frac{\lambda^2}{\ntapers}\left\{1+\lambda \tilde{f}(\freq) \right\}^2+o(1/\ntapers)\\&=
   \frac{1}{\ntapers} \{f(\freq)+o(1)\}^2 + o(1/\ntapers),
\end{align}
where the last equality follows from Lemma~\ref{Lemmasimplify}, which is a consequence of Proposition~\ref{Propcovar3b}.
We shall explore these results in practical scenarios in simulations in Section~\ref{sec:simulations}. 

An additional useful modification is to consider a one-dimensional smoothing-function $W(\k)$ and a bandwidth matrix $\bm{\Omega}$ and implement multidimensional smoothing, see e.g.~\cite{duong2005cross}, and   
averaging over directions, thus reducing the statistic to magnitude $t=\|\k\|$ only, 
\begin{equation}
\label{angularaveragedI}
    \overline{I}_h(t;\Omega)=\sum_{\k'\in {\cal K}} {\Omega}^{-1}W\left({\Omega}^{-1}\left(||\k'||- t \right) \right)\widetilde{I}_h\left(\k'\right),
\end{equation}
which we refer to as the ``rotationally averaged periodogram''. \cite{mugglestone1996practical} call a version of this function the $R$-spectrum. Alternatively, one could average also over the magnitudes to arrive at a summary function in direction. In 2D the convenient argument would be the polar angle; \cite{mugglestone1996practical} call this version the $\theta$-spectrum  and use it to study anisotropy.

\section{Isotropic spectral estimation}\label{sec:isotropy}
{%\color{red} 
In this Section we show how to form isotropic estimators and discuss their properties. Isotropic estimators are vitally important in two dimensions and higher as it becomes increasingly difficult to interpret spectra.}
 To simplify our representation, it is commonly assumed that we are analysing isotropic processes, and that only isotropic summaries need to be produced.

We therefore consider the special case of stationary and  isotropic point processes for which the complete covariance function depends only on distance, $\gamma(z)= \gamma_I(||z||)$ \cite[p.~310--311]{Daley2003}. Such radial functions \cite{grafakos2013fourier} are special as we do not require any orientational specificity in our spectral representation. The isotropy of $\gamma(z)$ transfers through the Fourier transform to the sdf so in turn $f(\freq)=f_I(||\freq||)$. We shall refer to the treatment of isotropic random fields when implementing analysis~\cite{Ponomarenko2006spectral}. An alternative is discussed by~\cite{errico1985spectra} and~\cite{Durran2017}, for random fields. 

The orientation invariance of the spectrum leads to dimensional reduction of the multidimensional Fourier transform. Three basic concepts are useful for the isotropic analysis: Recall the Bessel function of order $\besorder>-1/2$   
$$\bessel_\besorder(t) = \frac{t^\besorder}{2^\besorder\pi^{1/2}\Gamma(\besorder+1/2)}\int_0^\pi \sin^{2\besorder}(u)~\cos[t\ \cos(u)]\;du\qquad t \ge 0.$$
It is connected to the $d$-dimensional Fourier transform via \cite[Sec 3]{vembu1961fourier}
$$
%\bar{\xi}_d(2\pi t):=
\int_{\Sd}e^{-2\pi t i v\cdot u} du=2\pi~t^{-(d/2-1)}\bessel_{d/2-1}(2\pi t)\quad \forall v\in \Sd, \; t\geq 0,
$$
we can also note relationships between radial and Cartesian Fourier representations, see for example \cite{grafakos2013fourier} or \cite{estrada2014radial}.

Recall also the (1-dimensional) Hankel transform of order $m$ of a function $g:[0,\infty]\mapsto\R$ 
$$\mathcal{H}_\besorder[g](t) = \int_0^\infty g(r)\bessel_\besorder(tr)r~dr,\qquad t \ge 0.$$
Then if we assume $X$ is a stationary and isotropic process with a radial second order product density $\rho^{(2)}_I$, the sdf at wavenumber $\freq$ with  $t=\|\freq\|>0$ takes the form
\begin{align}
\nonumber
	\label{eq:isotropicsdf}
    f_I(t=\|\freq\|):=f(\freq)&= \lambda + \int_{\Rd}\left\{\rho^{(2)}(z)-\lambda^2\right\}e^{-2\pi i \freq\cdot z}\;dz\\
        &= \lambda + \int_{0}^\infty\left\{\rho_I^{(2)}(r)-\lambda^2\right\} r^{d-1} \int_{\Sd}e^{-2\pi ri~ \freq\cdot u}\;du\;dr\nonumber\\
    	&= \lambda +  2\pi\int_0^\infty \left\{\rho_I^{(2)}(r)-\lambda^2\right\} \left(\frac{r}{t}\right)^{d/2-1}\bessel_{d/2-1}(2\pi tr) r dr \nonumber\\
    	&\equiv \lambda +  2\pi~\mathcal{H}_{d/2-1}\left[\left\{\rho_I^{(2)}(r)-\lambda^2\right\} \left(\frac{r}{t}\right)^{d/2-1}\right](2\pi t),
\end{align}
i.e. the sdf is linearly related to the Hankel transform of the complete covariance function.

Finally we recall the radially averaged set covariance function,
$$\bar{\nu}_B(r) := \frac{1}{|S^{d-1}|}\int_{S^{d-1}}|B\cap B_{ru}|du , \quad r\ge 0$$
where $\nu_B(z) := |B\cap B_z|$ is the set covariance function of the set $B$ \cite[Appendix B.3]{IPSS2008}.

We start by noting the expectation of~\eqref{diggleD}; by direct calculation it follows that this takes the form:
\begin{align}
\nonumber
    \E\{ \bar{I}_{h}(\|\k\|)\} %&=\lambda+\int_{{\mathbb{R}}^{2d}}    h\left(\|x-y\| \right) \cdot \overline{J}_{d/2-1}\left(\|\k \| \cdot \|x-y\|\right) \cdot \|x-y\|^{1-d/2}\;\tilde{\rho}^{(2)}(\|x-y\|)dxdy\\
    &=\lambda+\int_{B}\int_{B\cap B_{-z}}
    h_I\left(\|z\| \right)\ \bessel_{d/2-1}\left(2\pi\|\k \| \|z\|\right)\ \|z\|^{d/2-1}\ \tilde{\rho}^{(2)}(\|z\|) \ dz dx.
\end{align}
As the change of variable to polar coordinates means we will multiply by $\|z\|^{d-1}$, and this will correspond to a $d/2-1$ Hankel transform of the taper $h(\cdot)$ multiplied by the product density. For more details on the Hankel transform see~\cite{grafakos2013fourier}.

To gain insight into how we might estimate the radial sdf besides numerically rotation-averaging the $d$-dimensional periodogram using Eq.~\eqref{angularaveragedI} we analytically rotation-average Bartlett's periodogram (Eq.~\ref{Ih}): with $t\ge 0$
\begin{align}
    \overline{I}_0(t)&:=\frac{1}{|S^{d-1}|}\int_{S^{d-1}} {I}_0(tu)\, du\\
    \nonumber
    &=\frac{1}{|B||S^{d-1}|}\sum_{x\in X\cap B} \sum_{y\in X\cap B} \int_{S^{d-1}} e^{-2\pi i (t u)\cdot (x-y)}\, du\\
    \nonumber
    &= \hat\lambda + \frac{1}{|B||\Sd|}\sum_{x,y\in X\cap B}^{\neq} \int_{S^{d-1}} e^{-2\pi t i (x-y)\cdot u}\ du\\
\label{eq:rawRA}
    &= \hat\lambda + \frac{2\pi}{|B||\Sd| t^{d/2-1}}\sum_{x,y\in X\cap B}^{\neq} \bessel_{d/2-1}(2\pi t\|x-y\|)~\|x-y\|^{-(d/2-1)},
\end{align}
which in the planar case simplifies to %special case of $d=2$ becomes
\begin{align}
\nonumber
&\stackrel{d=2}{=}~ \hat\lambda + \frac{1}{|B|}\sum_{x,y\in X\cap B}^{\neq} \bessel_{0}(2\pi t\|x-y\|).
\end{align}
%\tuomas{should we divide by surface area of $S^{d-1}$ to get an average?}
This explains the motivation behind the isotropic estimator discussed by \cite{Bartlett1964} and \cite{diggle1987nonparametric}. The aforementioned authors prefer to normalise by $N(B)$ rather than by $|B|$, corresponding to dividing \eqref{eq:rawRA} by $\hat\lambda$. 
Like the classical periodogram, the isotropic ``shortcut'' estimator of Eq.\eqref{eq:rawRA} is highly biased near $0$, as has already been observed by \cite{diggle1987nonparametric}. The biases are described in the following proposition. 

\begin{Proposition}%{Expectation of the isotropic estimator, $\bar{I}_0(t)$.}
	\label{BesselE0}
	The isotropic estimator given by \eqref{eq:rawRA} has expectation
	\begin{align}
	\nonumber
	\E\left\{ \bar{I}_0(t)\right\} =& 
%	\lambda +  2\pi\left|B\right|^{-1}\int_0^\infty \gamma(r) r^{d-1}G(\|\freq\|r) \overline{\nu}_B(r)\,dr+\lambda^2 d\cdot b_d \left|B\right|^{-1}H(\freq).\nonumber
	%\lambda + \frac{2\pi}{|B|t^{d/2-1}}\int_0^\infty \left\{\rho_I^{(2)}(r)-\lambda^2\right\} r^{d/2-1}\bar{\nu}_B(r)\bessel_{d/2-1}(2\pi tr) r dr  \\
 	%&~~ + \frac{2\pi\lambda^2}{|B|t^{d/2-1}}\int_0^\infty\bar{\nu}_B(r) r^{d/2-1}\bessel_{d/2-1}(2\pi tr) r dr 
	\lambda + 2\pi~\mathcal{H}_{d/2-1}\left[\left\{\rho_I^{(2)}(r)-\lambda^2\right\} \left(\frac{r}{t}\right)^{d/2-1}~\frac{\bar{\nu}_B(r)}{|B|}\right](2\pi t) +
	2\pi\lambda^2~\mathcal{H}_{d/2-1}\left[\left(\frac{r}{t}\right)^{d/2-1}~\frac{\bar{\nu}_B(r)}{|B|}\right](2\pi t)
	\end{align}
\end{Proposition}
\begin{proof}
	The proof follows by applying the Campbell's theorem and adding and subtracting $$2\pi\lambda^2~\mathcal{H}_{d/2-1}\left[\left(\frac{r}{t}\right)^{d/2-1}~\frac{\bar{\nu}_B(r)}{|B|}\right](2\pi t) = \lambda^2\frac{1}{|S^{d-1}||B|}\int_{S^{d-1}}T(B, tu) du.$$
\end{proof}
If we compare the formulas for the expectation of the isotropic sdf in Eq.\eqref{eq:isotropicsdf} and the expectation in Proposition \ref{BesselE0}, we recognise two sources of bias, just like with the periodogram. There is a convolution with a function $\bar{\nu}_B/|B|$, coming from the finite observation window like in the periodogram but this time radially averaged, and a centring bias term that is a radially averaged version of the $d$-dimensional bias. Diggle's truncation-to-local-minimum construction, see~\eqref{diggle4}, tries to correct for the biases, but among other things, it will not work well for a clustered process as the spectrum near zero will be underestimated.

To work towards a tapered estimator, consider rotation averaging the debiased, tapered estimator of \eqref{biascorrect}, in expectation,
\begin{align}
\label{orientav}
&\E\frac{1}{|\Sd|}\int_{\Sd}
\widetilde{I}_h(tu)\, du
\\
=&{\lambda}+\frac{2\pi}{|\Sd|t^{d/2-1}}\E\left[ \sum_{x,y\in X\cap B}^{\neq} h(x) h(y) \bessel_{d/2-1}\left(2\pi t \| x-y\| \right) \| x-y\|^{-(d/2-1)}\right] -\frac{\lambda^2}{|\Sd|}\int_{S^{d-1}}\left| H(t u)
\right|^2\,du. \nonumber 
\end{align}

The problem is that a radial function $h_I$ might not exist for which $h_I(||x-y||)=h(x)h(y)$. But if we consider data-tapering the observed differences $\{x-y:x,y\in X\cap B\}$ instead of in $B$, then we can device a taper $h_I$ in $B\oplus B = \bigcup_{x\in B}\{B+x\}$, thus we propose the following isotropic estimator:
\begin{align}
\label{eq:debiasedIsotropic}
\bar I_h(t) := \hat\lambda + \frac{2\pi}{|\Sd| t^{d/2-1}|B|}\sum_{x,y\in X\cap B}^{\neq} h_I(x-y)\bessel_{d/2-1}(2\pi t\|x-y\|)~\|x-y\|^{-(d/2-1)} - \widehat{\lambda^2}Bias_I(t).
\end{align}
The bias correction term consists of the rotation average of both the taper and the set covariance of $B$,
\begin{eqnarray*}
    Bias_I(t)&=& 2\pi\mathcal{H}_{d/2-1}\left[\left(\frac{r}{t}\right)^{d/2-1}\frac{\bar{p}_{B,h}}{|B|}\right](2\pi t), \quad \text{where}\\
    \bar{p}_{B,h}(r) &=& |S^{d-1}|^{-1}\int_{S^{d-1}}|B\cap B_{ru}|h_I(ru)du,
\end{eqnarray*}
and depends only on $t=|\freq|$. We warn the reader that the estimator of~\eqref{eq:debiasedIsotropic} is not necessarily positive, but like previous authors ~\cite{LiiMasry1994,bandyopadhyay2009asymptotic}, we sacrifice positivity in order to remove large part of the bias. Additional issues arise from estimation of the required $\lambda^2$, as the standard estimator $\hat{\lambda}^2$ is biased with a term depending on $\lambda/|B|$ and a term depending on the unknown second order product density of the underlying process. In our examples we use the slightly less biased estimator of $\widehat{\lambda^2}=N(B)[N(B)-1]/|B|^2$ which removes the first order bias. On the positive side, note that the debiased isotropic estimator without particular tapering (i.e. $h_I\equiv 1$) can be formed without any tuning parameters.

We may now ask what radial taper to use. Tapers are normally designed for processes observed in discrete time or space. We have discussed using separable tapers, and now note that the choice of the taper needs to match the observational domain. For example it is difficult to match a square observational domain with a radial taper, unless one chooses a compact taper inscribed by the box. There are therefore three considerations 1) demanding the taper be separable in space, 2) isotropic in space, or 3) exactly compact in space. Discrete and isotropic tapers have been been determined numerically by~\cite{simons2011spatiospectral}. This will not be an option for us as we need to evaluate the tapers at random locations.

Slepian and co-authors have studied the design of tapers in arbitrary dimensions~\cite{slepian1964prolate}. Often solutions are in terms of prolate spheroidal wavefunctions. Their discrete analogue is the prolate spheroidal wave sequences, which are similar to the set of Hermite function~\cite{xu1999multiple}. However, the isotropic estimator can be seen as an estimator for the Fourier transform of the (non-stationary) first moment of the difference process $Z=\{z=x-y:x,y\in X\cap B\}$, with the connection $\rho^{(2)}(z) = \lambda_Z(z)/|B\cap B_z|$.

Assuming the observation window is a cuboid $B=B_\square(\bm{l})$, %=\prod_{j=1}^d[0,\bside_j]$ \tuomas{[replace by $B_\square$]}, 
then the difference vector observation window is $B\oplus B=\prod_{j=1}^d[-\bside_j,\bside_j]$, and we can create the $d$-dimensional taper related to the Hermite functions as the product of the squared exponentials 
\begin{equation}
    h_I(z) = \prod_{j=1}^d e^{-az_j^2/(2\bside_j)^2}, \quad z\in\Rd
\end{equation}
We see that if $\bside_j=\bside$ for all $j$, then $h_I$ becomes radial. This is related to circular harmonic decompositions, see~\cite{jacovitti2000multiresolution}. The authors of~\cite{xu1999multiple} recommend using a scaling factor which we have adjusted for a sampling region of $z_j\in (-\bside_j,\bside_j)$, but which they match to the spheroidal wave functions. We propose setting $a=25$, as this numerically seems to not down--weight too much data, or have a too large jump near the border of $B\oplus B$.

In a sense our above construction is a ``fix'' as it only works for the 0th Hermite function, and using more than one taper will even if $\bside_j=\bside$ for all $j$ not lead to isotropic functions. This is not unsurprising as the sampling domain will leave an imprint, and capturing this set of information requires using more than the first taper, see analogous discussion in 1D~\cite[Ch.~6]{Percival1993}.
Also, if we have a spatial sampling that is highly elongated rectangle rather than a square then another solution needs to be chosen.   
Finally we may ask ourselves, what happens if we have a point process $X$ which is anisotropic but we still evaluate~\eqref{diggle} in 2D for the pattern? Because it incorporates angular averaging we expect the RHS of the expression to only depend on $\|\freq\|$ as it does.

\begin{Proposition}{Expectation of Diggle's 2D Estimator II.}\label{BesselE}
The large area expectation of Diggle's estimator given in~\eqref{diggle} when $B$ is a cuboidal box $B_\square(\bm{l})$ and when the spectrum is not necessarily isotropic takes the form in terms of $\widetilde{f}(\freq)$ as defined in~\eqref{eqn:fourierexcess} of 
\begin{align}
    \E\{ I_D(\freq)\}&=\lambda+\lambda^2\int_{{\mathbb{R}}^d}
    \widetilde{f}(\freq')\frac{\delta(\|\freq\|-\|\freq'\|)}{\|\freq\|}d\freq'
+\int_{{\mathbb{R}}^d}\frac{\delta(\|\freq\|-\|\freq'\|)}{\|\freq\|}
|B|^{-1}T(B,w-\freq')d\freq'+o(1),
\end{align}
where $T(B,\freq)$ is defined in Lemma~\ref{lemma1}.
\end{Proposition}

\begin{proof}
The proof is provided in Appendix~\ref{proofBesselE}.
\end{proof}
This shows what happens when we calculate isotropic summaries of quantities that are not isotropic per se. {Thus the two propositions determine what expectation the analytically orientationally averaged periodogram and the second shows us how Diggle's estimator mixes orientational information up to produce an isotropic estimator.}

\section{Simulations}\label{sec:simulations}

To study the behaviour of the estimators we have discussed, and additionally some of the debiasing results, we conducted an extensive simulation study.

\subsection{Details for executing the simulations study}

We simulated three stationary models in 2D to represent the archetypal spatial patterns of regularity, complete randomness and clustering. Throughout the simulations the intensity, i.e. expected number of points per unit area, was kept constant at $\lambda\equiv 0.01$. The processes were simulated in square observation windows $B_n=[-\bside_n/2, \bside_n/2]^2$ of increasing area, such that on average we observed $\lambda|B_n| = \lambda \bside_{n}^2 = n = 25, 50, 100, 200, 400$ or $800$ points. For example, $\bside_{n=100}=100$ and $\bside_{n=400}=200$ spatial units. We chose the regime both to explore the often encountered small data scenarios and to ascertain asymptotic behaviour. The models below are tuned so that they produce patterns with visually distinct structures at $\bside_{n=100}$.

Complete randomness i.e. the Poisson process has no adjustable parameters after $\lambda$ was fixed. Spatial regularity is represented by the Mat\'ern type II process \cite[Section 6.5.2]{IPSS2008} in two variants, one with hard-core radius $2$ (variant ``r2'') and one with radius $5$ (variant ``r5'') which correspond to about 40\% and 90\% of the maximum allowed radius for the process given $\lambda$, respectively. The clustering behaviour is represented by the Thomas process \cite[Section 6.3.2]{IPSS2008} with cluster intensity $\kappa$ and Gaussian dispersal kernel standard deviation $\sigma$ and again in two variants, one with many small or tight clusters ($\kappa = 0.6\lambda, \sigma=2$, variant ``MS'') and one with a few large clusters ($\kappa=0.3\lambda, \sigma=6$, variant ``FL''). The per-cluster expected point count $\mu$ was then fixed via the property $\mu\kappa=\lambda$ to $1.67$ and $3.33$, respectively. The theoretical pair correlation functions and spectral density functions of the models are depicted in Figure \ref{fig:fig-1}. Example patterns of each model and observation window combination are given in Supplement Figure \ref{figS:patterns}.

For the estimation of the 2D spectra we fixed the wavenumber grid to $\freq\in[-0.3, 0.3]^2$ with 101 equidistant steps in each dimension, giving the step~size~0.006 in both dimensions. This scale was chosen to cover the interesting range of non-constant values for all models (cf. Figure \ref{fig:fig-1}). To reduce the effect of high wavenumber noise only the values on the sub-grid $\freq\in[-0.2,0.2]^2$ were considered when integrating over $\freq$ for the quality metrics discussed below. When a rotationally averaged curve was to be computed using Eq. \eqref{angularaveragedI}, the averaging was done over the magnitude-grid $\|\freq\|=0.003, 0.006,...,0.300$ using a box kernel and if not otherwise stated, radius of $1.25\cdot 0.006$. The radial Hermite taper parameter was fixed to $a=25$.

We summarised the quality of each estimator $I_.=I_.(\freq;\x)$ under all combinations of a model $\mathcal{M}$ and observation window $B_n$, say $\mathcal{M}_n$, with an integrated summary.
First, we estimated per-wavenumber variance $V(\freq;\mathcal{M}_n)=\ \text{Var}\left[I_.(\freq;\x) | \x \sim  \mathcal{M}_n\right]$, bias $Bias(\freq; \mathcal{M}_n)=\ \E \left[I_.(\freq;\x) - f_{M}(\freq)|\x\sim \mathcal{M}_n\right]$ and mean square error $MSE(\freq;\mathcal{M}_n)\ =\ V(\freq;\mathcal{M}_n) + Bias^2(\freq;\mathcal{M}_n)$. The $f_{\mathcal{M}}$ stand for the theoretical sdf of the model $\mathcal{M}$. 
Then these were summarised further to $iVar(\mathcal{M}_n)\ =\sum_\freq V(\freq;\mathcal{M}_n)$, the integrated squared bias $iBias^2(\mathcal{M}_n)=\sum_\freq Bias^2(\freq;\mathcal{M}_n)$ and $iMSE(\mathcal{M}_n) = iVar(\mathcal{M}_n) + iBias^2(\mathcal{M}_n)$. 
Smaller $iVar, iBias^2$ and $iMSE$ indicate better quality. The quantities were estimated from 1000 simulations of every $\mathcal{M}_n$.

The estimator for a periodogram is given in equation \eqref{biascorrect} with the debiasing term included. Bartlett's periodogram (``periodogram'') has the taper $h_0(x)=|B_n|^{-\frac{1}{2}}\I_{B_n}(x)$ for window $B_n$. 
{For the multitapering (``mt $\ntapers$''), defined in equation \eqref{tapered} with $\ntapers=\nsubtapers^2$ tapers each having a parameter $\itapers~=~(\subtapers_1, \subtapers_2)~\in~\{1,...,\nsubtapers\}^2$, we used the orthogonal sine-tapers  $h_\itapers(x)=\I_{B_n}(x)\prod_{j=1}^2\sin[ \pi \subtapers_j (x_j+l_j/2)/l_j]$. 
For the kernel smoothed estimators (``smoothed $b$''), we first compute the Bartlett's periodogram and then, given the estimate $I_0$ on the wavenumber gird, we convolute it with a $b\times b$ discrete template having approximately Gaussian weights.}

All computations were done using the R-software \cite[v.3.6.3]{Rsoftware}. Simulations were done with the help the R-package \texttt{spatstat} \cite[v.1.64-1]{spatstat}. Each of the discussed estimators were programmed into an R-package \texttt{ppspectral}, available on request from the first author.
%\david{Can we provide library versions used? I find updates invariably mean things are never the same making it hard to reproduce results -maybe the github is sufficient, although this wont give the library version number used?}

\begin{table}[ht]
%	\footnotesize
	\centering
\caption{The fraction of bias removed by the proposed bias correction. Maximum is 1.00.}
	\begin{tabular}{llrrrrrr}
		\hline
		Model & Estimator & $n=25$ & 50 & 100 & 200 & 400 & 800 \\ 
		\hline
		Mat\'ernII r5 & periodogram & 1.00 & 1.00 & 1.00 & 0.98 & 0.99 & 1.00 \\ 
		Mat\'ernII r5 & mt $\ntapers=3^2$ & 0.99 & 1.00 & 1.00 & 1.00 & 1.00 & 1.00 \\ 
		Mat\'ernII r2 & periodogram & 1.00 & 1.00 & 1.00 & 0.99 & 0.99 & 1.00 \\ 
		Mat\'ernII r2 & mt $\ntapers=3^2$ & 1.00 & 1.00 & 1.00 & 1.00 & 1.00 & 1.00 \\ 
		Poisson & periodogram & 1.00 & 1.00 & 1.00 & 0.98 & 0.99 & 1.00 \\ 
		Poisson & mt $\ntapers=3^2$ & 1.00 & 1.00 & 1.00 & 1.00 & 1.00 & 1.00 \\ 
		Thomas FL & periodogram & 0.93 & 0.98 & 0.99 & 0.97 & 0.99 & 1.00 \\ 
		Thomas FL & mt $\ntapers=3^2$ & 0.35 & 0.86 & 0.98 & 1.00 & 1.00 & 1.00 \\ 
		Thomas MS & periodogram & 0.98 & 0.99 & 1.00 & 0.97 & 0.99 & 1.00 \\ 
		Thomas MS & mt $\ntapers=3^2$ & 0.89 & 0.98 & 1.00 & 1.00 & 1.00 & 1.00 \\ 
		\hline
	\end{tabular}

\label{tab:bias}
\end{table}

\subsection{The effect of debiasing}

As we saw in Figure \ref{fig:dataexample}, the debiasing term has a large effect on the quality of the estimates. Table \ref{tab:bias} provides the fractions of bias removed by the debiasing term for the Bartlett's periodogram and $\ntapers=3^2$ multitapered periodogram. Nearly all bias is removed, with the only notable exception being the 65\% bias left in the multitapered periodogram when data is an arguably very small sample of Thomas FL variant. 

%\david{Formatting (position) of table 1 title needs adjusting -moving down slightly}\tuomas{moved the titles on on top of the tables}

We illustrate the form of the biases in Figure \ref{fig:bias2d} for two model variants. The centring bias is always positive, as expected from equation \eqref{expectation}. 
Bartlett's periodogram bias (cf. Lemma~\ref{lemma1}) is concentrated along the axis where one of the sinc functions is constantly 1 and near the origin where the second sinc function grows to 1. 
The multitapered periodogram has a square-shape around the origin.

\begin{figure}[!ht]
	\centering
	\includegraphics[width=.9\linewidth]{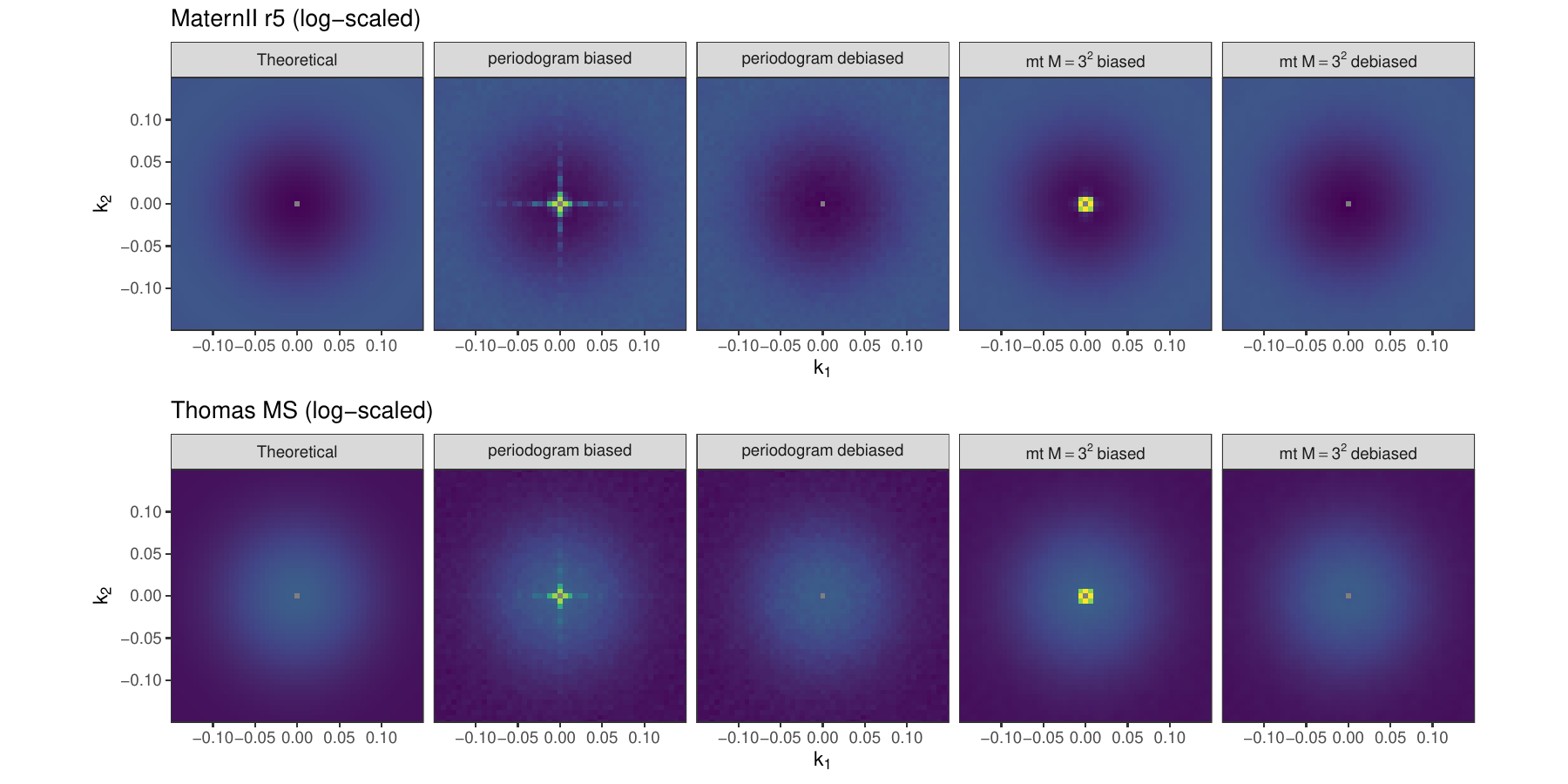}
	\caption{The effect of bias removal for two of the models described in the text. Theoretical sdf on the left, and then towards the right, estimated mean of: Bartlett's periodogram, biased and debiased; Multitapered $\ntapers=3^2$, biased and debiased. Estimated with sample sizes $n=800$.}
	\label{fig:bias2d}
\end{figure}

\subsection{Overall quality of the debiased estimators}

Figure \ref{fig:Quality} summarises the estimated integrated variance, squared bias and MSE of the estimators with various tuning parameters. Only the debiased estimators are included, and the values are given in $\log_{10}$-scale and relative to the baseline given by Bartlett's periodogram.

Multitaper with $\ntapers=1^2$ (i.e. only one taper) does not differ in iMSE from the periodogram as the benefit of averaging does not apply. With more than one taper the variance is reduced, e.g. $\ntapers=3^2$ it goes down by 90\%, and the bias stays at baseline, but adding more and more tapers ($\ntapers=6^2$) starts to introduce bias especially for low sample sizes. The bias is more prominent for the clustered cases, and based on the earlier discussion this is due to oversmoothing at small wavenumbers. Both the variance reduction and the oversmoothing happens with the post-hoc smoothing as well.
From the results we can confirm that some smoothing is beneficial but again a balance must be struck to avoid oversmoothing (cf. Section \ref{sec:linear}). The regular process is more resilient to oversmoothing with reasonable sample sizes.

\begin{figure}[htp]
	\centering
	\includegraphics[width=1\linewidth]{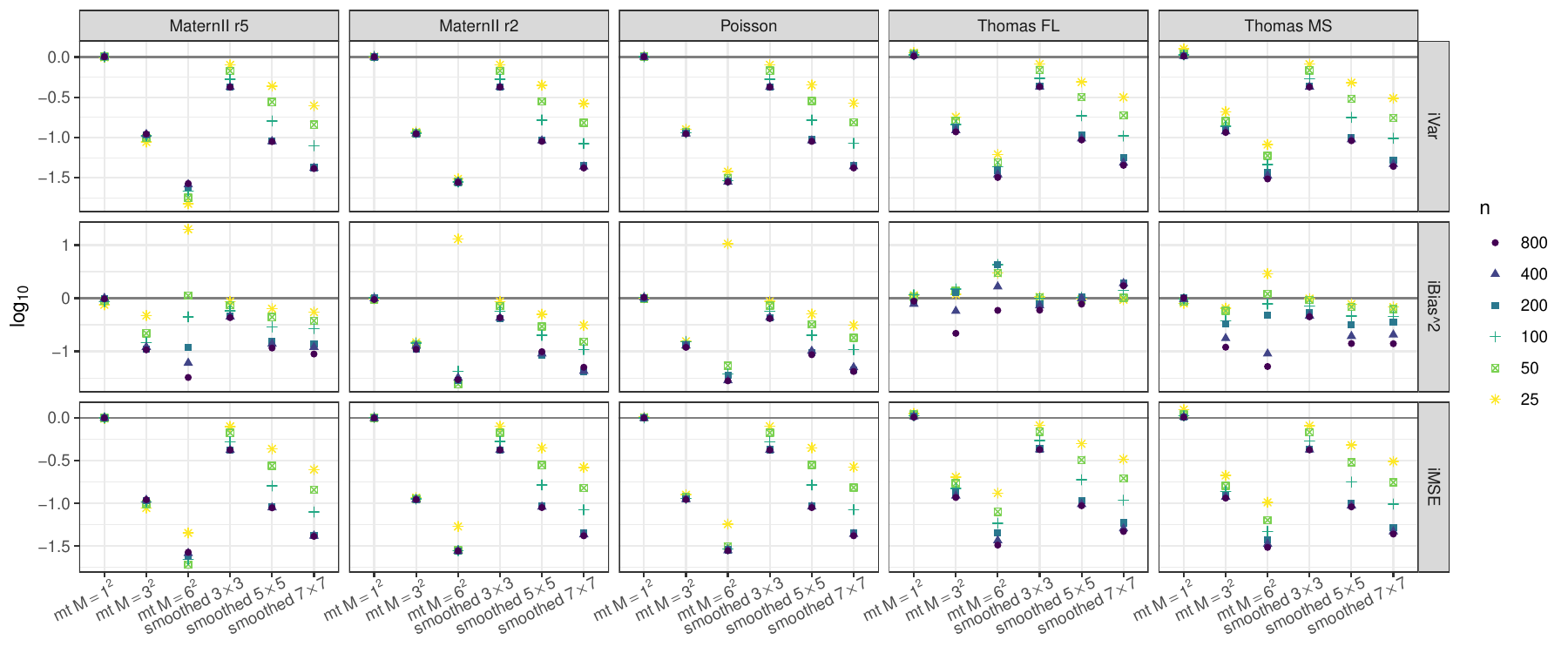}
	\caption{Relative integrated variance (top row), squared integrated bias (middle), and integrate mean square error (bottom) for different models and estimators of spectral density at various levels of sample size $n$. The values are relative to debiased Bartlett's periodogram, which has the value 0 in the presented $\log_{10}$-scale. For example, a value of -1 means the error is of order 0.1 relative to the periodogram.}
	\label{fig:Quality}
\end{figure}

\subsection{Rotational averaging and isotropic estimation}

Isotropic estimation, either by averaging radially or by the isotropic shortcuts discussed in Section 7, is an important dimension reduction step for easier visual inspection of the periodogram (particularly for 3D data). The kernel-based rotation averaging depends on a bandwidth parameter, and we studied how this affects the quality metrics. We set the bandwidth radius to $BF\cdot 0.006$ and varied the bandwidth factor $BF$ from 1 to 10. Figure \ref{fig:rotave-bandwidth} provides the $iBias^2$ and $iMSE$ for the debiased Bartlett's periodogram when comparing the rotationally averaged estimates to the true 1D isotropic sdfs (cf. Figure \ref{fig:fig-1}). 

The Poisson process has a constant sdf so oversmoothing will not be penalised. For the clustered processes even a small bandwidth introduces bias, but this is typically more than offset by the reduction in variance. The error for the few large clusters variant is more sensitive to changes in the bandwidth than the many small clusters. This is because the effective range of wavenumbers where the sdf exhibits structure is more concentrated, and thus easier to smooth out, than the few large cluster variant (cf. Figure \ref{fig:fig-1}). 

The strongly regular variant ``r5'' is similar to the clustered cases in that an optimal bandwidth is clearly present. For the less regular ``r2'' variant the error is like that of the Poisson process, possibly due to the target sdf being a weak oscillation around a constant and this being hard  to detect. If we look at the integrated squared bias we however do see that with large enough data oversmoothing is detectable.

We note that since multitapered periodograms and locally averaged periodograms are already smoothed, their rotational averaged estimates exhibit smoothing bias at smaller bandwidth factors (results not worth showing here).

\begin{figure}[htp]
	\centering
	\includegraphics[width=1\linewidth]{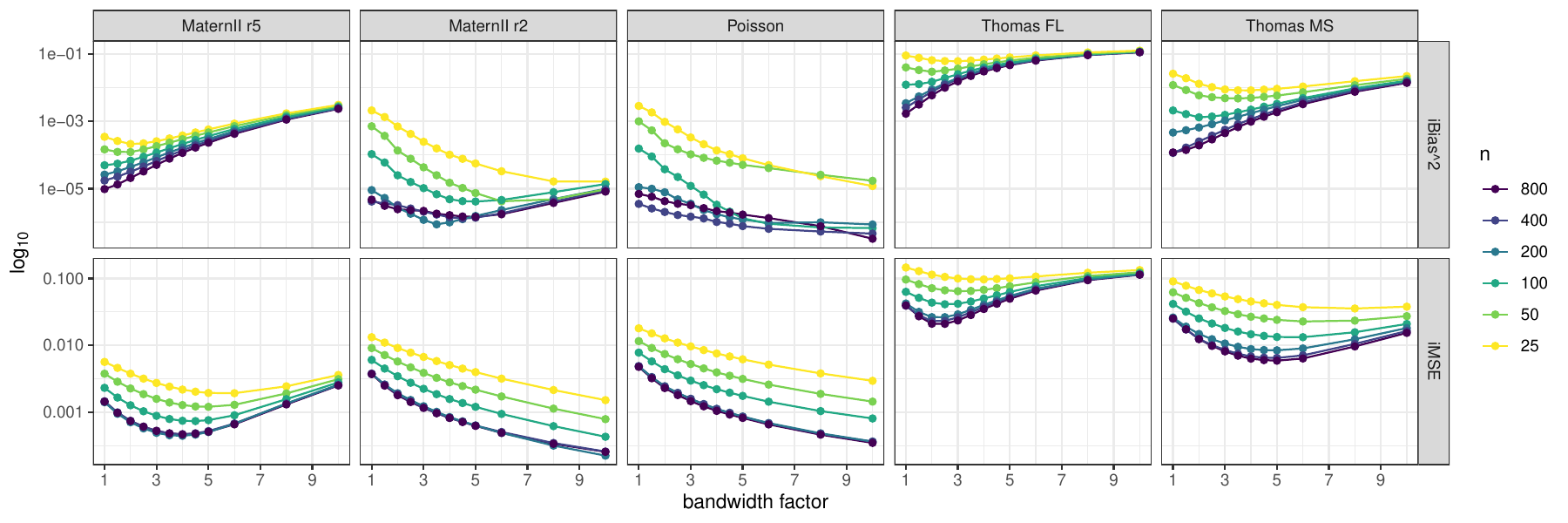}
	\caption{Integrated squared bias (top) and integrated MSE of the rotationally averaged Bartlett's periodogram when compared to the true 1D isotropic sdf. The rotational averaging bandwidth was $BF\cdot 0.006$, where the bandwidth factor $BF$ is on the x-axis. The models are described in the text.}
	\label{fig:rotave-bandwidth}
\end{figure}

\begin{table}[htp]
	\centering
	%\begin{tabular}{l|rrr|rrr|rrr}
 % \caption{}
\caption{Scores $iMSE,iBias$ and $iVar$ for isotropic estimators, relative to radially averaged debiased periodogram with optimised smoothing bandwidth.}
	\begin{tabular}{l|rrr|rrr}
		\hline
% 		& \multicolumn{3}{c}{No taper} &  \multicolumn{3}{c}{Hermitian taper} & \multicolumn{3}{c}{Fry tapering} \\ 
% 				\hline
% 	    Model    & iMSE & iBias$^2$ & iVar & iMSE & iBias$^2$ & iVar & iMSE & iBias$^2$ & iVar\\ 
% 	\hline
% Mat\'ernII r5 & 6.12 & 0.21 & 8.57 & 13.56 & 0.05 & 19.15 & 2.04 & 0.10 & 2.84 \\ 
%   Mat\'ernII r2 & 22.40 & 3.52 & 23.07 & 52.24 & 0.99 & 54.05 & 8.09 & 0.31 & 8.36 \\ 
%   Poisson & 15.72 & 10.53 & 15.89 & 47.09 & 28.85 & 47.11 & 7.17 & 5.26 & 7.25 \\ 
%   Thomas FL & 1.22 & 0.04 & 1.81 & 4.71 & 0.00 & 7.06 & 0.59 & 0.05 & 0.87 \\ 
%   Thomas MS & 4.34 & 0.01 & 6.50 & 15.01 & 0.04 & 22.48 & 2.20 & 0.01 & 3.28 \\ 
		& \multicolumn{3}{c}{No taper} &  \multicolumn{3}{c}{Squared-exponential taper} \\ 
				\hline
	    Model    & iMSE & iBias$^2$ & iVar  & iMSE & iBias$^2$ & iVar\\ 
	\hline
Mat\'ernII r5 & 6.12 & 0.21 & 8.57 & 3.51 & 0.29 & 4.90 \\ 
  Mat\'ernII r2 & 22.40 & 3.52 & 23.07 & 15.59 & 1.09 & 16.12 \\ 
  Poisson & 15.72 & 10.53 & 15.89 & 13.76 & 5.75 & 13.77 \\ 
  Thomas FL & 1.22 & 0.04 & 1.81 & 1.32 & 0.14 & 1.97 \\ 
  Thomas MS & 4.34 & 0.01 & 6.50 & 3.58 & 0.04 & 5.35 \\
\hline
	\end{tabular}
\label{tab:isotropic}
\end{table}

To compare the radial averaging to direct isotropic estimation, we estimated the debiased isotropic periodogram and the squared-exponential tapered isotropic periodogram with $a=25$. Table \ref{tab:isotropic} provides the relative quality scores for the isotropic and the tapered isotropic periodograms relative to radially averaged periodogram with the best bandwidth for each model and data size, cf. Figure \ref{fig:rotave-bandwidth}. The values are medians over the sample sizes $n\ge 50$.

In iMSE sense the radially averaged periodogram performs the best, chiefly because iMSE is of order iVar and the radial averaging greatly reduces the variance. However, in terms of bias the isotropic estimator provides more accuracy, except for Poisson when no oversmoothing is possible. Adding a smooth taper to the isotropic estimator further reduces the bias, but can increase the variance, likely due to some data near the edges is being filtered out. The bias occurs near $\|\freq\|=0$, mostly when $\|\freq\| < 1/(\sqrt{2}l_n)$ (illustration in Supplements Figure \ref{figS:rotave-v-isotropic}).

%\newpage
%\section{Simulation Studies and Data Analysis}

%%%%%%%%%%%%%%%%
%\begin{figure}
%\centering{\includegraphics[width=0.95\textwidth]{exp3_run2_rotmeans_all.pdf}}
%\caption{\label{fig1} Some appropriate text.}
%\end{figure}
%%%%%%%%%%%%%%

%%%%%%%%%%%%%%%%
%\begin{figure}
%\centering{\includegraphics[width=0.95\textwidth]{exp3_run2_metrics_rel_to_raw.pdf}}
%\caption{\label{fig2} Some appropriate text.}
%\end{figure}

%%%%%%%%%%%%%%

\section{Spectral properties of the Barro Colorado Island data}\label{sec:BCI}

The Barro Colorado Island (BCI) rainforest data set \cite{condit2019complete} contains a census of tree and shrub species from a 1000m by 500m region on Barro Colorado Island, Panama.
We illustrate the techniques developed in this paper on the point patterns corresponding from three species contained in the BCI data, using the 2015 census, using only individuals that are alive.

\begin{figure}[ht]
    \centering
    \includegraphics[width=.9\linewidth]{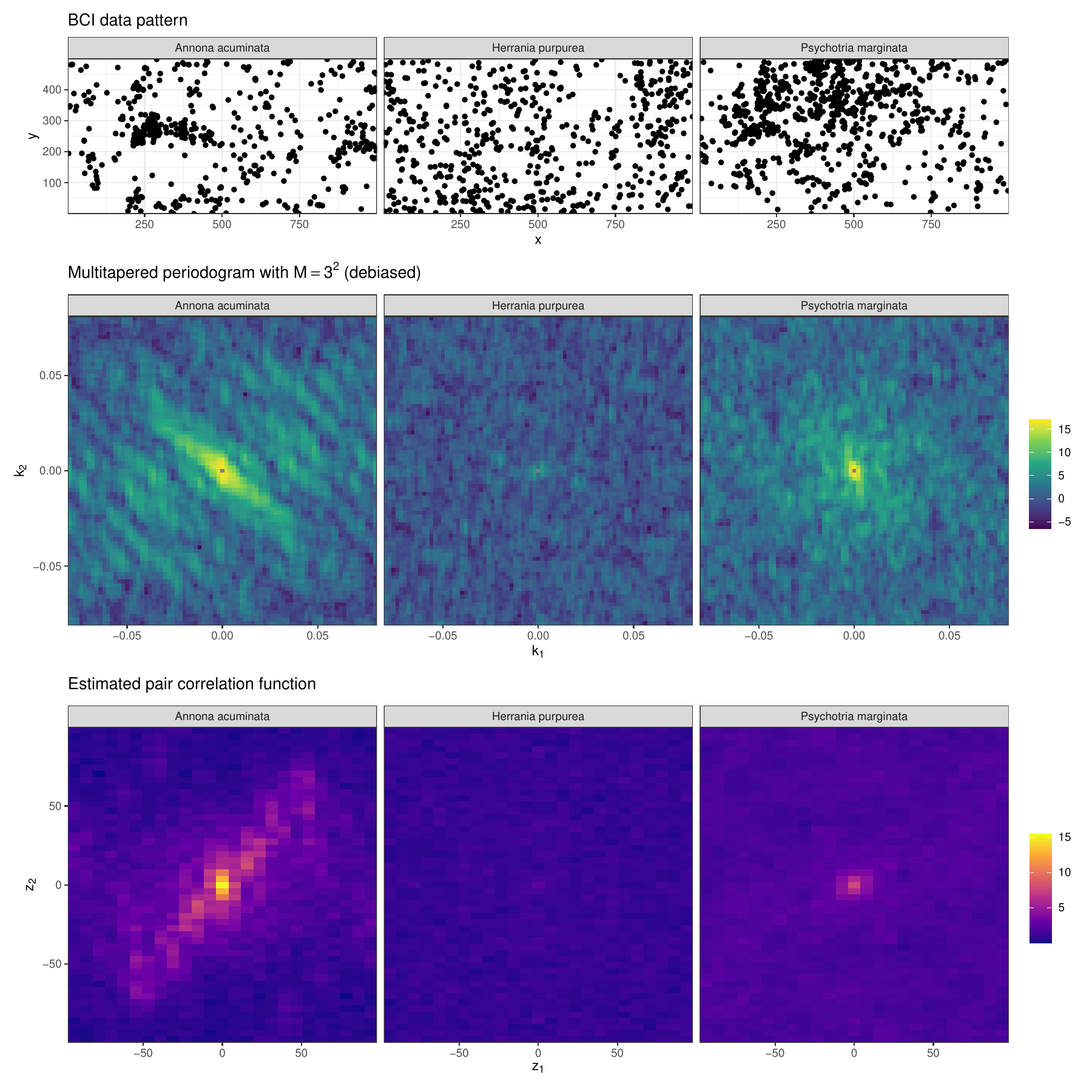}
    \caption{\emph{Top row:} Point patterns of three different species from the BCI data.
    \emph{Middle row:} Multitaper spectral estimates from the above point patterns, rescaled by the intensity and plotted on a decibel scale, i.e.\ $10\log_{10}(\bar{I}^t_M(\freq)/\hat\lambda)$ with $M=3^2$. 
    \emph{Bottom row:} Estimated pair correlation function for the three species, i.e.\ $\widehat\rho^{(2)}(z)/\hat\lambda^2$.}
    \label{fig:bci_example}
\end{figure}

Figure~\ref{fig:bci_example} shows an example analysis for three species from the BCI data, namely \textit{Annona acuminata}, \textit{Herrania purpurea}, and \textit{Psychotria marginata}. 
The top row shows the point patterns for each of the three species, the middle row shows estimates of their spectral density functions computed using multitapering, and the bottom row shows estimates of the pair correlation function for each species.
It is clear that there are substantial differences in the structures present in the spectral density functions of these species.
Firstly, in the case of \textit{Annona acuminata} (left), we see clearly anisotropic behaviour, similar to the anisotropic clustering shown in Figure~\ref{fig:dataexample}. This is to be expected as \textit{A. acuminata} is usually associated with very moist soils and has previously been associated with the stream (a linear structure) running through the plot \cite{harms2001habitat,flugge2014subcommunities}.
Secondly, in the case of \textit{Herrania purpurea} (middle) we see no structure at all, which is reassuring as this suggests that the features seen in the other two species are signal and not just noise. Again, this makes ecological sense as \textit{H. purpurea} is a lower canopy tree that has seeds widely dispersed by monkeys and other animals.
Finally for \textit{Psychotria marginata} (right) we see more isotropic behaviour, with information present at low wavenumbers until $\|k\|\approx 0.005$, corresponding to correlation on scales of length around 200. Additionally, there is further information until $\|k\|\approx 1/15$, which corresponds to the peak seen in the pair correlation function until scales around 15. This clustering is likely caused by seeds being relatively poorly dispersed \cite{theim2014spatial}, but also because \textit{P. marginata} is associated with moist soils leading to a build up of large clusters of trees in the low plateau and swamp regions of the plot in the upper left quadrant of the BCI plot \cite{harms2001habitat,flugge2014subcommunities}.

\section{Conclusions}\label{sec:conclusions}
Spectral analysis and Fourier features are classical and important tools for characterizing and understanding 
time series and random fields~\cite{Percival1993,Diggle2013},  in for example estimation and detection problems. 
When first introduced, the fundamental theory of spectral analysis for random fields and time series corresponded to understanding what Fourier transform to use, and determining the first and second moments of that Fourier transform for stationary time series and random fields~\cite{Robinson1982}. 
This led researchers to develop the by now well--established theory of spectral analysis (again for time series and random field) from which a large and sophisticated theory of spectral analysis was built. This set of methods has more recently been discovered in machine learning~\cite{lazaro2010sparse,ton2018spatial}. 
The corresponding theory for spatial point processes has been neglected apart from some notable and not very recent exceptions~\cite{Bartlett1963,Bartlett1964,mugglestone1996practical}. For point processes machine learning researchers have just started to discover the utility of Fourier-based methods~\cite{ilhan2020modeling,john2018large}.
Current state of the art for the spectral analysis of point processes is that we really do not even know what to compute, and know even less how to address its digital implementation.

This article has addressed this first outstanding problem. We have calculated the expectation of natural method of moments estimators. We have imported state of the art ideas from signal processing as part of this process, and have addressed how to taper for variance reduction in this setting. Introducing tapering in this setting required us to determine how to taper, where we chose to use continuous space tapers, not to interpolate as had been done by previous authors.
We showed that in the setting of Gaussian convergence, multitaper estimates of the spectral density have a decreased variance. Our simulation studies verify our theoretical results, and the general utility of tapering outside the most restrictive class where we can prove everything. 

This manuscript does not seek to address how to bring spectral analysis of point processes into the 21st century, and reproduce all results available for time series and random fields. Instead we have taken an important first step by adapting existing signal processing tools for time series and random fields and shown how they can be modified to help us estimate the spectral content of a spatial point process. This paves the way forward to develop more sophisticated theory which can be applied to complex spatial point pattern datasets and more sophisticated algorithms.

\section*{Acknowledgment}
%\section{Acknowledgements}
This work was supported in part by
by the UK Engineering and Physical Sciences Research Council
under Mathematical Sciences Leadership Fellowship EP/I005250/1, Developing Leaders Award EP/L001519/1, and Award EP/N007336/1;  the European Research Council under Grant CoG
2015-682172NETS within the Seventh European Union Framework Program; and Academy of Finland (project number 327211). 

\appendix

\section{Proof of Lemma~\ref{lemma1}}
\label{sec:Pf1}
\begin{proof}
We start the proof by noting that (using that $B$ is centred)
\begin{align}
\label{expPerio}
\E I_0(\freq) &= \lambda + |B|^{-1}\int_B\int_{B_{-x}}e^{-2\pi i\freq\cdot z}\rtwo(z)dz\\
\nonumber
&= \lambda + |B|^{-1}\int_B\int_{B_{-x}}e^{-2\pi i\freq\cdot z}[\kap{z}]dz dx + |B|^{-1}\lambda^2 \int_B\int_{B_{-x}}e^{-2\pi i\freq\cdot z}dz \\
\nonumber
&= %\overbrace{c|B| \lambda}^{%\mathcal{I}_0^{(1)}(\omega)}
 \lambda + |B|^{-1}\int_{\R^d}|B\cap B_{-z}|e^{-2\pi i\freq\cdot z}[\kap{z}]dz + |B|^{-1}\lambda^2 T(B, \freq)\\
 &=\lambda+\int_{{\mathbb{R}}^d}G^{(1)\ast}(\freq')
 \widetilde{G}^{(2)}(\freq'-\freq)\,d\freq'+|B|^{-1}\lambda^2 T(B, \freq),
 \nonumber
\end{align}
by the convolution theorem, after we have defined the functions:
\begin{align*}
\nonumber
G^{(1)}(\freq)&=|B|^{-1}\int_{R^d}  | B\cap B_{-z}|e^{-2\pi i\freq\cdot z}\,dz
=|B|^{-1}T(B,\freq)\\
\widetilde{G}^{(2)}(\freq)&=\int_{\R^d}e^{-2\pi i\freq\cdot z}[\kap{z}]dz=\lambda^2\widetilde{f}(\freq).
\end{align*}
Note that from Gradshteyn 3.741 with the change of variables $x_j=\pi \freq_j \bside_j$
that we get (recalling $|B|=\prod_j \bside_j$):
\begin{align}
\nonumber
\int_{{\mathbb{R}}^d} T(B,\freq)\;d\freq&=\int_{{\mathbb{R}}^d} 
\prod_j\frac{\sin^2(\pi \freq_j\bside_j)}{(\pi \freq_j)^2}d\freq_j\\
\nonumber
&=\int_{{\mathbb{R}}^d} 
\prod_j \bside_j^2 \frac{\sin^2(x_j)}{x_j^2}( dx_j/(\pi\bside_j))\\
&= |B|
\prod_j ( \pi/\pi)=  |B|.
\label{normT}
\end{align}

We have defined
\begin{equation}
T(B,\freq)=\prod_{j=1}^d \frac{\sin^2(\pi \freq_j\bside_j)}{(\pi \freq_j)^2},\quad \freq\in {\mathbb{R}}^d.
\end{equation}
We now go back to our definition of $\E I_0^{(2)}(\freq)$, and explore this more carefully. We also fix $0<\beta<1$, and expand the expectation as follows (realising that care needs to be taken around the pole):
\begin{align}
\nonumber
\E I_0^{(2)}(\freq)
&=\lambda^2 \int_{R^d}G^{(1)\ast}(\freq-\freq') \widetilde{G}^{(2)}(\freq') d\freq'\\
\nonumber
&=\lambda^2\prod_{j=1}^d \left(\int_{-\infty}^{\freq_j-\frac{1}{\bside_j^{1-\beta}}}+ 
\int_{\freq_j-\frac{1}{\bside_j^{1-\beta}}}^{\freq_j+\frac{1}{\bside_j^{1-\beta}}}+
\int_{\freq_j+\frac{1}{\bside_j^{1-\beta}}}^\infty\right) G^{(1)\ast}(\freq-\freq') \widetilde{G}^{(2)}(\freq') d\freq'\\
&=\E I_{1\dots 1}^{(2)}(\freq)+\E I_{12\dots}^{(2)}(\freq)+\dots +\E I_{d\cdots d}^{(2)}(\freq).
\end{align}
We can note directly, that taking a supremum over $G^{(1)\ast}(\freq-\freq') $ in the range of integration, and assuming the integral of $G^{(2)}(\freq)$ is finite (which is true from Parseval-Rayleigh relationships
\begin{equation}
\E I_{11}^{(2)}(\freq)={\cal O}(|B|^{-1}\prod_{j=1}^d\bside_j^{2-2\beta})={\cal O}(\prod_{j=1}^d \bside_j^{1-2\beta}),
\end{equation}
which requires $\beta>1/2$ to be $o(1)$. The same holds for the other elements apart from
$\E I_{2\dots 2}^{(2)}(\freq)$.
We therefore have
\begin{align}
\label{expsecond}
\E I_0^{(2)}(\freq)&=\lambda^2 \prod_{j=1}^d  
\int_{\freq_j-\frac{1}{\bside_j^{1-\beta}}}^{\freq_j+\frac{1}{\bside_j^{1-\beta}}} G^{(1)\ast}(\freq-\freq') \widetilde{G}^{(2)}(\freq') d\freq'
\left\{1+o(1)\right\}\\
\nonumber
&\equiv {\cal I}_0^{(2)}(\freq)
\left\{1+o(1)\right\},
\end{align}
so that
\[{\cal I}_0^{(2)}(\freq)\equiv \lambda^2 \prod_{j=1}^d  
\int_{\freq_j-\frac{1}{\bside_j^{1-\beta}}}^{\freq_j+\frac{1}{\bside_j^{1-\beta}}} G^{(1)\ast}(\freq-\freq') \widetilde{G}^{(2)}(\freq') d\freq'.\]

We assume that $\tilde{f}(\freq)$ is twice differentiable at $\freq$ (everywhere possibly but $\freq=0$) so that we get for all $\freq\neq 0$ and for $\|\freq-\freq''\|<\|\freq-\freq'\|$
\begin{equation}
\label{tildetaylor}
\widetilde{f}(\freq')=\widetilde{f}(\freq)+\nabla \widetilde{f}(\freq)^T\left(\freq'-\freq\right)+
\frac{1}{2}\left(\freq'-\freq\right)^T\tilde{\mathbf{H}}_f(\freq'')\left(\freq'-\freq\right),
\end{equation}
defining $\nabla \widetilde{f}(\freq)$ as the gradient and $\tilde{\mathbf{H}}_f(\freq)$ as the Hessian matrix.

%We can then substitute this back in and need to calculate moments of $G^{(1)}(\freq)$ to understand the bias. 
Substituting this expansion back into~\eqref{expsecond} we get 
with a change of variables 
%\tuomas{[obsolete x notation?]}:
\begin{align}
\nonumber
{\cal I}_0^{(2)}(\freq)&=\lambda^2 \int_{\freq_1-\frac{1}{\bside_1^{1-\beta}}}^{\freq_1+\frac{1}{\bside_1^{1-\beta}}}\dots \int_{\freq_p-\frac{1}{\bside_p^{1-\beta}}}^{\freq_p+\frac{1}{\bside_p^{1-\beta}}}G^{(1)\ast}(\freq-\freq') 
\left\{\widetilde{f}(\freq)+\nabla \widetilde{f}(\freq)^T\left(\freq'-\freq\right)+
\frac{1}{2}\left(\freq'-\freq\right)^T\right.\\
\nonumber
&\left.\tilde{\mathbf{H}}_f(\freq'')\left(\freq'-\freq\right) \right\}\;d\freq'\\
\nonumber
&=\lambda^2\widetilde{f}(\freq)\int_{\freq_1-\frac{1}{\bside_1^{1-\beta}}}^{\freq_1+\frac{1}{\bside_1^{1-\beta}}}\dots \int_{\freq_p-\frac{1}{\bside_p^{1-\beta}}}^{\freq_p+\frac{1}{\bside_p^{1-\beta}}}G^{(1)\ast}(\freq-\freq') \;d\freq'
\\
&+\frac{1}{2}{\mathrm{trace}}
\left\{\tilde{\mathbf{H}}_f(\freq'') 
\int_{\freq_1-\frac{1}{\bside^{1-\beta}}}^{\freq_1+\frac{1}{\bside^{1-\beta}}}\dots \int_{\freq_p-\frac{1}{\bside^{1-\beta}}}^{\freq_p+\frac{1}{\bside^{1-\beta}}}
|B|^{-1}T(B,\freq-\freq')  \cdot (\freq'-\freq) \left(\freq'-\freq\right)^T\;d\freq'\right\}.
\end{align}
We can then note that
\[\int_{\freq_1-\frac{1}{\bside^{1-\beta}}}^{\freq_1+\frac{1}{\bside^{1-\beta}}}\dots \int_{\freq_p-\frac{1}{\bside^{1-\beta}}}^{\freq_p+\frac{1}{\bside^{1-\beta}}}G^{(1)}(\freq-\freq') \;d\freq'=1+o(1).\]
%To obtain this result,
%we note that when $z$ is zero then $ | B\times B_{-z}|=|B|$. 
By direct calculation 
we then find
\begin{align*}
\int_{R^d}G^{(1)}(\freq)d\freq &=|B|^{-1}\int_{R^d}\int_{R^d}  | B\times B_{-z}|e^{-i2\pi \freq\cdot z}\,dzdw\\
&=|B|^{-1}\int_{R^d}  | B\times B_{-z}|\delta(z) \,dz\\
&=1.
\end{align*}
Note that this is now consistent with~\eqref{normT}, and achieves a value of $1$.
Second,
\begin{equation*}
\int_{\freq_1-\frac{1}{\bside^{1-\beta}}}^{\freq_1+\frac{1}{\bside^{1-\beta}}}\dots \int_{\freq_p-\frac{1}{\bside^{1-\beta}}}^{\freq_p+\frac{1}{\bside^{1-\beta}}}\left(\freq'-\freq\right)^T
G^{(1)\ast}(\freq-\freq') \;d\freq'=0,
\end{equation*}
as
\[ | B\times B_{-z}|= | B\times B_{z}|,\]
 implies that this function $| B\times B_{-z}|$
   is symmetric (or even) in $z$, and real-valued, which implies that its Fourier transform also is real-valued. The symmetric integral of an even times an odd function is always zero. 

We now also determine
\begin{equation}
{\mathbf{B}}=
\int_{\freq_1-\frac{1}{\bside_1^{1-\beta}}}^{\freq_1+\frac{1}{\bside_1^{1-\beta}}}\dots \int_{\freq_p-\frac{1}{\bside^{1-\beta}}}^{\freq_p+\frac{1}{\bside^{1-\beta}}}
G^{(1)\ast}(\freq-\freq')\cdot (\freq-\freq') \left(\freq-\freq'\right)^T\;d\freq'.
\end{equation}
Let us start by a change of variables. Let us first define the matrix
\[\bm{L}={\mathrm{diag}}\left(\bside_1,\dots,\bside_d\right).\]
Thus the change of variables is
\begin{equation}
\nu=\pi \bm{L}\left(\freq-\freq'\right) \Longleftrightarrow 
\freq'=\freq-\bm{L}^{-1}\nu/\pi.
\end{equation}
%\newpage
With this change of variables we get that
\begin{align}
\nonumber
{\mathbf{B}}&=
\int_{-\bside_1^\beta/2}^{\bside_1^\beta/2}\dots \int_{-\bside_p^\beta/2}^{\bside_p^\beta/2}
G^{(1)\ast}(\freq-\freq')\cdot (\freq-\freq') \left(\freq-\freq'\right)^T\;d\freq'\\
\nonumber
&=\int_{-\bside_1^\beta/2}^{\bside_1^\beta/2}\dots \int_{-\bside_p^\beta/2}^{\bside_p^\beta/2}
\prod_{j=1}^d \frac{\pi^2\bside_j^2\sin^2(\nu_j)}{|B|\nu_j^2}\bm{L}^{-1}\pi^{-d}\nu \nu^T \bm{L}^{-1}\pi^{-d}\;d\nu_j  (1/\bside_j)\\
&=\int_{-\bside_1^\beta/2}^{\bside_1^\beta/2}\dots \int_{-\bside_p^\beta/2}^{\bside_p^\beta/2}
\prod_{j=1}^d \frac{\bside_j\sin^2(\nu_j)}{|B|\nu_j^2}\bm{L}^{-1}\nu \nu^T \bm{L}^{-1}\;d\nu_j.
\end{align}
We then find if all the $\bside_j=\bside$ then
if $0<\beta$
\begin{align}
\nonumber
{\mathrm{trace}}\left\{{\mathbf{B}}\right\}&=
\int_{-\bside^\beta/2}^{\bside^\beta/2}\dots \int_{-\bside^\beta/2}^{\bside^\beta/2}
\prod_{j=1}^d \frac{\sin^2(\nu_j)}{\nu_j^2}\sum_{q=1}^d (\nu_q/\bside)^2 \;d\nu_j\\
\nonumber
&=
d\cdot \frac{1}{\bside^2}\left(\int_{-\infty}^{\infty} \frac{\sin^2(x)}{x^2}\,dx \right)^{d-1}\int_{-\bside^\beta/2}^{\bside^\beta/2}\sin^2(x)\,dx
\\
&=\frac{d}{\bside^2} \pi^{d-1}\bside^\beta
\left\{1+o(1) \right\}.
\end{align}
This decreases with increasing $\bside$ as long as $\beta<2$ (which follows as we have assumed $\beta<1$), and as the next error term is $O(\bside^{-\beta})$. We therefore want to chose $\beta\rightarrow 1$, or for some fixed $\epsilon>0$ take $\beta=1-\epsilon$.
Putting all of these components together we get that
\begin{align*}
\E I_0(\freq)&= \lambda + |B|^{-1}\int_{\R^d}|B\cap B_{-z}|e^{-2\pi i\freq\cdot z}[\kap{z}]dz + |B|^{-1}\lambda^2 T(B, \freq)\\
&= \lambda + \lambda^2\widetilde{f}(\freq)+\frac{1}{2}
{\mathrm{trace}}\left\{\tilde{\mathbf{H}}_f(\freq'') \frac{d}{\bside^{2-\beta}} \pi^{d-1}
\left\{1+o(1) \right\}
\right\}+
|B|^{-1}\lambda^2 T(B, \freq)\\
&=\lambda + (-\lambda+f(\freq))\left\{1+o(1) \right\}+\frac{1}{2}
{\mathrm{trace}}\left\{\tilde{\mathbf{H}}_f(\freq'') \frac{d}{\bside^{2-\beta}} \pi^{d-1}
\left\{1+o(1) \right\}
\right\}+
|B|^{-1}\lambda^2 T(B, \freq)\\
&=f(\freq)\left\{1+o(1) \right\}+o(1)+
|B|^{-1}\lambda^2 T(B, \freq).
\end{align*}
%\sofia{notation fixed, checked up to.}
\end{proof}

\section{Proof of Proposition~\ref{DFTprop}}\label{DFTproppf}
\begin{proof}
We start from using Eqn.~\eqref{eqn:Campbell} and determine that
\begin{align}
    \E\left\{ J_h(\freq)\right\}&=\E \sum_{x\in X}h(x)e^{-2\pi i \freq\cdot  x}= \lambda\int_{R^d}h(x)e^{-2\pi i \freq\cdot  x}\,dx
    \label{Hk}
    =\lambda H(\freq).
\end{align}
%\sofia{Checked.}
Subsequently we find using Eqn.~\eqref{eqn:Campbell} 
\begin{align}
\nonumber
    \E\left\{\left| J_h(\freq)\right|^2\right\}
    &= \E\left\{\sum_{x\in X}\sum_{y\in X}
    h(x)e^{-2\pi i \freq\cdot  x}h^\ast(y)e^{2\pi i \freq\cdot  y}\right\}\\
  \nonumber
    &=\E \sum_{x\in X}\left|h(x)\right|^2+\E\left\{\sum_{x\in X,y\in X}^{\neq}h(x)e^{-2\pi i \freq\cdot  x}h^\ast(y)e^{2\pi i \freq\cdot  y}\right\}\\
    \nonumber
    &=\lambda \int_{R^d}\left|h(x)\right|^2dx+\iint_{R^d\times R^d}{\rho}^{(2)}(x-y)h(x)e^{-2\pi i \freq\cdot  x}h^\ast(y)e^{2\pi i \freq\cdot  y}\,dx\,dy\\
    &=\lambda +\iint_{R^d\times R^d}{\rho}^{(2)}(x-y)h(x)e^{-2\pi i \freq\cdot  x}h^\ast(y)e^{2\pi i \freq\cdot  y}\,dx\,dy.
\end{align}
%\sofia{Checked.}
We now calculate the variance as
\begin{align}
\nonumber
    \var \left\{ J_h(\freq)\right\}
    &= \E\left\{\left|J_h(\freq)\right|^2\right\}-
    \left|\E\left\{ J_h(\freq)\right\}\right|^2\\
    &=\lambda +\iint_{R^d\times R^d}{\rho}^{(2)}(x-y)h(x)e^{-2\pi i \freq\cdot  x}h^\ast(y)e^{2\pi i \freq\cdot  y}\,dx\,dy-\left|\lambda H(\freq)
    \right|^2.
    \label{var:expr}
\end{align}
We now define the renormalised quantity
\begin{equation}
    \widetilde{\rho}^{(2)}(z)=\frac{ {\rho}^{(2)}(z)-\lambda^2}{\lambda^2},\quad z\in{\mathbb{R}}^d.
\end{equation}
The expression in~\eqref{var:expr} then can be simplified to
\begin{align}
\nonumber
    \var \left\{ J_h(\freq)\right\}
    &=\lambda +\iint_{R^d\times R^d}{\rho}^{(2)}(x-y)h(x)e^{-2\pi i \freq\cdot  x}h^\ast(y)e^{2\pi i \freq\cdot  y}\,dx\,dy-\left|\lambda H(\freq)
    \right|^2\\
    &= \nonumber
    \lambda +\lambda^2\iint_{R^d\times R^d}\left(\widetilde{\rho}^{(2)}(x-y)+1\right)h(x)e^{-2\pi i \freq\cdot  x}h^\ast(y)e^{2\pi i \freq\cdot  y}\,dx\,dy-\left|\lambda H(\freq)
    \right|^2\\
    &=\lambda+\lambda^2\iint_{R^d\times R^d}\widetilde{\rho}^{(2)}(x-y)h(x)e^{-2\pi i \freq\cdot  x}h^\ast(y)e^{2\pi i \freq\cdot  y}\,dx\,dy\nonumber \\
    &=\lambda+ \lambda^2\var_1 \left\{ J_h(\freq)\right\},
    \end{align}
    where we define
    \begin{equation}
        \var_1 \left\{ J_h(\freq)\right\}
        \equiv \iint_{R^d\times R^d}\widetilde{\rho}^{(2)}(x-y)h(x)e^{-2\pi i \freq\cdot  x}h^\ast(y)e^{2\pi i \freq\cdot  y}\,dx\,dy.
    \end{equation}
    
    To simplify this expression we note that
    \begin{align}
    \nonumber
   \var_1 \left\{ J_h(\freq)\right\}     &=\int_{R^d\times R^d} 
        \widetilde{\rho}^{(2)}(x-y)h(x)e^{-2\pi i \freq\cdot  x}h^\ast(y)e^{2\pi i \freq\cdot  y}\,dx\,dy\\
        \nonumber
        &=\int_{R^d\times R^d} 
        \left[\int_{R^d}\widetilde{f}^{(2)}(\freq') e^{2\pi i (x-y)\cdot \freq'} d\freq'\right] \cdot h(x)e^{-2\pi i \freq\cdot  x}h^\ast(y)e^{2\pi i \freq\cdot  y}\,dx\,dy\\
        \nonumber
        &=\int_{R^d\times R^d} 
        \int_{R^d}\widetilde{f}^{(2)}(\freq')e^{-2\pi i (x-y)\cdot(\freq-\freq')}d\freq' \cdot h(x)h^\ast(y)\,dx\,dy\\
        &=\int_{R^d}\widetilde{f}^{(2)}(\freq') \left|H(\freq-\freq')\right|^2dw' .
    \end{align}
    %\sofia{Checked.}
    This yields the desired expression for the variance. We return to the expression of the relation of the DFT.
    We aim to show that
    \[\rel\left\{ J_h(\freq)\right\}=\lambda \int_{\mathbb{R}^d} 
H(\freq'-2\freq) H(\freq')d\freq'+\int_{\mathbb{R}^d} U(\freq,z) e^{-2\pi i \freq\cdot  z}\{\rho(z)-\lambda^2\}dxdz.\]
    We have by the definition of the relation~\cite{schreier2010statistical}
    \begin{align}
\nonumber
    \rel \left\{ J_h(\freq)\right\}
    &= \E\left\{\left( J_h(\freq)\right)^2\right\}-
    \left(\E\left\{ J_h(\freq)\right\}\right)^2.
    \label{rel:expr}
\end{align}
We then note that
\begin{align}
\nonumber
    \E\left\{\left( J_h(\freq)\right)^2\right\}&=
    \E\left\{\sum_{x\in X}\sum_{y\in X}
    h(x)e^{-2\pi i \freq\cdot  x}h(y)e^{-2\pi i \freq\cdot y}\right\}\\
  \nonumber
    &=\E \sum_{x\in X}\left\{h(x)\right\}^2e^{-4\pi i \freq\cdot  x}+\E\left\{\sum_{x\in X,y\in X}^{\neq}h(x)e^{-2\pi i \freq\cdot  x}h(y)e^{-2\pi i \freq\cdot y}\right\}\\
    \nonumber
    &=\lambda \int_{R^d}\left\{h(x)\right\}^2e^{-4\pi i \freq\cdot  x}dx+\iint_{R^d\times R^d}\bar{\rho}^{(2)}(x-y)h(x)e^{-2\pi i \freq\cdot  x}h(y)e^{-2\pi i \freq\cdot  y}\,dx\,dy\\
    &=\lambda\int_{R^d}\left\{h(x)\right\}^2 e^{-4\pi i \freq\cdot  x} dx +
     \rel_1 \left\{ J_h(\freq)\right\}.
\end{align}
We now seek to simplify this and write
\begin{equation}
    \rel_1 \left\{ J_h(\freq)\right\}
    =\lambda^2 \iint_{R^d\times R^d}\left(\widetilde{\rho}^{(2)}(x-y) +1\right)h(x)e^{-2\pi i \freq\cdot  x}h(y)e^{-2\pi i \freq\cdot  y}\,dx\,dy.
\end{equation}
We additionally have from~\eqref{Hk}:
\[ \E^2\left\{ J_h(\freq)\right\}=\lambda^2 H^2(\freq).\]
Putting all of this together we get that 
\begin{align*}
     \rel \left\{ J_h(\freq)\right\}
     &=\lambda\int_{R^d}\left\{h(x)\right\}^2 e^{-4\pi i \freq\cdot  x} dx+\lambda^2 \iint_{R^d\times R^d}\left(\widetilde{\rho}^{(2)}(x-y) +1\right)h(x)e^{-2\pi i \freq\cdot  x}h(y)e^{-2\pi i \freq\cdot  y}\,dx\,dy\\
     &-\lambda^2 H^2(\freq)\\
     &=\lambda\int_{R^d}\left\{h(x)\right\}^2 e^{-4\pi i \freq\cdot  x} dx+\lambda^2 \iint_{R^d\times R^d} \widetilde{\rho}^{(2)}(x-y) h(x)e^{-2\pi i \freq\cdot  x}h(y)e^{-2\pi i \freq\cdot  y}\,dx\,dy.
\end{align*}
We now implement a change of variables, and set $z=x-y$. We then
have 
\begin{equation*}
    \iint_{R^d\times R^d} \!\!\!\!\widetilde{\rho}^{(2)}(x-y) h(x)e^{-2\pi i \freq\cdot  x}h(y)e^{-2\pi i \freq\cdot  y}\,dx\,dy=
     \iint_{R^d\times R^d} \!\!\!\!\widetilde{\rho}^{(2)}(z) h(x)e^{-2\pi i \freq\cdot  x}h(x-z)e^{-2\pi i \freq\cdot  (x-z)}\,dx\,dz,
\end{equation*}
and so with the definition of
\begin{equation}
    U(\k,z)\equiv \int h(x) h(z+x) e^{-2\pi i \freq\cdot (2x)}
    dx,
\end{equation}
the expression is a consequence of the above expressions. 
%\sofia{Checked.}
\end{proof}

\section{Proof of Lemma~\ref{lemdebias}}\label{pflemdebias}
\begin{proof}
We start from
\[
\widetilde{I}_h(\freq)=\left| J_h(\freq)-\lambda H(\freq)\right|^2,\]
and so take
\begin{equation*}
    \E \left\{ \widetilde{I}_h(\freq)\right\}
    =\E\{ \left| J_h(\freq)\right|^2\}-
   \lambda \E \{J_h(\freq) H^\ast(\freq)\}
   -
   \lambda \E\{ J_h^\ast(\freq) H(\freq)\}+
  \lambda^2 \E \{\left|H(\freq)\right|^2\}.
\end{equation*}
We now use Campbell's theorem to evaluate these expectations.
We already have from Proposition~\ref{DFTprop} that
\begin{align}
\nonumber
    \E\left\{\left| J_h(\freq)\right|^2\right\}
    &=\lambda +\iint_{R^d\times R^d}{\rho}^{(2)}(x-y)h(x)e^{-2\pi i \freq\cdot  x}h^\ast(y)e^{2\pi i \freq\cdot  y}\,dx\,dy.
    \end{align}
Also we note that
\begin{align}
    H(\freq)&=\int_{R^d}h(x)e^{-2\pi i \freq\cdot  x}\,dx.
\end{align}
We also note
\begin{align}
\nonumber
    \E\left\{ J_h(\freq)\right\}
    &=\lambda H(\freq).
    \end{align}
Therefore we have
\begin{align}
\nonumber
  \!  \E \left\{ \widetilde{I}_h(\freq)\right\}\!
    &=\!\E\{ \left| J_h(\freq)\right|^2\}-
   \lambda \E \{J_h(\freq) H^\ast(\freq)\}
   -
   \lambda \E\{ J_h^\ast(\freq) H(\freq)\}+\lambda^2
   \E \{\left|H(\freq)\right|^2\}\\
   \nonumber
   &=\E\{ \left| J_h(\freq)\right|^2\}-\lambda^2
   \left|H(\freq)\right|^2\\
   &=\lambda +\iint_{R^d\times R^d}\{{\rho}^{(2)}(x-y)-\lambda^2\}h(x)e^{-2\pi i \freq\cdot  x}h^\ast(y)e^{2\pi i \freq\cdot  y}\,dx\,dy.
   \label{convy}
\end{align}
% We additionally note
% \begin{equation*}
% \int_{R^d}\{{\rho}^{(2)}(z)-\lambda^2\}e^{-2\pi i \freq\cdot  z}\,dz=f(\freq)-\lambda.
% \end{equation*}
We note that the multiplication in~\eqref{convy} leads to a convolution in the wavenumber domain. If we rewrite
\begin{align}
{\rho}^{(2)}(z)-\lambda^2&=
\int_{R^d} (f(\freq)-\lambda)e^{2\pi i \freq\cdot  z}\,d\freq,
\end{align}
then rewriting~\eqref{convy} we get
\begin{align}
\nonumber
\!  \E \left\{ \widetilde{I}_h(\freq)\right\}\! 
&=\lambda +\iint_{R^d\times R^d}\int_{R^d} (f(u)-\lambda)\exp(-2\pi i u\cdot (x-y)+2\pi i \k\cdot  (x-y))\,du h(x)h^\ast(y)\,dx\,dy\\
\nonumber
&=\lambda +\int_{R^d} (f(u)-\lambda) \|H(\k-u)\|^2\,du=
\int_{R^d} f(u) \|H(\k-u)\|^2\,du,
\end{align}
as given.
%\sofia{Checked.}
\end{proof}

\section{Proof of Corollary~\ref{cor1}}\label{proofcor1}
\begin{proof}
Starting from Lemma~\ref{lemdebias} we can note that
\begin{equation}
\label{EIhw}
    \E\{\widetilde{I}_h(\freq) \}=\int_{{\mathbb{R}}^d}
    \left| H(\freq'-\freq)\right|^2 f(\freq')\, d\freq'.
\end{equation}
Now we need to assume properties of $|H(\freq)|$ in order to simplify this expression. Assume the cuboid domain has side length $\bside$, and that we have selected a data taper so that it has unit norm:
\[\int_{{\mathbb{R}}^d} \|h(x)\|^2\,dx=1. \]
We assume the data taper $h(x)$ is compactly supported in space and well--concentrated in wavenumber to region ${\mathcal W}\subset {\mathbb{R}}^d$ so that
\begin{equation}\iint_{\mathcal W} |H(\freq)|^2\,d\freq=1-\varepsilon_l,\end{equation}
where $l>0$ is the minimum length scale in any dimension of ${\mathcal W}$, and we assume $\epsilon_l\rightarrow 0$ as $|l|\rightarrow \infty.$ We now repeat the arguments posed in Appendix~\ref{sec:Pf1}, for a different window $h(x)$ then the (spatial) box--car to the region. 

We now return to~\eqref{EIhw} and present a similar argument to Appendix~\ref{sec:Pf1}. We note that
\begin{align}
\nonumber
\E\{\widetilde{I}_h(\freq) \}&=\int_{{\mathbb{R}}^d}
    \left| H(\freq'-\freq)\right|^2 f(\freq')\, d\freq'\\
    \nonumber
&=\int_{{\mathbb{R}}^d}
    \left| H(\freq'-\freq)\right|^2 \left(\lambda+\lambda^2\widetilde{f}(\freq') \right)\, d\freq'\\
    \nonumber
&=\lambda+\lambda^2\int_{{\mathbb{R}}^d}
     \left| H(\freq'-\freq)\right|^2\widetilde{f}(\freq')\,d\freq'\\
     \nonumber
&=\lambda+\lambda^2\int_{\mathcal W}
\left| H(\freq'-\freq)\right|^2\widetilde{f}(\freq')\,d\freq'+
\lambda^2\int_{{\mathbb{R}}^d\backslash {\mathcal W}}
     \left| H(\freq'-\freq)\right|^2\widetilde{f}(\freq')\,d\freq'\\
     &=\lambda+\E_1\{\widetilde{I}_h(\freq) \}+
     \E_2\{\widetilde{I}_h(\freq) \},
\end{align}
where $f^{[n]}$ denotes the $n$th derivative of the function $f$ and we use this equation to define the latter two terms. We note that 
\begin{align}
\left|  \E_2\{\widetilde{I}_h(\freq) \}\right| &\leq 
\lambda^2 \varepsilon_l \|\widetilde{f}\|_0.
\end{align}
Thus we only need to understand the remaining term in the expression. We then obtain with a Lagrange form of the remainder
\begin{align}
\nonumber
\E_1\{\widetilde{I}_h(\freq) \}&=\lambda^2\int_{\mathcal W}
\left| H(\freq'-\freq)\right|^2\widetilde{f}(\freq')\,d\freq'\\
\nonumber
&=\lambda^2\int_{\mathcal W}
\left| H(\freq'-\freq)\right|^2\left[\widetilde{f}(\freq)+\widetilde{f}^{[1]}(\freq)(\freq'-\freq)+\frac{1}{2}\widetilde{f}^{[2]}(\freq)(\freq'-\freq)^2+R^{[3]}(\freq,\freq')\right]\,d\freq'\\
&=\lambda^2 (1-\varepsilon_l)\widetilde{f}(\freq)+\frac{1}{2}\lambda^2\widetilde{f}^{[2]}(\freq) \int_{\mathcal W}
\left| H(\freq'-\freq)\right|^2(\freq'-\freq)^2\,d\freq'+\widetilde{R}^{[3]}(\freq).
\end{align}
We can note that 
\begin{equation}
|{R}^{[3]}(\freq,\freq')|\leq \frac{1}{3!}\sup_{\freq''} |\widetilde{f}^{[3]}(\freq'')|\int |\freq'-\freq|^3\left| H(\freq'-\freq)\right|^2\,d\freq'.
\end{equation}
Putting these terms together we obtain that
\begin{align}
\E\{\widetilde{I}_h(\freq) \}&=\lambda+\lambda^2 (1-\varepsilon_l)\widetilde{f}(\freq)+\frac{1}{2}\lambda^2\widetilde{f}^{[2]}(\freq) \int_{\mathcal W}
\left| H(\freq'-\freq)\right|^2(\freq'-\freq)^2\,d\freq'+\widetilde{R}^{[3]}(\freq)\\
&=\lambda+\lambda^2 \widetilde{f}(\freq)+o(1),
\end{align}
which completes the expression.
\end{proof}

\section{Proof of Theorem~\ref{thm1}}
\label{sec:Pf2}
\begin{proof}
The expectation of the tapered Fourier transform is (see~\eqref{eqn:Campbell})
\begin{align}
\nonumber
\E I^t(\freq) &= 
 \int_{{\mathbb{R}}^{2d}} 
h(x)h(y)e^{-2\pi i\freq\cdot (x-y)} \rho^{(2)}(x,y)\,dx\,dy+\int_{{\mathbb{R}}^{d}} 
h^2(x) \rho^{(1)}(x)\,dx
\\
&=\lambda\int_\R h^2(x) dx + \int_{B^2} h(x)h(y)e^{-2\pi i\freq\cdot (x-y)}\rtwo(x-y)dxdy\\
\nonumber
&= \lambda\cdot 1 + \int_{B^2} h(x)h(y)e^{-2\pi i\freq\cdot (x-y)}[\rtwo(x-y)-\lambda^2]dxdy \\
\nonumber
&+ \lambda^2 \int_{B^2} h(x)h(y)e^{-i2\pi \freq\cdot (x-y)}dxdy\\
\nonumber
&= \lambda  + \int_{\R^d}\int_{\R^d}h(x)h(y)e^{-2\pi i\freq\cdot (x-y)}[\rtwo(x-y)-\lambda^2]dxdy+\lambda^2 h(B,\freq)h(B,-\freq)\\
\nonumber
&= \lambda  + \int_{\R^d}\int_{\R^d}h(x)h(y)e^{-2\pi i\freq\cdot (x-y)}[\rtwo(x-y)-\lambda^2]dxdy+\lambda^2 \left|H(\freq)\right|^2\\
&=\lambda  + \E I^{t,(2)}(\freq)+\lambda^2 \left|H(\freq)\right|^2,
\label{expectationper}
\end{align}
this defining the quantity
$\E I^{t,(2)}(\freq).$

The expression $\E I^{t,(2)}(\freq)$ mirrors what we can see in~\eqref{expPerio}.
We see directly the benefits of tapering--$|H(\freq)|$ will decay rapidly from zero so there will be no effect from the third term--there will be less blurring in the second term apart from very close to $\freq=0$.

From which we see that we can re-express this understanding in the Fourier domain, again using the convolution theorem. We define
\begin{equation*}
\E I^{t,(2)}(\freq)= \int_{\R^d}\int_{\R^d}h(x)h(y)e^{-2\pi i\freq\cdot (x-y)}[\rtwo(x-y)-\lambda^2]dxdy.
\end{equation*}
To explore this further we again do a change of variables:
\begin{equation*}
\E I^{t,(2)}(\freq)= \int_{\mathbb{R}^d}
\int_{B_{z}}h(x)h(x-z)e^{-i2\pi \freq\cdot z}[\rtwo(z)-\lambda^2]dxdz.
\end{equation*}
We define the inner product window using that $h(x)$ is compactly supported to get
\[W_h(z)=\int_{B_z}h(x)h(x-z)dx=\int_{\R^d}h(x)h(x-z)dx=
\int_{\R^d} \left| H(\freq)\right|^2 e^{2\pi i\freq\cdot z}\,d\freq.\]
We could also divide by $|B\cap B_{-z}|h(x)h(x-z)$
as long as we are inside the domain. Outside $B$ we cannot, and so
this is why we cannot remove all bias in the spectrum even if we bias--correct as suggested  by~\cite{guyon1982parameter}.

With this window we can note that
\begin{equation}
\label{rapered}
\E I^{t,(2)}(\freq)= \lambda^2\int_{\mathbb{R}^d}
\left| H(\freq'-\freq)\right|^2\widetilde{f}(\freq')\;d\freq'
.\end{equation}
Ideally we would like $\left| H(\freq)\right|^2=\delta(\freq)$, but this is not possible as
$|B|$ is finite. If we assume 
$\widetilde{f}(\freq)$ is smooth and possessing two derivatives then if $B$ is growing we can obtain nearly unbiasedness. 
We now assume that $\left| H(\freq)\right|^2$
is concentrated to a region in wavenumber ${\cal W}$ and again use a Taylor expansion. There are two approaches to this, either~\cite{Riedel1995} or~\cite{thomson1990quadratic}.

The systematic bias: $\lambda^2 h(B,\freq)h(B,-\freq)=\left| H(\freq)\right|^2$ replaces the sinc functions for the square box. As we have chosen the function $h$ to be well--concentrated~\cite{thomson1990quadratic,Riedel1995,walden2000unified}, the effect of this is limited.

We assume that 
\begin{equation}
\int_{{\cal W}}\left| H(\freq)\right|^2\;d\freq=1-\varepsilon_{\bside},\quad {\mathrm{where}}\quad \varepsilon_{\bside}=o(1).
\end{equation}
We also assume that $\tilde{f}(\freq)$ is upper bounded by $\|\widetilde{f}\|_0$. We then have
\begin{align}
\E I^{t,(2)}(\freq)&= \lambda^2\int_{\mathbb{R}^d\backslash {\cal W}}
\left| H(\freq'-\freq)\right|^2\widetilde{f}(\freq')\;d\freq'
+\lambda^2\int_{ {\cal W}}
\left| H(\freq'-\freq)\right|^2\widetilde{f}(\freq')\;d\freq'.
\end{align}
We note that
\[\int_{\mathbb{R}^d\backslash {\cal W}}
\left| H(\freq'-\freq)\right|^2\widetilde{f}(\freq')\;d\freq'
\leq \|\widetilde{f}\|_0\{1-\{ 1-\varepsilon_{\bside}\}\}.\]
We can yet again utilise the Taylor expansion of
\eqref{tildetaylor} inside ${\cal W}$.
We find
\begin{align}
\nonumber
\E I^{t,(2)}(\freq)&=\lambda^2\int_{ {\cal W}}
\left| H(\freq'-\freq)\right|^2\widetilde{f}(\freq')\;d\freq'
\left\{1+o(1)\right\}\\
\nonumber
&=\lambda^2\int_{ {\cal W}}
\left| H(\freq'-\freq)\right|^2\left\{\widetilde{f}(\freq)+\nabla \widetilde{f}(\freq)^T\left(\freq'-\freq\right)+
\frac{1}{2}\left(\freq'-\freq\right)^T\tilde{\mathbf{H}}_f(\freq'')\left(\freq'-\freq\right)\right\}\;d\freq'
\left\{1+o(1)\right\}\\
&=\lambda^2\widetilde{f}(\freq)+0+
\lambda^2\int_{ {\cal W}}\left| H(\freq'-\freq)\right|^2\frac{1}{2}\left(\freq'-\freq\right)^T\tilde{\mathbf{H}}_f(\freq'')
\left(\freq'-\freq\right)\;d\freq'.
\end{align}
We note that
\begin{align*}
&\|\int_{ {\cal W}}\left| H(\freq'-\freq)\right|^2\frac{1}{2}\left(\freq'-\freq\right)^T\tilde{\mathbf{H}}_f(\freq'')
\left(\freq'-\freq\right)\;d\freq'\|^2\\
&\leq |\tilde{\mathbf{H}}_f(\freq'')|
\int_{ {\cal W}}\left| H(\freq'-\freq)\right|^2\frac{1}{2}\left(\freq'-\freq\right)^T
\left(\freq'-\freq\right)\;d\freq'\\
&=|\tilde{\mathbf{H}}_f(\freq'')|
\int_{ {\cal W}}\left| H(\freq)\right|^2 w^2 d\freq=
{\cal O}(1/\bside).
\end{align*}
Then
\begin{align*}
\E I^{t}(\freq)
&=\lambda  + \E I^{t,(2)}(\freq)+\lambda^2  \left|H(\freq)\right|^2\\
&=f(\k)+0+
{\cal O}(1/\bside)+\lambda^2  \left|H(\freq)\right|^2.
\end{align*}
We have that 
$\left|H(\freq)\right|^2$ decays in $\|\freq\|$. If $\freq\notin {\mathcal W}$ then the influence of this term becomes negligible.
%\qed
%\sofia{checked,notation corrected.}
\end{proof}

\section{Proof of Theorem~\ref{Propcovar}}
\label{sec:Pf3}
\begin{proof}For brevity we define the notation $\phi_\k(x)=e^{-2\pi i\k\cdot x}$. From first principles using Campbell's theorem we may deduce that the uncentred second moment of $I^t_\itapers(\freq_1)$ with $I^t_\itapersalt(\freq_2)$ takes the
form of
\begin{align}
\nonumber
 \E \left\{I^t_\itapers(\freq_1) I^t_\itapersalt(\freq_2)\right\} =& \E\left[\sum^{\neq}_{x,y\in X}
 h_\itapers(x) h_\itapers(y) \phi_{\freq_1}(x-y) + \sum_{x\in X}
 h_\itapers^2(x)  \right]\left[\sum_{z,v\in X}^{\neq} h_\itapersalt(z) h_\itapersalt(v)\phi_{\freq_2}(z-v) + \sum_{q\in X}
 h_\itapersalt^2(q) \right] \\
 \nonumber
=& \E\left\{ \sum^{\neq}_{x,y\in X}h_\itapers(x) h_\itapers(y)\phi_{\freq_1}(x-y)\sum_{z,v\in X}^{\neq} h_\itapersalt(z) h_\itapersalt(v)\phi_{\freq_2}(z-v)\right\} \\
\nonumber
& + \E \left\{\sum_{x\in X}
 h_\itapers^2(x) \sum^{\neq}_{y,z\in X}h_\itapersalt(y)h_\itapersalt(z)\phi_{\freq_2}(y-z)+\sum_{x\in X}
 h_\itapersalt^2(x) \sum^{\neq}_{y,z\in X}h_\itapers(y)h_\itapers(z)\phi_{\freq_1}(y-z)\right\}\\ 
 \nonumber
& + \E \left[\sum_{x\in X}
 h_\itapers^2(x) \sum_{y\in X}
 h_\itapersalt^2(y)\right] \\
= & V_1(\freq_1,\freq_2) + V_2(\freq_1,\freq_2) + V_3 .
\end{align}
This latter equation defines the three terms $V_1(\freq_1,\freq_2) $, $V_2(\freq_1,\freq_2)$ and $V_3$ (where the latter does not depend on $\freq_1$ and $\freq_2$ even if $V_2$ is the sum of two terms, one only depending on $\freq_1$ and the other only depending on $\freq_2$).
We now need to further split these terms up in order to learn about the cases when we can have repeats of locations or not. We start with $V_1$. For convenience we suppress ``$\in X$'' notation. First, we note:
\begin{align}
\nonumber
V_1(\freq_1,\freq_2)&=\E\left\{ \sum^{\neq}_{x,y}h_\itapers(x) h_\itapers(y)\phi_{\freq_1}(x-y)\sum_{z,v}^{\neq} h_\itapersalt(z) h_\itapersalt(v)\phi_{\freq_2}(z-v)\right\} \\
 \nonumber
 &=\E\sum^{\neq}_{x,y}\sum_{z,v}^{\neq}
 \I\left(x\neq z,x\neq v \right)
 \I\left(y\neq z,y\neq v \right)
 h_\itapers(x) h_\itapers(y)h_\itapersalt(z) h_\itapersalt(v)\phi_{\freq_1}(x-y) \phi_{\freq_2}(z-v)\\
 \nonumber
 &+\E\sum^{\neq}_{x,y}\sum_{z,v}^{\neq}
 \I\left(x=z ,x\neq v \right)
 \I\left(y\neq z,y\neq v \right)
 h_\itapers(x) h_\itapers(y)h_\itapersalt(z) h_\itapersalt(v)\phi_{\freq_1}(x-y) \phi_{\freq_2}(x-v)\\
 \nonumber
 &+\E\sum^{\neq}_{x,y}\sum_{z,v}^{\neq}
 \I\left(x\neq z,x= v \right)
 \I\left(y\neq z,y\neq v \right)
 h_\itapers(x) h_\itapers(y)h_\itapersalt(z) h_\itapersalt(v)\phi_{\freq_1}(x-y) \phi_{\freq_2}(z-v)\\
 \nonumber
 &+\E\sum^{\neq}_{x,y}\sum_{z,v}^{\neq}
 \I\left(x\neq z,x\neq v \right)
 \I\left(y= z,y\neq v \right)
 h_\itapers(x) h_\itapers(y)h_\itapersalt(z) h_\itapersalt(v)\phi_{\freq_1}(x-y) \phi_{\freq_2}(z-v)\\
 \nonumber
 &+\E\sum^{\neq}_{x,y}\sum_{z,v}^{\neq}
 \I\left(x\neq z,x\neq v \right)
 \I\left(y\neq z,y= v \right)
 h_\itapers(x) h_\itapers(y)h_\itapersalt(z) h_\itapersalt(v)\phi_{\freq_1}(x-y) \phi_{\freq_2}(z-v)\\
 \nonumber
 &+\E\sum^{\neq}_{x,y}\sum_{z,v}^{\neq}
 \I\left(x= z,x\neq v \right)
 \I\left(y= v,y\neq z \right)
 h_\itapers(x) h_\itapers(y)h_\itapersalt(z) h_\itapersalt(v)\phi_{\freq_1}(x-y) \phi_{\freq_2}(z-v)\\
 \nonumber
 &+\E\sum^{\neq}_{x,y}\sum_{z,v}^{\neq}
 \I\left(x=v,x\neq z \right)
 \I\left(y= z,y\neq v \right)
 h_\itapers(x) h_\itapers(y)h_\itapersalt(z) h_\itapersalt(v)\phi_{\freq_1}(x-y) \phi_{\freq_2}(z-v).
\end{align}
We have here assumed that $h_\itapers(x)$ is compactly supported on $B$ for all values of $p$ used. Using the simplification implied by~\eqref{eqn:Campbell} we obtain the expression of
\begin{align}
\nonumber
V_1(\freq_1,\freq_2)&=\E\left\{ \sum^{\neq}_{x,y}h_\itapers(x) h_\itapers(y)\phi_{\freq_1}(x-y)\sum_{z,v}^{\neq} h_\itapersalt(z) h_\itapersalt(v)\phi_{\freq_2}(z-v)\right\} \\
 \nonumber
 &=\int_{B^4}
 \rho^{(4)}(x,y,z,v)
 h_\itapers(x) h_\itapers(y)h_\itapersalt(z) h_\itapersalt(v)\phi_{\freq_1}(x-y) \phi_{\freq_2}(z-v)dxdydvdz\\
 \nonumber
 &+\int_{B^3}\rho^{(3)}(x,y,v)
 h_\itapers(x) h_\itapers(y)h_\itapersalt(x) h_\itapersalt(v)\phi_{\freq_1}(x-y) \phi_{\freq_2}(x-v)dxdydv\\
 \nonumber
 &+\int_{B^3}\rho^{(3)}(x,y,z)
 h_\itapers(x) h_\itapers(y)h_\itapersalt(z) h_\itapersalt(x)\phi_{\freq_1}(x-y) \phi_{\freq_2}(z-x)dxdydz\\
 \nonumber
 &+\int_{B^3}\rho^{(3)}(x,y,v)
 h_\itapers(x) h_\itapers(y)h_\itapersalt(y) h_\itapersalt(v)\phi_{\freq_1}(x-y) \phi_{\freq_2}(y-v)dxdydv\\
 \nonumber
 &+\int_{B^3}\rho^{(3)}(x,y,z)
 h_\itapers(x) h_\itapers(y)h_\itapersalt(z) h_\itapersalt(y)\phi_{\freq_1}(x-y) \phi_{\freq_2}(z-y)dxdydz\\
 \nonumber
 &+\int_{B^2}\rho^{(2)}(x,y)
 h_\itapers(x) h_\itapers(y)h_\itapersalt(x) h_\itapersalt(y)\phi_{\freq_1}(x-y) \phi_{\freq_2}(x-y)dxdy\\
 \nonumber
 &+\int_{B^2}\rho^{(2)}(x,y)
 h_\itapers(x) h_\itapers(y)h_\itapersalt(y) h_\itapersalt(x)\phi_{\freq_1}(x-y) \phi_{\freq_2}(y-x)dxdy.
\end{align}
We will be able to use the orthogonality between the data tapers $\{h_\itapersalt(x)\}$ but this will be easier in the Fourier domain where we can assume local smoothness of the Fourier transform of $\rho^{(n)}(\dots)$, and carry out the usual Taylor series. Now we start by implementing the calculations for the Poisson case, to see the mechanics.

The next term in the expansion is then
\begin{align*}
V_2&(\freq_1,\freq_2)=\E \left\{\sum_{v}
 h_\itapers^2(v) \sum^{\neq}_{x,y} h_\itapersalt(x)h_\itapersalt(y)\phi_{\freq_2}(x-y)+\sum_{x}
 h_\itapersalt^2(x) \sum^{\neq}_{v,y}h_\itapers(x)h_\itapers(y) h_\itapersalt(v)h_\itapersalt(y)\phi_{\freq_1}(x-y)\right\}\\
 &= \int_{B^3} \rho^{(3)}(x,y,v)h_\itapers^2(v) h_\itapersalt(x)h_\itapersalt(y)\phi_{\freq_2}(x-y)\;dxdydv
 +\iint_{B^2}\rho^{(2)}(x,y) h_\itapersalt(x)h_\itapersalt(y)h_\itapers^2(x)\phi_{\freq_2}(x-y)\;dxdy\\
 &+\iint_{B^2}\rho^{(2)}(x,y)h_\itapers^2(y)h_\itapersalt(x)h_\itapersalt(y)\phi_{\freq_2}(x-y)\;dxdy
 + \int_{B^3} \rho^{(3)}(x,y,v)h_\itapersalt^2(v)h_\itapers(x)h_\itapers(y)\phi_{\freq_1}(x-y)\;dxdydv\\
 &
 +\iint_{B^2}\rho^{(2)}(x,y)h_\itapersalt^2(x)h_\itapers(x)h_\itapers(y)
 \phi_{\freq_1}(x-y)\;dxdy
 +\iint_{B^2}\rho^{(2)}(x,y)h_\itapersalt^2(y)h_\itapers(x)h_\itapers(y)
 \phi_{\freq_1}(x-y)\;dxdy.
\end{align*}
The final term is in turn 
using~\eqref{eqn:Campbell}
\begin{align}
\nonumber
V_3&=\E \left[\sum_{x}
 h_\itapers^2(x) \sum_{y}
 h_\itapersalt^2(y)\right]=\E \sum_{x\neq y}
 h_\itapers^2(x)h_\itapersalt^2(y)+\E\sum_{x}
 h_\itapers^2(x)h_\itapersalt^2(x)\\
 &=\int_{B^2}\rho^{(2)}(x,y)h_\itapers^2(x)h_\itapersalt^2(y)\,dx\,dy+\lambda 
 \int_{B} h_\itapers^2(x)h_\itapersalt^2(x)\,dx.
\end{align}
We can put these terms together to determine
the covariance between the periodogram at two different wavenumbers and using two different tapers as follows:
\begin{align*}
\cov&\{I^t_\itapers(\freq_1), I^t_\itapersalt(\freq_2)\}
=\int_{B^4}
 \rho^{(4)}(x,y,z,v)
 h_\itapers(x) h_\itapers(y)h_\itapersalt(z) h_\itapersalt(v)f_{\freq_1}(x-y) f_{\freq_2}(z-v)dxdydvdz\\
 \nonumber
 &+\int_{B^3}\rho^{(3)}(x,y,v)
 h_\itapers(x) h_\itapers(y)h_\itapersalt(x) h_\itapersalt(v)f_{\freq_1}(x-y) f_{\freq_2}(x-v)dxdydv\\
 \nonumber
 &+\int_{B^3}\rho^{(3)}(x,y,z)
 h_\itapers(x) h_\itapers(y)h_\itapersalt(z) h_\itapersalt(x)f_{\freq_1}(x-y) f_{\freq_2}(z-x)dxdydz\\
 \nonumber
 &+\int_{B^3}\rho^{(3)}(x,y,v)
 h_\itapers(x) h_\itapers(y)h_\itapersalt(y) h_\itapersalt(v)f_{\freq_1}(x-y) f_{\freq_2}(y-v)dxdydv\\
 \nonumber
 &+\int_{B^3}\rho^{(3)}(x,y,z)
 h_\itapers(x) h_\itapers(y)h_\itapersalt(z) h_\itapersalt(y)f_{\k_1}(x-y) f_{\k_2}(z-y)dxdydz\\
 \nonumber
 &+\int_{B^2}\rho^{(2)}(x,y)
 h_\itapers(x) h_\itapers(y)h_\itapersalt(x) h_\itapersalt(y)f_{\k_1}(x-y) f_{\k_2}(x-y)dxdy\\
 \nonumber
 &+\int_{B^2}\rho^{(2)}(x,y)
 h_\itapers(x) h_\itapers(y)h_\itapersalt(y) h_\itapersalt(x)f_{\k_1}(x-y) f_{\k_2}(y-x)dxdy\\
 &+\int_{B^3} \rho^{(3)}(x,y,v)h_\itapers^2(v)h_\itapersalt(x)h_\itapersalt(y)f_{\k_2}(x-y)dxdydv
 +\iint_{B^2}\rho^{(2)}(x,y)h_\itapers^2(x)h_\itapersalt(x)h_\itapersalt(y)f_{\k_2}(x-y)
 dxdy\\
 &+\iint_{B^2}\rho^{(2)}(x,y)h_\itapers^2(y)h_\itapersalt(x)h_\itapersalt(y)f_{\k_2}(x-y)dxdy
 +\int_{B^3} \rho^{(3)}(x,y,v)h_\itapersalt^2(v)h_\itapers(x)h_\itapers(y)f_{\k_1}(x-y)dxdy\\
& +\iint_{B^2}\rho^{(2)}(x,y)h_\itapersalt^2(x)h_\itapers(x)h_\itapers(y)f_{\k_1}(x-y)
dxdy +\iint_{B^2}\rho^{(2)}(x,y)h_\itapersalt^2(y)h_\itapers(x)h_\itapers(y)f_{\k_1}(x-y)dxdy\\
 &+\int_{B^2}\rho^{(2)}(x,y)h_\itapers^2(x)h_\itapersalt^2(y)\,dx\,dy+\lambda 
 \int_{B} h_\itapers^2(x)h_\itapersalt^2(x)\,dx\\
 &-\left\{ \lambda + \int_{B^2} h_\itapers(x)h_\itapers(y)e^{-i{\k_1}^T(x-y)}\rtwo(x-y)dxdy \right\}\left\{ \lambda + \int_{B^2} h_\itapersalt(x)h_\itapersalt(y)e^{-i{\k_2}^T(x-y)}\rtwo(x-y)dxdy \right\}\\
 &=\int_{B^4}\rho^{(4)}(x,y,z,v)m^{(4)}(x,y,z,v)\,dxdydzdv+\int_{B^3}\rho^{(3)}(x,y,z)m^{(3)}
 (x,y,z)\,dxdydz\\
 &+\int_{B^2}\rho^{(2)}(x,y)m^{(2)}
 (x,y)\,dxdy\\
 &-\left\{ \lambda + \int_{B^2} h_\itapers(x)h_\itapers(y)e^{-i{\k_1}^T(x-y)}\rtwo(x-y)dxdy \right\}\left\{ \lambda + \int_{B^2} h_\itapersalt(x)h_\itapersalt(y)e^{-i{\k_2}^T(x-y)}\rtwo(x-y)dxdy \right\},
\end{align*}
this defining $m^{(4)}(x,y,z,v)$, $m^{(3)}(x,y,z)$
and $m^{(2)}(x,y)$, respectively.

We start by
noting that the multiplier of $\rho^{(3)}$
takes the form of
\begin{align*}
m^{(3)}(x,y,z)&= h_\itapers(x) h_\itapers(y)h_\itapersalt(x) h_\itapersalt(z)\phi_{\k_1}(x-y) \phi_{\k_2}(x-z)+h_\itapers(x) h_\itapers(y)h_\itapersalt(z) h_\itapersalt(x)\phi_{\k_1}(x-y) \phi_{-\k_2}(x-z)\\
&+h_\itapers(x) h_\itapers(y)h_\itapersalt(y) h_\itapersalt(z)\phi_{\k_1}(x-y) \phi_{\k_2}(y-z)+
h_\itapers(x) h_\itapers(y)h_\itapersalt(z) h_\itapersalt(y)\phi_{\k_1}(x-y) \phi_{-{\k_2}}(y-z)\\
&+h_\itapers^2(z)h_\itapersalt(x)h_\itapersalt(y)\phi_{\k_2}(x-y)+h_\itapersalt^2(z)h_\itapers(x)h_\itapers(y)\phi_{\k_1}(x-y)\\
&= h_\itapers(x) h_\itapers(y)h_\itapersalt(x) h_\itapersalt(z)\phi_{\k_1}(x-y)\left\{ \phi_{\k_2}(x-z)+  \phi_{-{\k_2}}(x-z)\right\}\\
&+h_\itapers(x) h_\itapers(y)h_\itapersalt(y) h_\itapersalt(z)\phi_{\k_1}(x-y)\left\{ \phi_{\k_2}(y-z)+  \phi_{-{\k_2}}(y-z)\right\}\\
&+\phi_{\k_2}(x-y)h_\itapers^2(z)h_\itapersalt(x)h_\itapersalt(y)+\phi_{\k_1}(x-y) h_\itapersalt^2(z)h_\itapers(x)h_\itapers(y)
\end{align*}
Looking at the positive term multiplying $
\rho^{(2)}(x,y)$ we can simplify to
\begin{align}
\nonumber
m^{(2)}(x,y)&=h_\itapers(x) h_\itapers(y)h_\itapersalt(x) h_\itapersalt(y)\phi_{\k_1}(x-y) \phi_{\k_2}(x-y)
+h_\itapers(x) h_\itapers(y)h_\itapersalt(y) h_\itapersalt(x)\phi_{\k_1}(x-y) \phi_{\k_2}(y-x)\\
\nonumber
&+h_\itapers^2(x)h_\itapersalt(x)h_\itapersalt(y)\phi_{\k_2}(x-y)+h_\itapers^2(y)h_\itapersalt(x)h_\itapersalt(y)
\phi_{\k_1}(x-y)+h_\itapersalt^2(x)h_\itapers(x)h_\itapers(y)\phi_{\k_1}(x-y)\\
&+\nonumber
h_\itapersalt^2(x)h_\itapers(x)h_\itapers(y)\phi_{\k_2}(x-y)+h_\itapers^2(x)h_\itapersalt^2(y)\\
\nonumber
&=h_\itapers(x) h_\itapers(y)h_\itapersalt(x) h_\itapersalt(y)f_{\k_1}(x-y)\left\{ 
f_{\k_2}(x-y)+f_{-{\k_2}}(x-y)\right\}+h_\itapers^2(x)h_\itapersalt(x)h_\itapersalt(y)\\
\nonumber
&\times \left\{f_{\k_2}(x-y)+f_{\k_1}(x-y)\right\}+h_\itapersalt^2(x)h_\itapers(x)h_\itapers(y)\left\{f_{\k_2}(x-y)+f_{\k_1}(x-y)\right\}+h_\itapers^2(x)h_\itapersalt^2(y)\\
\nonumber
&=h_\itapers(x) h_\itapers(y)h_\itapersalt(x) h_\itapersalt(y)\phi_{\k_1}(x-y)\left\{ 
\phi_{\k_2}(x-y)+\phi_{-{k_2}}(x-y)\right\}
\\ \nonumber
& + \left\{h_\itapers^2(x)h_\itapersalt(x)h_\itapersalt(y) + 
h_\itapersalt^2(x)h_\itapers(x)h_\itapers(y)\right\}  \cdot\left\{\phi_{\k_2}(x-y)+\phi_{\k_1}(x-y)\right\}+h_\itapers^2(x)h_\itapersalt^2(y).
\end{align}
Then it follows that we can simplify the final term:
\begin{align*}
\nonumber
&\left\{ \lambda + \int_{B^2} h_\itapers(x)h_\itapers(y)e^{-i{\k_1}^T(x-y)}\rtwo(x-y)dxdy \right\}\left\{ \lambda + \int_{B^2} h_\itapersalt(x)h_\itapersalt(y)e^{-i{\k_2}^T(x-y)}\rtwo(x-y)dxdy \right\}\\
&=\lambda^2+\lambda \int_{B^2} (h_\itapers(x)h_\itapers(y)\phi_{\k_1}(x-y)+h_\itapersalt(x)h_\itapersalt(y)\phi_{\k_2}(x-y))\rtwo(x-y)dxdy\\
&+ \left\{  \int_{B^2} h_\itapers(x)h_\itapers(y)\phi_{\k_1}(x-y)\rtwo(x-y)dxdy \right\}\left\{  \int_{B^2} h_\itapersalt(x)h_\itapersalt(y)\phi_{\k_2}(x-y)\rtwo(x-y)dxdy \right\}
\end{align*}
We then get for the tapered periodogram:
\begin{align*}
\cov& \{I^t_\itapers({\k_1}) , I^t_\itapersalt({\k_2})\}
=\int_{B^4}
 \rho^{(4)}(x,y,z,v)
 h_\itapers(x) h_\itapers(y)h_\itapersalt(z) h_\itapersalt(v)\phi_{\k_1}(x-y) \phi_{\k_2}(z-v)\;dxdydvdz\\
 \nonumber
 &+\int_{B^3} \rho^{(3)}(x,y,v)\left[h_\itapers(x) h_\itapers(y)h_\itapersalt(x) h_\itapersalt(v)\phi_{\k_1}(x-y)\left\{ \phi_{\k_2}(x-v)+  \phi_{-{\k_2}}(x-v)\right\}\right.\\
&+h_\itapers(x) h_\itapers(y)h_\itapersalt(y) h_\itapersalt(v)\phi_{\k_1}(x-y)\left\{ \phi_{\k_2}(y-v)+  \phi_{-{\k_2}}(y-v)\right\}\\
&\left.+\phi_{\k_2}(x-y)h_\itapers^2(v)h_\itapersalt(x)h_\itapersalt(y)+\phi_{\k_1}(x-y) h_\itapersalt^2(v)h_\itapers(x)h_\itapers(y)\right]\;
dxdydv\\
&+\int_{B^2}\rho^{(2)}(x,y)\left[ h_\itapers(x) h_\itapers(y)h_\itapersalt(x) h_\itapersalt(y)\phi_{\k_1}(x-y)\left\{ 
\phi_{\k_2}(x-y)+\phi_{-{\k_2}}(x-y)\right\}\right.\\
&+\left\{h_\itapers^2(x)h_\itapersalt(x)h_\itapersalt(y)+ h_\itapersalt^2(x)h_\itapers(x)h_\itapers(y)\right\}
\left.\left\{\phi_{\k_2}(x-y)+\phi_{\k_1}(x-y)\right\}+h_\itapers^2(x)h_\itapersalt^2(y)\right]\,dxdy\\
 &+\lambda 
 \int_{B} h_\itapers^2(x)h_\itapersalt^2(x)\,dx-\lambda^2-\lambda \int_{B^2} h_\itapers(x)h_\itapers(y)h_\itapersalt(x)h_\itapersalt(y)(\phi_{\k_1}(x-y)+\phi_{\k_2}(x-y))\rtwo(x,y)dxdy\\
 &-\left\{  \int_{B^2} h_\itapers(x)h_\itapers(y)\phi_{\k_1}(x-y)\rtwo(x,y)dxdy \right\}\cdot \left\{  \int_{B^2} h_\itapersalt(x)h_\itapersalt(y)\phi_{\k_2}(x-y)\rtwo(x,y)dxdy \right\}.
\end{align*}
This gives the covariance between the tapered periodogram with different tapers, at different wavenumbers, and it quite a useful general expression, as the following expression will show.

%\tuomas{[did not go through all these... could probably be simplified using commutative/symmetry etc shortcuts, but prolly not worth it :/ ]}
%\sofia{checked and corrected notation.}
\end{proof}

\section{Proof of Corollary~\ref{Propcovar2} Part I}
\label{sec:Pf4}
\begin{proof}
This proposition both determines the variance of a single taper periodogram, and the cross-covariance between the tapers at a fixed wavenumber $\freq$ for a Poisson Process.
If we start from the assumption of Poissonian 
(namely~\eqref{assumpPoisson}) then we find
\begin{align*}
\cov& \{I^t_\itapers({\k}) , I^t_\itapersalt({\k})\}
=\int_{B^4}
 \lambda^4
 h_\itapers(x) h_\itapers(y)h_\itapersalt(z) h_\itapersalt(v)\phi_{\k}(x-y) \phi_{\k}(z-v)\;dxdydvdz\\
 \nonumber
 &+\int_{B^3} \lambda^3\left[h_\itapers(x) h_\itapers(y)h_\itapersalt(x) h_\itapersalt(v)\phi_{\k}(x-y)\left\{ \phi_{\k}(x-v)+  \phi_{-{\k}}(x-v)\right\}\right.\\
&+h_\itapers(x) h_\itapers(y)h_\itapersalt(y) h_\itapersalt(v)\phi_{\k}(x-y)\left\{ \phi_{\k}(y-v)+  \phi_{-{\k}}(y-v)\right\}\\
&\left.+\phi_{\k}(x-y)h_\itapers^2(v)h_\itapersalt(x)h_\itapersalt(y)+\phi_{\k}(x-y) h_\itapersalt^2(v)h_\itapers(x)h_\itapers(y)\right]\;
dxdydv\\
&+\int_{B^2}\lambda^2\left[ h_\itapers(x) h_\itapers(y)h_\itapersalt(x) h_\itapersalt(y)\phi_{\k}(x-y)\left\{ 
\phi_{\k}(x-y)+\phi_{-{\k}}(x-y)\right\}\right.\\
&+\left\{h_\itapers^2(x)h_\itapersalt(x)h_\itapersalt(y)+ h_\itapersalt^2(x)h_\itapers(x)h_\itapers(y)\right\}
\left.\left\{\phi_{\k}(x-y)+\phi_{\k}(x-y)\right\}+h_\itapers^2(x)h_\itapersalt^2(y)\right]\,dxdy\\
 &+\lambda 
 \int_{B} h_\itapers^2(x)h_\itapersalt^2(x)\,dx-\lambda^2-\lambda \int_{B^2} h_\itapers(x)h_\itapers(y)h_\itapersalt(x)h_\itapersalt(y)(\phi_{\k}(x-y)+\phi_{\k}(x-y))\lambda^2dxdy\\
 &-\left\{  \int_{B^2} h_\itapers(x)h_\itapers(y)\phi_{\freq}(x-y)\lambda^2dxdy \right\}\cdot \left\{  \int_{B^2} h_\itapersalt(x)h_\itapersalt(y)\phi_{\k}(x-y)\lambda^2dxdy \right\}.
\end{align*}
We now need to massage this expression in order to understand the correlation. We first note that the last and first terms cancel exactly, leaving us with
\begin{align*}
\cov& \{I^t_\itapers({\k}) , I^t_\itapersalt({\k})\}
=\lambda^3\int_{B^3} \left[h_\itapers(x) h_\itapers(y)h_\itapersalt(x) h_\itapersalt(v)\phi_{\k}(x-y)\left\{ \phi_{\k}(x-v)+  \phi_{-{\k}}(x-v)\right\}\right.\\
&\left.+h_\itapers(x) h_\itapers(y)h_\itapersalt(y) h_\itapersalt(v)\phi_{\k}(x-y)\left\{ \phi_{\k}(y-v)+  \phi_{-{\k}}(y-v)\right\}\right]\;
dxdydv\\
&+\lambda^3\int_{B^2} \phi_{\k}(x-y)h_\itapersalt(x)h_\itapersalt(y)+\phi_{\k}(x-y) h_\itapers(x)h_\itapers(y)\;dxdy\\
&+\lambda^2\int_{B^2}\left[ h_\itapers(x) h_\itapers(y)h_\itapersalt(x) h_\itapersalt(y)\phi_{\k}(x-y)\left\{ 
\phi_{\k}(x-y)+\phi_{-{\k}}(x-y)\right\}\right.\\
&+\left\{h_\itapers^2(x)h_\itapersalt(x)h_\itapersalt(y)+ h_\itapersalt^2(x)h_\itapers(x)h_\itapers(y)\right\}
\left.2\phi_{\k}(x-y)\right]\,dxdy\\
 &+\lambda 
 \int_{B} h_\itapers^2(x)h_\itapersalt^2(x)\,dx-2\lambda^3 \int_{B^2} h_\itapers(x)h_\itapers(y)h_\itapersalt(x)h_\itapersalt(y)\phi_{\k}(x-y)\;dxdy\\
 &=\sum_{j=1}^7 T_j(\k),
\end{align*}
this defining the sequence $\{T_j(\k)\}$.
We need some extra relationships to simplify these expressions. We note that
\begin{align}
\label{rel1}
&\int_B  h_\itapers(x)h_\itapersalt(x)\;dx=\delta_{\itapers\itapersalt}\\
&\int_B h_\itapers(x) \phi_{2\k}(x)\;dx\approx 0,\quad \k>b_h>0,
\label{rel2}\\
&\int H_\itapers(\k'-2\k)H_\itapersalt(\k')\;d\k'\approx 0\quad \k>b_h>0,
\end{align}
where $b_h$ is the bandwidth of the window defined in~\eqref{eqn:defbandwidth}. We find
\begin{align}
\nonumber
T_1(\k)&=\lambda^3\int_{B^3} h_\itapers(x) h_\itapers(y)h_\itapersalt(x) h_\itapersalt(v)\phi_{\k}(x-y)\left\{ \phi_{\k}(x-v)+  \phi_{-{\k}}(x-v)\right\}\;dxdydv\\
&=\lambda^3H_\itapers(-\k)H_\itapersalt(-\k)\int H_\itapers(\k'-2\k)H_\itapersalt(\k')\;d\k' +\lambda^3\delta_{\itapers\itapersalt}H_\itapers(\k)H_\itapersalt(-\k).
\end{align}
The next term is
\begin{align}
\nonumber
T_2(\k)&=\lambda^3\int_{B^3}h_\itapers(x) h_\itapers(y)h_\itapersalt(y) h_\itapersalt(v)\phi_{\k}(x-y)\left\{ \phi_{\k}(y-v)+  \phi_{-{\k}}(y-v)\right\}\;
dxdydv\\
&=\lambda^3\delta_{\itapers\itapersalt}H_\itapers(\k)H_\itapersalt(-\k)+\lambda^3H_\itapers(-\k)H_\itapersalt(-\k)\int H_\itapers(\k'-2\k)H_\itapersalt(\k')\;d\k'. 
\end{align}
The following term is
{\small \begin{align*}
T_3(\k)&=\lambda^3\int_{B^2} \phi_{\k}(x-y)h_\itapersalt(x)h_\itapersalt(y)+\phi_{\k}(x-y) h_\itapers(x)h_\itapers(y)\;dxdy\\
&=\lambda^3 H_\itapersalt(\k)H_\itapersalt(-\k)+\lambda^3 H_\itapers(\k)H_\itapers(-\k).
\end{align*}
The next term is
\begin{align*}
T_4(\k)&=\lambda^2\int_{B^2}\left[ h_\itapers(x) h_\itapers(y)h_\itapersalt(x) h_\itapersalt(y)\phi_{\k}(x-y)\left\{ 
\phi_{\k}(x-y)+\phi_{-{\k}}(x-y)\right\}\right]dxdy\\
&=\lambda^2\left|\int H_\itapers(\k'-2\k)H_\itapersalt(\k')\;d\k'\right|^2+
\lambda^2\delta_{\itapers\itapersalt}.
\end{align*}
The next term then takes the form
\begin{align*}
&T_5(\k)=2\lambda^2\int_{B^2}\left\{h_\itapers^2(x)h_\itapersalt(x)h_\itapersalt(y)+ h_\itapersalt^2(x)h_\itapers(x)h_\itapers(y)\right\}
\phi_{\k}(x-y)\,dxdy\\
&=2\lambda^2\left\{H_\itapersalt(-\k) \int_{{\mathbb{R}}^{2d}}H_\itapers(\k'')H_\itapers(\k'-\k'') H_\itapersalt(\k-\k')\,d\k'd\k''+H_\itapers(-\k) \int_{{\mathbb{R}}^{2d}}H_\itapersalt(\k'')H_\itapersalt(\k'-\k'') H_\itapers(\k-\k')\,d\k'd\k''\right\}.
\end{align*}}
The next term takes the form of
\begin{align*}
T_6(\k)&=\lambda 
 \int_{B} h_\itapers^2(x)h_\itapersalt^2(x)\,dx\sim \frac{\lambda}{|B|}.
\end{align*}
And the final term is negative:
\begin{align*}
T_7(\k)&=-2\lambda^3 \int_{B^2} h_\itapers(x)h_\itapers(y)h_\itapersalt(x)h_\itapersalt(y)\phi_{\k}(x-y)\;dxdy\\
&=-2\lambda^3\left|\int H_\itapers(\k-\k') H_\itapersalt(\k') \;d\k'\right|^2.
\end{align*}
As we have assumed the wavenumber $\k$ is sufficiently large
both $|H_\itapers(\k)|\approx |H_\itapersalt(\k)|\approx 0$. This means the only surviving contributions are
\[\cov\{I^t_\itapers({\k}) , I^t_\itapersalt({\k})\}\approx \lambda 
 \int_{B} h_\itapers^2(x)h_\itapersalt^2(x)\,dx+\lambda^2\delta_{\itapers\itapersalt}.\]
This yields the stated result. We can derive the $o(1)$ terms formally should we wish by quantifying the leakage in $\bside$, the length of a side.
%\sofia{checked and corrected notation.}
\end{proof}
%GOT TO
\section{Proof of Corollary~\ref{Propcovar3} Part II}
\label{sec:Pf5}
\begin{proof}
%\tuomas{[Can we use Theorem V.1 for a shortcut?]}
If we start from the assumption of 
Poissonian (namely~\eqref{assumpPoisson}) then
\begin{align*}
\cov& \{I^t_\itapers({\k_1}) , I^t_\itapers({\k_2})\}
=\int_{B^4}
 \lambda^4
 h_\itapers(x) h_\itapers(y)h_\itapers(z) h_\itapers(v)\phi_{\k_1}(x-y) \phi_{\k_2}(z-v)\;dxdydvdz\\
 \nonumber
 &+\int_{B^3} \lambda^3\left[h_\itapers^2(x) h_\itapers(y)h_\itapers(v)\phi_{\k_1}(x-y)\left\{ \phi_{\k_2}(x-v)+  \phi_{-{\k_2}}(x-v)\right\}\right.\\
&+h_\itapers(x) h_\itapers^2(y) h_\itapers(v)\phi_{\k_1}(x-y)\left\{ \phi_{\k_2}(y-v)+  \phi_{-{\k_2}}(y-v)\right\}\\
&\left.+\phi_{\k_2}(x-y)h_\itapers^2(v)h_\itapers(x)h_\itapers(y)+
\phi_{\k_1}(x-y) h_\itapers^2(v)h_\itapers(x)h_\itapers(y)\right]\;
dxdydv\\
&+\int_{B^2}\lambda^2\left[ h_\itapers(x) h_\itapers(y)h_\itapers(x) h_\itapers(y)\phi_{\k_1}(x-y)\left\{ 
\phi_{\k_2}(x-y)+\phi_{-{\k_2}}(x-y)\right\}\right.\\
&+\left\{h_\itapers^2(x)h_\itapers(x)h_\itapers(y)+ h_\itapers^2(x)h_\itapers(x)h_\itapers(y)\right\}
\left.\left\{\phi_{\k_2}(x-y)+\phi_{\k_1}(x-y)\right\}+h_\itapers^2(x)h_\itapers^2(y)\right]\,dxdy\\
 &+\lambda 
 \int_{B} h_\itapers^4(x)\,dx-\lambda^2-\lambda \int_{B^2} h_\itapers^2(x)h_\itapers^2(y)(\phi_{\k_1}(x-y)+
 \phi_{\k_2}(x-y))\lambda^2 dxdy\\
 &-\left\{  \int_{B^2} h_\itapers(x)h_\itapers(y)\phi_{\k_1}(x-y)\lambda^2dxdy \right\}\cdot \left\{  \int_{B^2} h_\itapers(x)h_\itapers(y)\phi_{\k_2}(x-y)\lambda^2dxdy \right\}
 \end{align*}
 Continuing on with the calculations we find:
 \begin{align*}
\cov& \{I^t_\itapers({\k_1}) , I^t_\itapers({\k_2})\} =\int_{B^3} \lambda^3\left[h_\itapers^2(x) h_\itapers(y)h_\itapers(v)\phi_{\k_1}(x-y)\left\{ \phi_{\k_2}(x-v)+  \phi_{-{\k_2}}(x-v)\right\}\right.\\
&+h_\itapers(x) h_\itapers^2(y) h_\itapers(v)\phi_{\k_1}(x-y)\left\{ \phi_{\k_2}(y-v)+  \phi_{-{\k_2}}(y-v)\right\}\\
&\left.+\phi_{\k_2}(x-y)h_\itapers^2(v)h_\itapers(x)h_\itapers(y)+
\phi_{\k_1}(x-y) h_\itapers^2(v)h_\itapers(x)h_\itapers(y)\right]\;
dxdydv\\
&+\int_{B^2}\lambda^2\left[ h_\itapers^2(x) h_\itapers^2(y) \phi_{\k_1}(x-y)\left\{ 
\phi_{\k_2}(x-y)+\phi_{-{\k_2}}(x-y)\right\}\right.\\
&+\left\{h_\itapers^3(x)h_\itapers(y)+ h_\itapers^3(x)h_\itapers(y)\right\}
\left.\left\{\phi_{\k_2}(x-y)+\phi_{\k_1}(x-y)\right\}\right]\,dxdy\\
 &+\lambda 
 \int_{B} h_\itapers^4(x)\,dx-\lambda^3 \int_{B^2} h_\itapers^2(x)h_\itapers^2(y)(\phi_{\k_1}(x-y)+\phi_{\k_2}(x-y)) dxdy.
\end{align*}
We like in the previous case split this into seven parts
$\{\widetilde{T}_j(\k)\}$. We find that
\begin{align*}
\widetilde{T}_1(\k)&=\lambda^3\int_{B^3} h_\itapers^2(x) h_\itapers(y)h_\itapers(v)\phi_{\k_1}(x-y)\left\{ \phi_{\k_2}(x-v)+  \phi_{-{\k_2}}(x-v)\right\}\;dxdydv\\
&=\lambda^3 H_\itapers(-\k_1)H_\itapers(-\k_2)\int_{{\mathbb{R}}^d}H_\itapers(\k')H_\itapers(\k_1+\k_2-\k')\;d\k'\\
&+
\lambda^3 H_\itapers(-\k_1)H_\itapers(\k_2)\int_{{\mathbb{R}}^d}H_\itapers(\k')H_\itapers(\k_1-\k_2-\k')\;d\k'.
\end{align*}
If $\k_1-\k_2,\k_1+\k_2\neq 0$ and is some fixed number, then this becomes negligible. The next term is
\begin{align*}
\widetilde{T}_2(\k)&=\lambda^3\int_{B^3}
h_\itapers(x) h_\itapers^2(y) h_\itapers(v)\phi_{\k_1}(x-y)\left\{ \phi_{\k_2}(y-v)+  \phi_{-{\k_2}}(y-v)\right\}dxdydv\\
&=\lambda^3 H_\itapers(\k_1)H_\itapers(-\k_2)\int_{{\mathbb{R}}^d}H_\itapers(\k')H_\itapers(\k_2-\k_1-\k')\;d\k'
\\
&+\lambda^3 H_\itapers(\k_1)H_\itapers(-\k_2)\int_{{\mathbb{R}}^d}H_\itapers(\k')H_\itapers(\k_2+\k_1-\k')\;d\k'.
\end{align*}
If $\k_1-\k_2,\k_1+\k_2\neq 0$ and is some fixed number, then this becomes negligible. The next term is
\begin{align*}
\widetilde{T}_3(\k)&=\lambda^3\int_{B^3} \left[\phi_{\k_2}(x-y)h_\itapers^2(v)h_\itapers(x)h_\itapers(y)+\phi_{\k_1}(x-y) h_\itapers^2(v)h_\itapers(x)h_\itapers(y)\right]\;
dxdydv\\
&=\lambda^3\left|H_{\itapers}(\k_2)\right|^2+\lambda^3\left|H_{\itapers}(\k_1)\right|^2.
\end{align*}
As $\k_1,\k_2>b_h$ the bandwidth this contribution becomes negligible. The next contribution in the expression takes the form
\begin{align*}
\widetilde{T}_4(\k)&=\int_{B^2}\lambda^2\left[ h_\itapers^2(x) h_\itapers^2(y) \phi_{\k_1}(x-y)\left\{ 
\phi_{\k_2}(x-y)+\phi_{-{\k_2}}(x-y)\right\}\right]dxdy\\
&=\lambda^2\left|\int_{{\mathbb{R}}^d}H_\itapers(\k')H_\itapers(\k_1+\k_2-\k')\;d\k'\right|^2+\lambda^2\left|\int_{{\mathbb{R}}^d}H_\itapers(\k')H_\itapers(\k_1-\k_2-\k')\;d\k'\right|^2.
\end{align*}
Then the next term is
\begin{align*}
\widetilde{T}_5(\k)&=2\lambda^2\int_{B^2}h_\itapers^3(x)h_\itapers(y)
\left.\left\{\phi_{\k_2}(x-y)+\phi_{\k_1}(x-y)\right\}\right]\,dxdy.
\end{align*}
Figuring out the size of this contribution is a little bit more complex. 
The integral will do a Fourier transform in $y$ of $h_\itapers(y)$ in $\k_1$ and $\k_2$, and also the conjugate will be taken. The Fourier transform in $y$ will be supported when $|\k_1|<b_h$ and $|\k_2|<b_h$, but will not be supported when either magnitude gets too large.  The second Fourier transform in $x$ will be somewhat more spread--this is because we are Fourier transforming $h_\itapers^3(x)$, which will be a third order convolution once Fourier transformed. But as long as we can use the joint concentration in $|\k_1|<b_h$ and $|k_2|<b_h$ this will still be negligible. 

The next term is given by
\begin{align*}
\widetilde{T}_6(\k)&=\lambda 
 \int_{B} h_\itapers^4(x)\,dx=\lambda \|h_\itapers\|_4^4=o(1).
\end{align*}
To see why this is true, note that if $h_\itapers(x)$ is constant then we can easily calculate the higher order norms. We find by using the Cauchy--Schwarz inequality that if $p\neq q$ that
\begin{align}
\nonumber
T_6^2(\k)&=\left[\int_B h_\itapers^2(x) h_\itapersalt^2(x)\;dx\right]^2\\
\nonumber
&\leq \int_B h_\itapers^4(x) \int_B h_\itapersalt^4(x)\\
\Rightarrow 0\leq \int_B h_\itapers^2(x) h_\itapersalt^2(x)\;dx&\leq \max\{\int_B h_\itapers^4(x), \int_B h_\itapersalt^4(x)\}.
\end{align}
We also note that
\begin{align}
1^2&=\left[\int_B 1\cdot h^2(x)\;dx\right]^2 \leq \int_B 1^2\;dx \cdot \int_B h^4(x)\;dx\\
\Rightarrow &\frac{1}{|B|}\leq \int_B h^4(x)\;dx.
\end{align}
We now need to determine an upper bound. We then look at
\begin{align}
\left[\int_B h^4(x)\;dx\right]^2&\leq \int_B 1^2\;dx \cdot \int_B h^8(x)\;dx=|B|\cdot \int_B h^8(x)\;dx.
\end{align}
Then we have
\begin{equation}
\frac{1}{|B|}\leq \int_B h^4(x)\;dx\leq \sqrt{|B|\cdot \int_B h^8(x)\;dx}.
\end{equation}
We can for any taper $h(x)$ calculate $\int_B h^4(x)\;dx$ explicitly. For the constant function we get
$\frac{1}{|B|}$. For tapers well--concentrated we would expect a similar decrease, but for any choice of taper we can calculate the value of the 4th norm explicitly.

Thus we understand a bit more about this term. Moving on to the next aspect of the computation,
finally, 
\begin{align*}
\widetilde{T}_7(\k)&=-\lambda^3 \int_{B^2} h_\itapers^2(x)h_\itapers^2(y)(\phi_{\k_1}(x-y)+\phi_{\k_2}(x-y)) dxdy\\
&=-\lambda^3\left|\int_{{\mathbb{R}}^d}H_\itapers(\k')H_\itapers(\k_1-\k')\;d\k'\right|^2-\lambda^3\left|\int_{{\mathbb{R}}^d}H_\itapers(\k')H_\itapers(\k_2-\k')\;d\k'\right|^2.
\end{align*}
This shows each individual contribution as long as 
$|\k_1-\k_2|,|\k_1+\k_2|>b_h$. This completes the proof. 
%\sofia{checked.}
\end{proof}

\section{Proof of Proposition~\ref{Propcovar3b}}\label{proofcovar3b}
\begin{proof}
Assume $X$ satisfies the assumptions given for~\eqref{CLT} and~\eqref{CLT2}. Note from Proposition~\ref{DFTprop} that 
\begin{align}
    \nonumber
 0\leq   \var \left\{ J_h(\freq)\right\}     
        &=\lambda+\lambda^2\int_{R^d}\widetilde{f}^{(2)}(\freq') \left|H(\freq-\freq')\right|^2dw' \\
        \nonumber
        &\leq \lambda+\lambda^2 \|\widetilde{f}^{(2)}\|_0
        \int_{R^d} \left|H(\freq-\freq')\right|^2dw'\\
        &=\lambda+\lambda^2 \|\widetilde{f}^{(2)}\|_0<\infty.
    \end{align}
We can then deduce from~\cite[Theorem 6.1]{dasgupta2008asymptotic}
that $\tilde{J}_h(\freq)$ is uniformly integrable. 
However to be able to compute the covariance of the periodogram from the convergence of law to the Gaussian, then 
we need to show that $|\tilde{J}_\itapers(\freq)|^4$, or even 
$|\tilde{J}_\itapers(\freq_1)|^2|\tilde{J}_\itapersalt(\freq_2)|^2$
are uniformly integrable. We now apply 
~\cite[Theorem 6.2]{dasgupta2008asymptotic} and assume $\|\widetilde{f}^{(2)}\|_0$,  $\|\widetilde{f}^{(3)}\|_0$, $\|\widetilde{f}^{(4)}\|_0$,  $\|\widetilde{f}^{(5)}\|_0$ and
$\|\widetilde{f}^{(6)}\|_0$ are all finite which assures $|\tilde{J}_\itapers(\freq)|^4$ and 
$|\tilde{J}_\itapers(\freq_1)|^2|\tilde{J}_\itapersalt(\freq_2)|^2$ are uniformly integrable. We can then deduce that as $\tilde{J}_\itapers(\freq)$ has converged in law to a Gaussian random variable, the moments of $\tilde{J}_\itapers(\freq)$ can be computed from the Gaussian law.

It follows that Isserlis'~\cite{isserlis1918ona} theorem can be applied by using~\cite[Theorem 6.2]{dasgupta2008asymptotic} and so  
\begin{align*}
\cov& \{\widetilde{I}^t_\itapers(\k_1) , \widetilde{I}^t_\itapersalt(\k_2)\}=\E\{
\widetilde{J}_\itapers(\k_1)\widetilde{J}^\ast_\itapers(\k_1)\widetilde{J}_\itapersalt(\k_2)\widetilde{J}^\ast_\itapersalt(\k_2)\}+o(1)-\E\{
\widetilde{J}_\itapers(\k_1)\widetilde{J}^\ast_\itapers(\k_1)\}\E\{\widetilde{J}_\itapersalt(\k_2)\widetilde{J}^\ast_\itapersalt(\k_2)\}\\
&=\E\{
\widetilde{J}_\itapers(\k_1)\widetilde{J}_\itapersalt(\k_2)\}\E\{\widetilde{J}^\ast_\itapers(\k_1)\widetilde{J}^\ast_\itapersalt(\k_2)\}
+\E\{
\widetilde{J}_\itapers(\k_1)\widetilde{J}^\ast_\itapersalt(\k_2)\}\E\{\widetilde{J}^\ast_\itapers(\k_1)\widetilde{J}_\itapersalt(\k_2)\}+o(1)\\
&=o(1)+\left|\E\{
\widetilde{J}_\itapers(\k_1)\widetilde{J}^\ast_\itapersalt(\k_2)\}\right|^2.
\end{align*}
We note that the same sort of calculations as 
Proposition~\ref{DFTprop} can be applied and so for $\freq_1\neq \freq_2$
\begin{align*}
\E\{
\widetilde{J}_\itapers(\k_1)\widetilde{J}^\ast_\itapersalt(\k_2)\}  &=\E\{
({J}_\itapers(\k_1)-\lambda H_\itapers(\k_1))({J}^\ast_\itapersalt(\k_2)-\lambda H^\ast_\itapersalt(\k_1) )\}\\
&=\E\{
{J}_\itapers(\k_1){J}^\ast_\itapersalt(\k_2)\}-\lambda^2 H_\itapers(\k_1)H^\ast_\itapersalt(\k_2).
\end{align*}
We now calculate the covariance as
\begin{align}
\nonumber
    \cov \left\{ {J}_\itapers(\k_1),{J}_\itapersalt(\k_2)\right\}
    =& \E\{
{J}_\itapers(\k_1){J}^\ast_\itapersalt(\k_2)\}-\lambda^2 H_\itapers(\k_1)H^\ast_\itapersalt(\k_2)\\
    =&\lambda\delta_{\itapers\itapersalt}\delta_{\freq_1 \freq_2} +\iint_{R^d\times R^d}{\rho}^{(2)}(x-y)h_\itapers(x) e^{-2\pi i \freq_1\cdot  x} h^\ast_\itapersalt(y)e^{2\pi i \freq_2\cdot  y}\,dx\,dy\nonumber \\ 
    &-\lambda^2 H_\itapers(\k_1)H^\ast_\itapersalt(\k_2).
    \label{cov:expr}
\end{align}
We now define the renormalised quantity
\begin{equation}
    \widetilde{\rho}^{(2)}(z)=\frac{ {\rho}^{(2)}(z)-\lambda^2}{\lambda^2},\quad z\in{\mathbb{R}}^d.
\end{equation}
The expression in~\eqref{cov:expr} then can be simplified to
\begin{align}
\nonumber
    \cov &\left\{  {J}_\itapers(\k_1),{J}_\itapersalt(\k_2)\right\}\\
    %=\lambda\delta_{\itapers\itapersalt} \delta_{\freq_1 \freq_2}+\iint_{R^d\times R^d}{\rho}^{(2)}(x-y)h_\itapers(x)e^{-2\pi i \freq_1\cdot  x}h^\ast_\itapersalt(y)e^{2\pi i \freq_2\cdot  y}\,dx\,dy-\lambda^2 H_\itapers(\k_1)H^\ast_\itapersalt(\k_2)\\
    &= \nonumber
    \lambda \delta_{\itapers\itapersalt}\delta_{\freq_1 \freq_2}+\lambda^2\iint_{R^d\times R^d}\left(\widetilde{\rho}^{(2)}(x-y)+1\right)h_\itapers(x)e^{-2\pi i \freq_1\cdot  x}h^\ast_\itapersalt(y)e^{2\pi i \freq_2\cdot  y}\,dx\,dy-\lambda^2 H_\itapers(\k_1)H^\ast_\itapersalt(\k_2)\\
    &=\lambda\delta_{\itapers\itapersalt}\delta_{\freq_1 \freq_2}+\lambda^2\iint_{R^d\times R^d}\widetilde{\rho}^{(2)}(x-y)h_\itapers(x)e^{-2\pi i \freq_1\cdot  x}h^\ast_\itapersalt(y)e^{2\pi i \freq_2\cdot  y}\,dx\,dy\nonumber \\
    &=\lambda\delta_{\itapers\itapersalt}\delta_{\freq_1 \freq_2}+ \lambda^2\cov_1 \left\{ {J}_\itapers(\k_1),{J}_\itapersalt(\k_2)\right\},
    \end{align}
    where we define
    \begin{equation}
        \cov_1 \left\{ {J}_\itapers(\k_1),{J}_\itapersalt(\k_2)\right\}
        \equiv \iint_{R^d\times R^d}\widetilde{\rho}^{(2)}(x-y)h_\itapers(x)e^{-2\pi i \freq_1\cdot  x}h^\ast_\itapersalt(y)e^{2\pi i \freq_2\cdot  y}\,dx\,dy.
    \end{equation}
    To simplify this expression we note that
    \begin{align}
    \nonumber
   \cov_1 \left\{  {J}_\itapers(\k_1),{J}_\itapersalt(\k_2)\right\}     &=\iint_{R^d\times R^d} 
        \widetilde{\rho}^{(2)}(x-y)h_\itapers(x)e^{-2\pi i \freq_1\cdot  x}h^\ast_\itapersalt(y)e^{2\pi i \freq_2\cdot  y}\,dx\,dy\\
        \nonumber
        &=\iint_{R^d\times R^d} 
        \int_{R^d}\widetilde{f}^{(2)}(\k')e^{2\pi i (x-y)\cdot \freq'}d\freq' \cdot h_\itapers(x)e^{-2\pi i \freq_1\cdot  x}h^\ast_\itapersalt(y)e^{2\pi i \freq_2\cdot  y}\,dx\,dy\\
        \nonumber
        &=\iint_{R^d\times R^d} 
        \int_{R^d}\widetilde{f}^{(2)}(\k')e^{-2\pi i x\cdot(\freq_1-\freq')}e^{2\pi i y\cdot(\freq_2-\freq')}d\freq' \cdot h_\itapers(x)h^\ast_\itapersalt(y)\,dx\,dy\\
        &=\iint_{R^d}\widetilde{f}^{(2)}(\freq') H_\itapers(\freq_1-\freq') H^\ast_\itapersalt(\freq_2-\freq')d\freq' .
    \end{align}
    %\sofia{Checked.}
    Note that it is not dimensionally contradictory to Theorem~\ref{Propcovar} as the periodogram is the modulus square of the Fourier transform.
    We note that as $H_\itapers(\freq)$ has concentrated support we can apply similar arguments to those of Appendix D, as will follow.
    \end{proof}
    
\section{Proof of Lemma~\ref{Lemmasimplify}}\label{proofcovar4b}

\begin{Lemma}\label{Lemmasimplify}
Assume $X$ satisfies the assumptions given for~\eqref{CLT} and that $\|\widetilde{f}^{(j)}\|_0<\infty$ for $j=2,3,4,5,6$. 
Assume the two multitapers $h_\itapers(x)$ and $h_\itapersalt(x)$ are orthogonal and are well concentrated on the compact set ${\cal W}\subset\Rd$ with {{length scale $\bside$ so that for some chosen $\epsilon_\bside=o(1/\bside)$
\[\int_{\cal W}|H_\itapers(\freq)|^2\;d\freq=1-\epsilon_\bside.\]}}
Then assuming $\freq_1,\freq_2\in {\cal K}_\ntapers(B_\square(\bm{l}))$ 
\begin{align}
    \nonumber
        \cov_1 \left\{ {J}_\itapers(\k_1),{J}_\itapersalt(\k_2)\right\}
        &\equiv\int_{R^d}\widetilde{f}^{(2)}(\freq') H_\itapers(\freq_1-\freq') H^\ast_\itapersalt(\freq_2-\freq')d\freq'+o(1) \\
        &=\widetilde{f}^{(2)}(\freq_1)\delta_{\itapersalt\itapers}
   \delta_{\freq_1 \freq_2}+{\cal O}(1/\bside).
    \end{align}
\end{Lemma}

\begin{proof}
We assume that for a choice of $\varepsilon_l$ we can define a wavenumber region ${\cal W}$
\begin{equation}
\int_{{\cal W}}\left| H_\itapers(\freq)\right|^2\;d\freq=1-\varepsilon_{\bside},\quad {\mathrm{where}}\quad \varepsilon_{\bside}=o(1).
\end{equation}
We also assume that $\tilde{f}(\freq)$ is upper bounded by $\|\widetilde{f}\|_0$. We then have
\begin{align}
\cov_1 \left\{  {J}_\itapers(\k_1),{J}_\itapersalt(\k_2)\right\}   &= \int_{\mathbb{R}^d\backslash {\cal W}}
H_\itapers(\freq_1-\freq') H^\ast_\itapersalt(\freq_2-\freq')\widetilde{f}(\freq')\;d\freq'
+\int_{ {\cal W}}
H_\itapers(\freq_1-\freq') H^\ast_\itapersalt(\freq_2-\freq')\widetilde{f}(\freq')\;d\freq'.
\end{align}
We note that
\[\left|\int_{\mathbb{R}^d\backslash {\cal W}}
H_\itapers(\freq_1-\freq') H^\ast_\itapersalt(\freq_2-\freq')\widetilde{f}(\freq')\;d\freq'\right|
\leq \|\widetilde{f}\|_0\{1-\{ 1-\varepsilon_{\bside}\}\}.\]
We can yet again utilise the Taylor expansion of
\eqref{tildetaylor} inside ${\cal W}$.
We find
\begin{align}
\nonumber
&\cov_1 \left\{  {J}_\itapers(\k_1),{J}_\itapersalt(\k_2)\right\} =\int_{ {\cal W}}
H_\itapers(\freq_1-\freq') H^\ast_\itapersalt(\freq_2-\freq')\widetilde{f}(\freq')\;d\freq'
\left\{1+o(1)\right\}\\
\nonumber
&=\int_{ {\cal W}}
H_\itapers(\freq_1-\freq') H^\ast_\itapersalt(\freq_2-\freq')\left\{\widetilde{f}(\freq_1)+\nabla \widetilde{f}(\freq_1)^T\left(\freq'-\freq_1\right)+
\frac{1}{2}\left(\freq'-\freq_1\right)^T\tilde{\mathbf{H}}_f(\freq'')\left(\freq'-\freq_1\right)\right\}\;d\freq'
\left\{1+o(1)\right\}\\
\nonumber&=\delta_{\freq_1 \freq_2}\int_{ {\cal W}}
H_\itapers(\freq_1-\freq') H^\ast_\itapersalt(\freq_1-\freq')\widetilde{f}(\freq_1)
\;d\freq'+
\int_{ {\cal W}}H_\itapers(\freq_1-\freq') H^\ast_\itapersalt(\freq_2-\freq')
\\
&\nonumber \times\left\{\nabla \widetilde{f}(\freq_1)^T\left(\freq'-\freq_1\right)+
\frac{1}{2}\left(\freq'-\freq_1\right)^T\tilde{\mathbf{H}}_f(\freq'')\left(\freq'-\freq_1\right)\right\}\;d\freq'
\left\{1+o(1)\right\}
\\
&=\widetilde{f}(\freq_1)\delta_{\itapersalt\itapers}\delta_{\freq_1 \freq_2}+{\cal O}(1/\sqrt{\bside})+
\int_{ {\cal W}}H_\itapers(\freq_1-\freq') H^\ast_\itapersalt(\freq_2-\freq')\frac{1}{2}\left(\freq'-\freq_1\right)^T\tilde{\mathbf{H}}_f(\freq'')
\left(\freq'-\freq_1\right)\;d\freq'.
\end{align}
We note that
\begin{align*}
&\left\|\int_{ {\cal W}}H_\itapers(\freq_1-\freq') H^\ast_\itapersalt(\freq_2-\freq')\frac{1}{2}\left(\freq'-\freq_1\right)^T\tilde{\mathbf{H}}_f(\freq'')
\left(\freq'-\freq_1\right)\;d\freq' \right\|^2\\
&\leq |\tilde{\mathbf{H}}_f(\freq'')|
\int_{ {\cal W}}H_\itapers(\freq_1-\freq') H^\ast_\itapersalt(\freq_2-\freq')\frac{1}{2}\left(\freq'-\freq_1\right)^T
\left(\freq'-\freq_1\right)\;d\freq'\\
&=|\tilde{\mathbf{H}}_f(\freq'')|
\int_{ {\cal W}}H_\itapers(\freq_1-\freq') H^\ast_\itapersalt(\freq_2-\freq') w^2 d\freq=
{\cal O}(1/\bside).
\end{align*}
Then
\begin{align*}
\cov \left\{  {J}_\itapers(\k_1),{J}_\itapersalt(\k_2)\right\}
&=\lambda\delta_{\itapers\itapersalt}\delta_{\freq_1 \freq_2}  +\lambda^2\cov_1 \left\{  {J}_\itapers(\k_1),{J}_\itapersalt(\k_2)\right\}\\
&=\lambda\delta_{\itapers\itapersalt} \delta_{\freq_1 \freq_2} +\lambda^2\widetilde{f}(\freq)\delta_{\itapersalt\itapers}\delta_{\freq_1 \freq_2}+{\cal O}(1/\sqrt{\bside})+
{\cal O}(1/\bside).
\end{align*}
\end{proof}

%\section{Proof of Proposition~\ref{BesselE0}}\label{proofBesselE0}
%\begin{proof}
%Note that by direct computation \tuomas{[Proof of Pro VII.1. given in text; Also $I_D$ here but $\bar I_0$ in text?]}
%\begin{align}
%\nonumber
%    \E\left\{\bar I_0(|\freq|)\right\}
%    &=\lambda+\left|B\right|^{-1}
%    \int_{B^2}\rho^{(2)}(\|x-y\|)G(w\|x-y\|)\,dx\,dy\\
%    \nonumber
%    &=\lambda+\left|B\right|^{-1}\int \left|B\cap %B_x\right|\rho^{(2)}(\|x\|)G(w\|x\|)\,dx\\
%    \nonumber
%    &=\lambda+\left|B\right|^{-1}\int_0^\infty
%    \rho^{(2)}(r) r^{d-1}G(w %r)\int_{S^{d-1}}\left|B\cap B_{ru}\right|\,du\,dr\\
%    &\nonumber =
%    \lambda+d\cdot  \left|B\right|^{-1}\int_0^\infty
%    \rho^{(2)}(r) r^{d-1}G(w r) %\overline{\nu}_B(r)\,dr\\
%    &= \nonumber
%     \lambda+d \left|B\right|^{-1}\int_0^\infty
 %   \gamma(r) r^{d-1}G(w r) \overline{\nu}_B(r)\,dr+\lambda^2 d\cdot b_d \left|B\right|^{-1}H(\freq).
%\end{align}

%\end{proof}

\section{Proof of Proposition~\ref{BesselE}}\label{proofBesselE}
\begin{proof}
An interesting question is what if we use Diggle's estimator even if the point process is not isotropic. We recall that the estimator takes in 2D the form 
\begin{equation}
\bar{I}_0(|\freq|)=\hat{\lambda}+|B|^{-1}\sum_{x,y\in X\cap B}^{\neq}J_0\left(2\pi |\freq| \cdot \|x-y\|\right).
\end{equation}
We can still compute the estimator for any observed point-process $X$, even if $X$ was not an isotropic process.
The estimator $\bar{I}_0(|\freq|)$ has expectation
\begin{align}
\label{1steq}
\E \{\bar{I}_0(|\freq|)\}&={\lambda}+|B|^{-1}\sum_{x,y\in X\cap B}^{\neq} \E J_0\left(2\pi |\freq|\cdot \|x-y\|\right)\\
\nonumber
&={\lambda}+|B|^{-1}\int_B\int_{B_{-x}}  J_0\left(2\pi |\freq|\cdot |z|\right)\rho^{(2)}(z)\,dz\,dx\\
&={\lambda}+|B|^{-1}\int_B\int_{B_{-x}}  J_0\left(2\pi |\freq|\cdot |z|\right)\left\{\rho^{(2)}(z)-\lambda^2\right\}\,dz\,dx+\lambda^2|B|^{-1}\int_B\int_{B_{-x}}  J_0\left(2\pi |\freq|\cdot |z|\right)\,dz\,dx
\nonumber\\
&={\lambda}+|B|^{-1}\int_{{\mathbb{R}}^d}  
|B\times B_{-z}|\cdot
J_0\left(2\pi |\freq|\cdot |z|\right)\left\{\rho^{(2)}(z)-\lambda^2\right\}\,dz+\lambda^2|B|^{-1}\int_B|B\times B_{-z}| \cdot  J_0\left(2\pi |\freq|\cdot |z|\right)\,dz
\nonumber\\
&=\E \bar{I}_0^{(1)}(|\freq|)+\E\bar{I}_0^{(2)}(|\freq|)+\E\bar{I}_0^{(3)}(|\freq|),
\label{expect-split}
\end{align}
the latter defining the form of these three contributions.

Just like before we shall explicitly demonstrate the effects of this convolution. We have that the Fourier transform of the Bessel function is
\begin{align}
\nonumber
{\cal F}\left\{J_0\left(2\pi |\freq|\cdot |z|\right) \right\}(u)
&=\iint J_0\left(2\pi |\freq|\cdot |z|\right) e^{-2\pi i z\cdot u}dz\\
\nonumber
&=\int_0^{\infty}\int_0^{2\pi} J_0\left(2\pi |\freq| r\right) e^{-2\pi i r |u| \cos(\phi-\phi_u)}2\pi |\freq|d|\freq| d\phi\\
&=\int_0^{\infty} J_0\left(2\pi |\freq|r\right)
2\pi |\freq| d|\freq| \int_0^{2\pi}e^{-2\pi i |\freq| |u| \cos(\phi-\phi_u)}
d\phi.
\end{align}
We now note that
\[J_0(x)=\frac{1}{2\pi}\int_{-\pi}^{\pi}e^{-ix\sin t}dt .\]
Thus
\begin{align}
\nonumber
{\cal F}\left\{J_0\left(2\pi |\freq|\cdot |z|\right) \right\}
(u)=\int_0^{\infty} J_0\left(2\pi |\freq|r\right)
(2\pi)^2  J_0\left(2\pi r |u| \right)\;rdr
=(2\pi)^2\frac{\delta(|\freq|-|u|)}{|\freq|}.
\end{align}
We now use the convolution theorem to deduce that:
\begin{align}
\nonumber
\E\bar{I}_0^{(3)}(|\freq|)&=|B|^{-1}\int_B|B\times B_{-z}| \cdot  J_0\left(2\pi |\freq||z|\right)\,dz\\&=\int_{{\mathbb{R}}^d}
\frac{\delta(|\freq|-|u|)}{|\freq|}|B|^{-1}T(B,u)\;du.
\end{align}
Thus the low magnitude wavenumber bias is determined by this term. The reason this is a low wavenumber term is the form of $|T(B,u)|$: this is concentrated near $|u|=0$, and on top the convolution is aggregating over all wavenumbers with the same modulus. We have assumed rectangular sampling domain.

The second term in the expectation of~\eqref{expect-split} takes the form of 
\[\E\bar{I}_0^{(2)}(|\freq|)=
|B|^{-1}\int_{{\mathbb{R}}^d}  
|B\times B_{-z}|\cdot
J_0\left(2\pi |\freq||z|\right)\left\{\rho^{(2)}(z)-\lambda^2\right\}\,dz.\]
As multiplications in wavenumber are convolutions in space, we need to compute
\begin{align}
\nonumber
{\cal F}\left\{ |B|^{-1}\cdot |B\times B_{-z}|\cdot
J_0\left(2\pi |\freq||z|\right)\right\}(u)&=\int_{{\mathbb{R}}^d}|B|^{-1}\cdot |B\times B_{-z}|\cdot
J_0\left(2\pi |\freq||z|\right) e^{-2\pi i u^Tz }\;dz\\
\nonumber
&=\int_{{\mathbb{R}}^d}
\frac{\delta(|\freq|-|u'|)}{|\freq|}|B|^{-1}T(B,u'-u)\;du'
\\
{\cal F}\left\{\rho^{(2)}(z)-\lambda^2 \right\}(u)
&=\lambda^2\widetilde{f}(u).
\end{align}
With these pieces we have that
\begin{equation}
\label{EI2b}
\E\bar{I}_0^{(2)}(|\freq|)=\int_{{\mathbb{R}}^d} \int_{{\mathbb{R}}^d} \lambda^2\widetilde{f}(u') 
\frac{\delta(|\freq|-|u''|)}{|\freq|}|B|^{-1}T(B,u''-u')\;du''
du'.
\end{equation}
We now see that further blurring is present in~\eqref{EI2b} from averaging out the direction. In the standard non-isotropic case this was just a convolution of $\widetilde{f}(\freq)$ with
$|B|^{-1}T(B,\freq)$. To get a feeling for its behaviour we note that as $|B|^{-1}T(B,\freq)\rightarrow \delta(\freq)/(2\pi)$. 
In this case we get for $d=2$
\begin{align}
\nonumber
\E\bar{I}_0^{(2)}(|\freq|)&\rightarrow\int_{{\mathbb{R}}^2} \int_{{\mathbb{R}}^2} \lambda^2\widetilde{f}(u') 
\frac{\delta(|\freq|-|u''|)}{2\pi |\freq|}\delta(u''-u')\;du''
du'.\end{align}
So this expression is what follows; an orientationally averaged spectral density. Now assume additionally that the spectrum is isotropic, namely  $\widetilde{f}(u')=\widetilde{f}_0(|u'|)$, and then we get as $|B|\rightarrow\infty$,
\begin{align}
\E\bar{I}_0^{(2)}(|\freq|)\nonumber
&\rightarrow\int_{{\mathbb{R}}^2}  \lambda^2\widetilde{f}_0(|u''|)
\frac{\delta(|\freq|-|u''|)}{2\pi |\freq|}\;du''\\
\nonumber
&= \lambda^2 \int_{{\mathbb{R}}^+} \widetilde{f}_0(u'') 
\frac{\delta(|\freq|-|u''|)}{|\freq|}|u''|\;d|u''|\\
&= \widetilde{f}_0(|\freq|).
\end{align}
This shows that asymptotically we would recover the isotropic spectral density from this component. 
Finally we can write
\[\E\bar{I}_0(|\freq|)=\lambda+\widetilde{f}_0(|\freq|)+o(1)+\int_{{\mathbb{R}}^d}
\frac{\delta(|\freq|-|u|)}{|\freq|}|B|^{-1}T(B,u)\;du.\]
\end{proof}

%\tuomas{[got up to here]}

\section{Additional Figures}

\begin{figure}[htp]
	\centering
	\includegraphics[width=.8\linewidth]{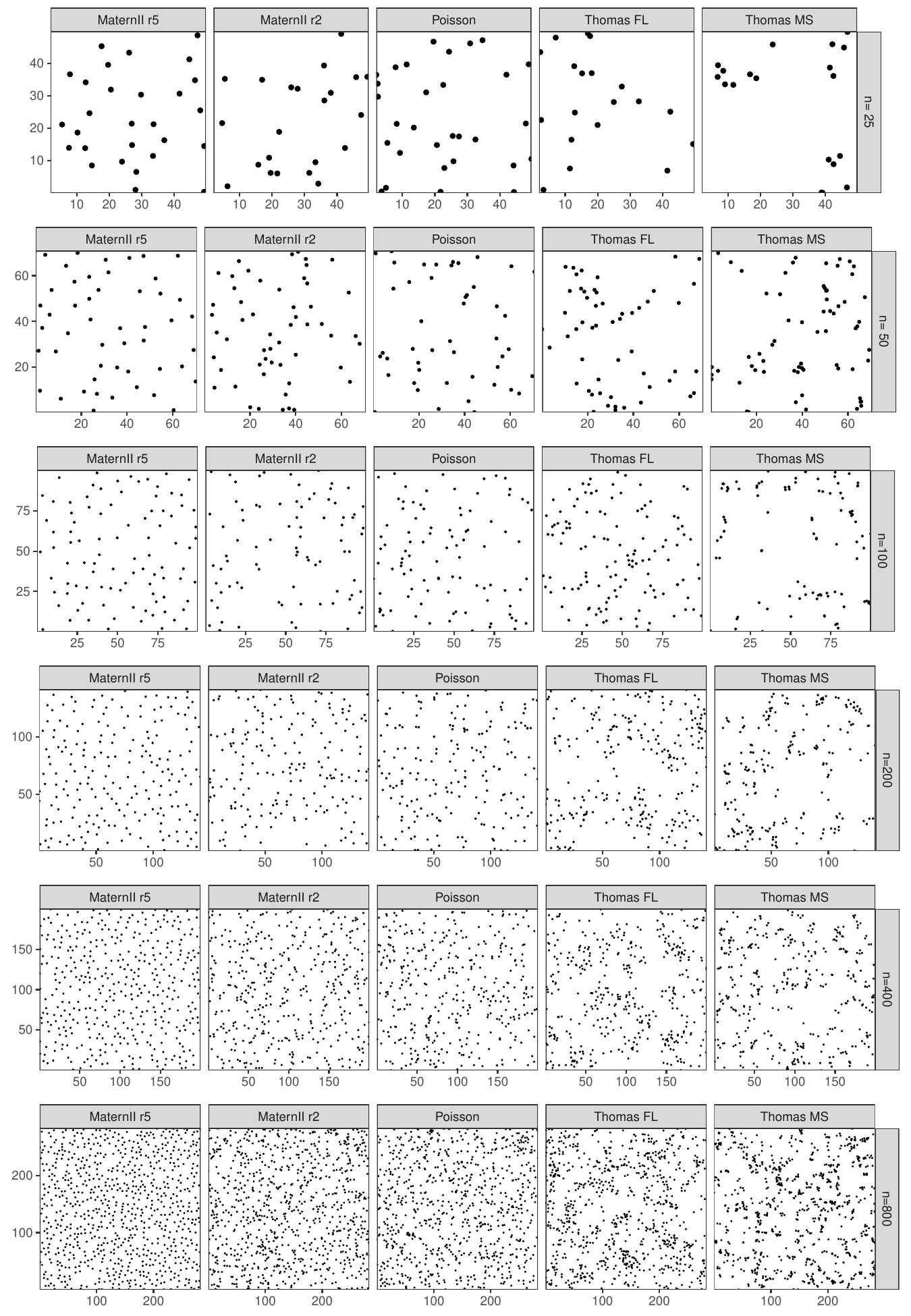}
	\caption{Examples of the simulated patterns used in the simulation trials for studying the quality of the estimators. Details of each model are given in Section \ref{sec:simulations}. Note that the exact number of points in each realisation varies slightly.}
	\label{figS:patterns}
\end{figure}

\begin{figure}[htp]
	\centering
	\includegraphics[width=0.8\linewidth]{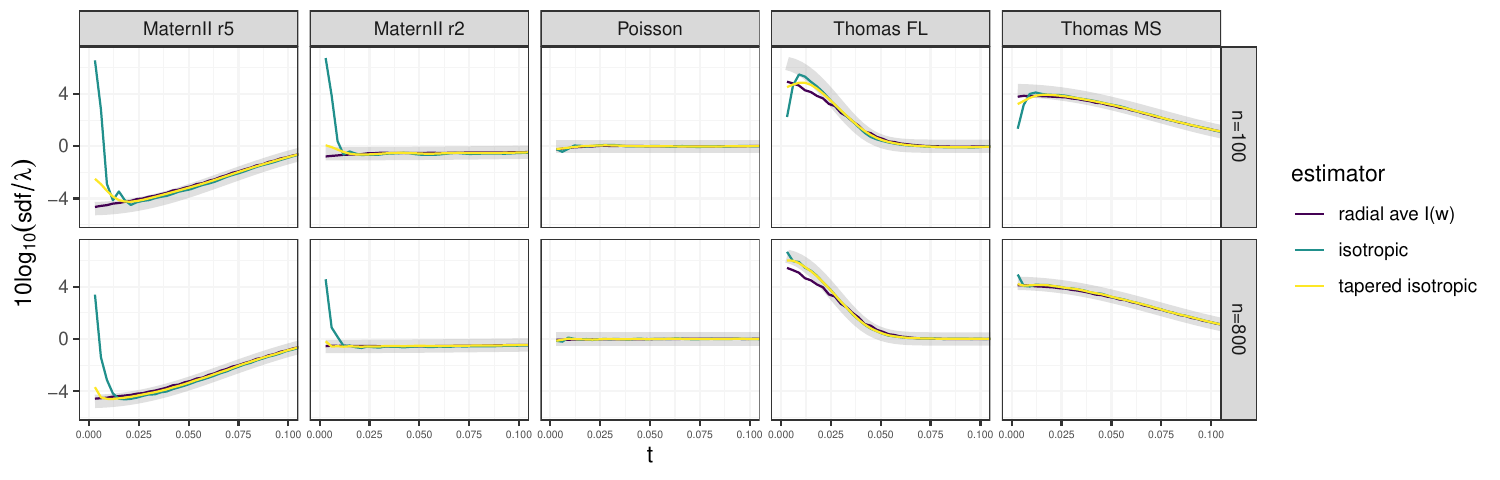}
	\caption{Isotropic $sdf/\lambda$, mean estimates versus true curve (thick gray line).}
	\label{figS:rotave-v-isotropic}
\end{figure}

\end{document}